\documentclass[12pt]{article}

\usepackage{latexsym,amsmath,amssymb,amsfonts,amsthm,bbm,mathrsfs,breakcites,dsfont,xcolor}
\usepackage{stmaryrd,epsf}
\usepackage{natbib}
\usepackage{url}
\usepackage{dsfont}
\usepackage{enumerate}
\usepackage{graphicx}
\usepackage{graphics}
\usepackage{psfrag}
\usepackage{caption}
\usepackage{subcaption}
\usepackage{float}
\usepackage{tabularx}
\usepackage{cellspace}
\usepackage{multirow}
\usepackage[Export]{adjustbox}
\usepackage[ruled,vlined,linesnumbered]{algorithm2e}

\usepackage[colorlinks,linkcolor=red,citecolor=blue,urlcolor=blue,breaklinks]{hyperref}

\newtheorem{theorem}{Theorem}
\newtheorem{lemma}[theorem]{Lemma}

\DeclareMathOperator*{\argmin}{argmin}

\DeclareMathOperator*{\sargmin}{sargmin}

\DeclareMathOperator{\rank}{rank}

\usepackage{xcolor}
\usepackage[draft,inline,nomargin,index]{fixme}
\fxsetup{theme=color,mode=multiuser}
\FXRegisterAuthor{tc}{atc}{\color{red} TC}
\FXRegisterAuthor{jj}{ajj}{\color{red} JJ}

\title{Random-projection ensemble dimension reduction}

\author{
 Wenxing Zhou and Timothy I. Cannings\\
 {\small \emph{School of Mathematics and Maxwell Institute for Mathematical Sciences}} \\ \emph{ University of Edinburgh} \\
}

\date{}

\begin{document}
\maketitle
\begin{abstract}
We introduce a new framework for dimension reduction in the context of high-dimensional regression. Our proposal is to aggregate an ensemble of random projections, which have been carefully chosen based on the empirical regression performance after being applied to the covariates. More precisely, we consider disjoint groups of independent random projections, apply a base regression method after each projection, and retain the projection in each group based on the empirical performance.   We aggregate the selected projections by taking the singular value decomposition of their empirical average and then output the leading order singular vectors. A particularly appealing aspect of our approach is that the singular values provide a measure of the relative importance of the corresponding projection directions, which can be used to select the final projection dimension. We investigate in detail (and provide default recommendations for) various aspects of our general framework, including the projection distribution and the base regression method, as well as the number of random projections used.  Additionally, we investigate the possibility of further reducing the dimension by applying our algorithm twice in cases where projection dimension recommended in the initial application is too large.  Our theoretical results show that the error of our algorithm stabilises as the number of groups of projections increases.  We demonstrate the excellent empirical performance of our proposal in a large numerical study using simulated and real data.  
\end{abstract}

\noindent \textbf{Keywords:} High-dimensional, mean central subspace, random projection, singular value decomposition, sufficient dimension reduction. 

\def\spacingset#1{\renewcommand{\baselinestretch}%
{#1}\small\normalsize} \spacingset{1}

\section{Introduction}
\label{sec:intro}
In regression, we seek to understand the relationship between a $p$-dimensional vector of predictors $X$ and a real-valued response variable $Y$, based on a sample of $n$ observations from the target population.  In many modern datasets, the number of predictors $p$ is often large, and can even exceed the number of observations $n$ in the dataset.  This `large $p$, small $n$' setting frequently arises across various applications, for example, in biomedicine, -omics data often consists of measurements on many thousands of genes \citep{horvath2013dna,bell2019dna}.  In these high-dimensional problems, many conventional regression methods suffer from the \emph{curse of dimensionality}  and may even become intractable. For instance, ordinary least squares is not applicable when $p > n$ because the solution involves the inverse of a singular sample covariance matrix. Moreover nonparametric methods, which are often based on smoothness assumptions, hinge on having a sufficient quantity of data around the points of interest for prediction. As a result, the number of observations $n$ required for such methods scales exponentially as the number of predictors $p$ increases \citep[see, for example,][Chapter~1.2]{wainwright2019high}. 

There are two fundamental strategies that attempt to address this curse of dimensionality. The first approach balances model fit and model complexity via regularisation. In particular, in linear regression problems, the extremely popular and widely studied \emph{ridge} \citep{Hoerl1970ridge}, \emph{lasso} \citep{tibshirani1996regression} and \emph{elastic net} \citep{zou2005regularization} algorithms provide elegant solutions. The latter two approaches operate on the assumption of \emph{sparsity} that only a few variables have explanatory effects on the response, enabling these algorithms to identify a suitable subset of the predictors while simultaneously performing prediction. This is particularly appealing due to the simple interpretation that can be derived from the output. However, if the relationship between the predictors and the response is nonlinear, or if many of the variables are important, then these approaches may fail. The second strategy is to directly target a low-dimensional representation of the predictors, while still retaining the essential information on the relationship between $X$ and $Y$.  In particular, one general approach, called \emph{sufficient dimension reduction} (SDR) \citep{cook1991sliced, Cook2007dimension, adragni2009sufficient}, aims to find a function $f:\mathbb{R}^p \rightarrow \mathbb{R}^{d}$, for some $d < p$, such that $X$ and $Y$ are independent given $f(X)$.  A closely related problem focuses on the so-called \emph{central mean subspace} \citep{cook2002dimension}, where the conditional mean of $Y$ given $X$ can be written in terms of $f(X)$. By restricting $f$ to be linear, these approaches are able to offer interpretable results, since the components in the low-dimensional representation may be more easily understood than those in the original $p$-dimensional vector.

In this paper, we propose a new method called \emph{random projection ensemble dimension reduction} (see Algorithm~\ref{alg:RPEDR estimator}). The main idea is to apply many low-dimensional random projections to the data, fit a base regression model after each projection and choose \emph{good} projections according their empirical performance. We then aggregate the chosen projections using  singular value decomposition.  Figure~\ref{fig:intro} below presents an example of our random projection-based method in a toy regression problem. We see that a randomly chosen 2-dimensional  projection retains very little of the relationship between $X$ and $Y$, but a projection carefully selected from a group of $L = 200$ random ones yields more promising results. Ultimately, by repeating this process and aggregating the chosen projections, we are able to retain a large amount of the structure observed after applying the oracle projection.  

\begin{figure*}[!ht]
    \centering
        \includegraphics[width=0.24\textwidth]{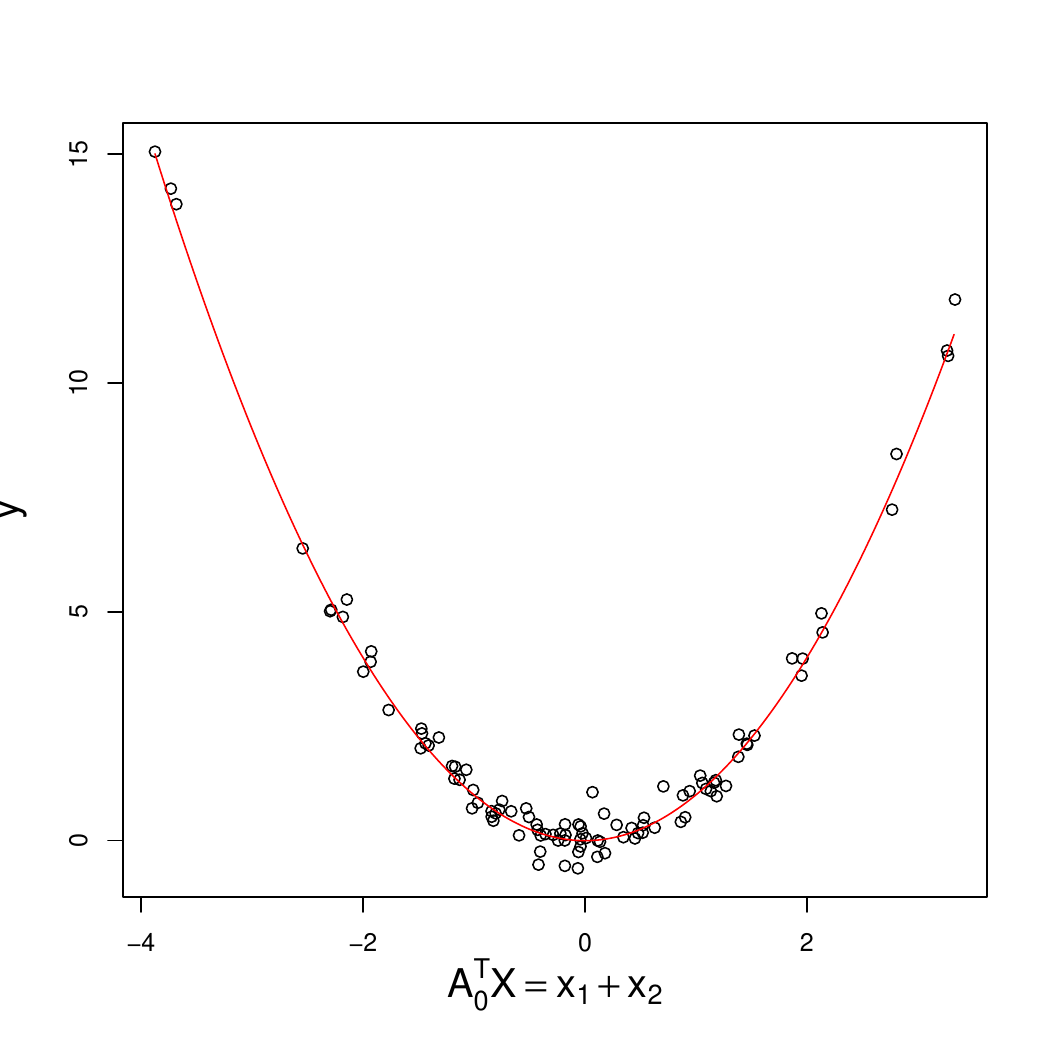}
        \includegraphics[width=0.24\textwidth]{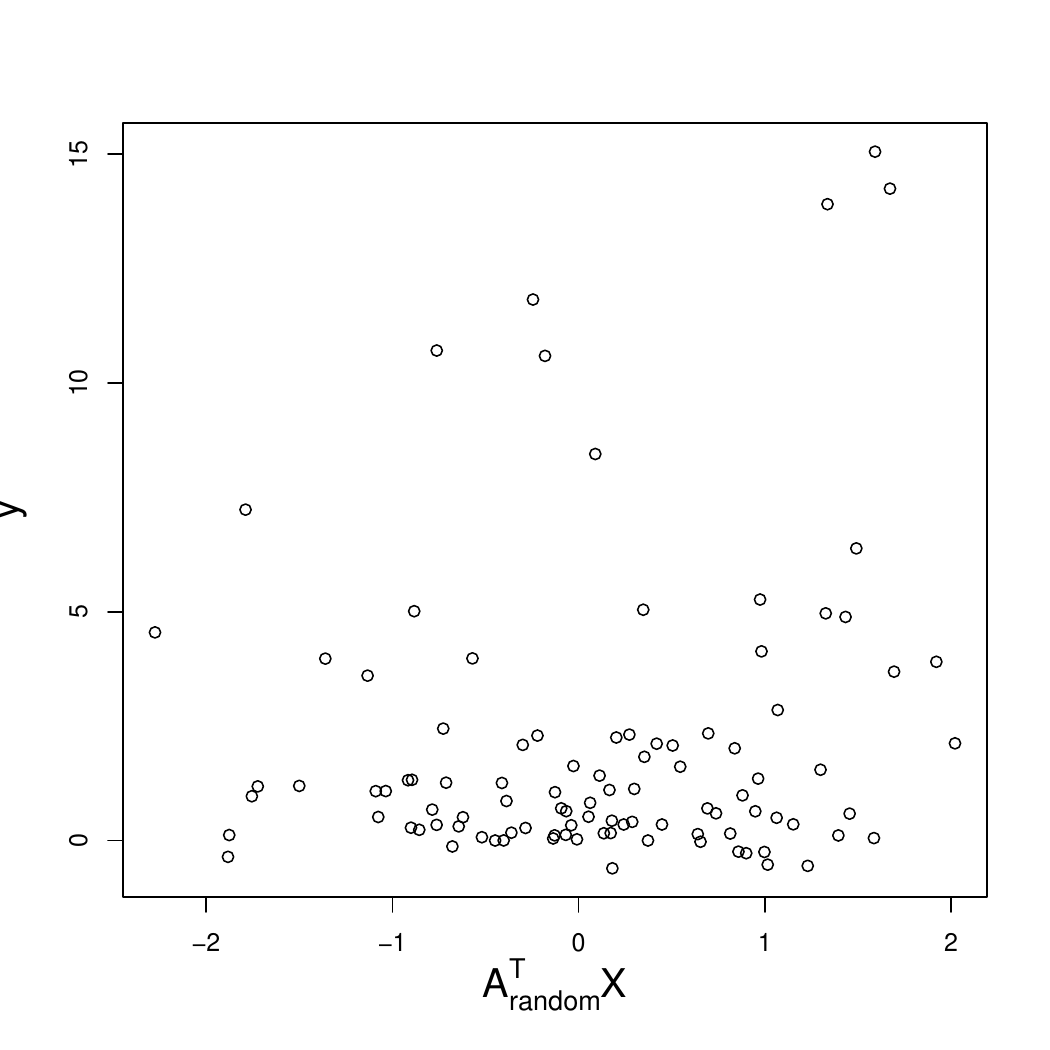}
        \includegraphics[width=0.24\textwidth]{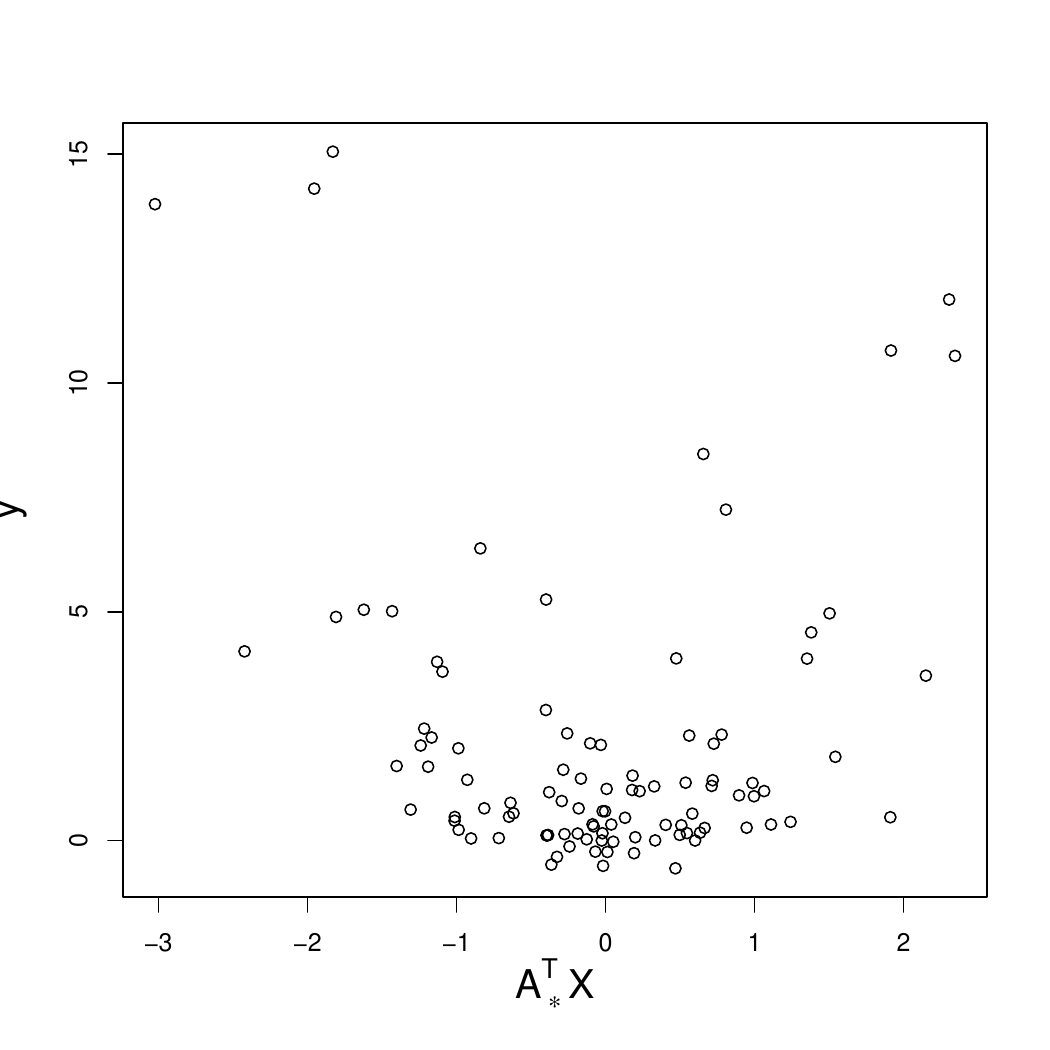}
        \includegraphics[width=0.24\textwidth]{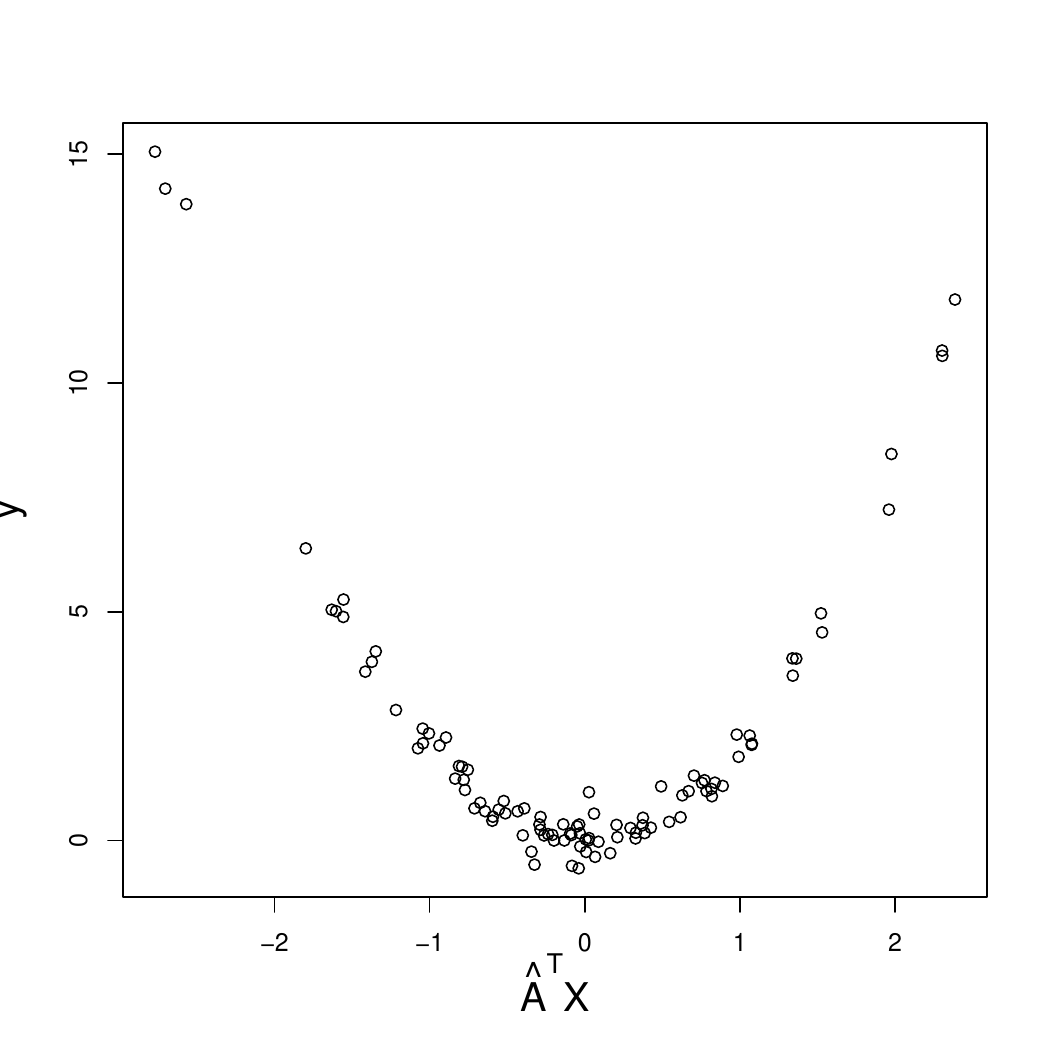}
    \caption{\small Dimension reduction example with $n = 100$ observations in $p=20$ dimensions. The true model is $Y = A_0^T X + 0.3\epsilon$, with $A_0 = (1/\sqrt{2}, 1/\sqrt{2}, 0, \ldots, 0)^T \in \mathbb{R}^{20}$, $X \sim N_{20}(0,I)$ and $\epsilon \sim N(0,1)$.  Far left: $Y$ plotted against the oracle projection $x_1 + x_2$, the red line shows the true regression function. Middle left: $Y$ plotted against a random projection. Middle right: $Y$ plotted against the best projection among a group of 200 random projections. Far right: $Y$ plotted against the projection estimated by our proposed method (with $L = 200$ groups).}  
    \label{fig:intro}
\end{figure*}

Our proposed framework is highly flexible and allows for user-specified choices of the random projection distribution and base regression method depending on the problem at hand.  We investigate in detail the choice of several aspects of our algorithm. This includes, in Section~\ref{sec:RPchoice}, the choice of the random projection distribution, where we find that simply using Gaussian projections is not always the best approach, particularly when the rows of the true projection matrix are sparse. In such cases projections with rescaled Cauchy entries perform better.  Based on our results, we recommend using an equal mixture of rescaled Gaussian and Cauchy projections, which are competitive in both the sparse and dense regimes.   We also compare the performance of our algorithm with different base regression methods. As demonstrated in Section~\ref{sec:basechoice}, simple parametric (linear and quadratic) models perform very well when they are correctly specified in the low-dimensional space, but potentially suffer prohibitively if they are misspecified. In this case, we recommend using more flexible nonparametric models, with multivariate adaptive regression splines (MARS) \citep{Friedman1991MARS} emerging as a particularly effective choice.  Our theoretical results (Theorem~\ref{thm:choiceofL}) show that the error of our approach converges at a rate no slower than $L^{-1/2}$, where $L$ is the number of groups of projections considered.  This rate is observed in our empirical investigations, which also consider the choices of the number of projections in each group and the dimension of the random projections; see Section~\ref{sec:choieofLM}.  

One of the key features of our algorithm is that the output from the singular value decomposition step, namely the diagonal matrix of singular values, provides a measure of relative importance of the corresponding singular vectors.  These values can therefore be used to determine a suitable projection dimension if this is not known in advance.  We propose in Algorithm~\ref{alg:DimensionEstimator} in Section~\ref{sec:choiceofd0} to choose the projection dimension by comparing the singular values with those obtained when the group size is one (i.e.~when no selection of \emph{good} random projections is performed).  

In Section~\ref{sec:doubleRPEDR}, we demonstrate that the performance of our algorithm can perhaps be further improved by applying it a second time on the projected data from the initial application.  In some cases, the combination of Algorithms~\ref{alg:RPEDR estimator} and~\ref{alg:DimensionEstimator} performs well in terms of the false negative rate (i.e.~they output a projection that contains a majority of the true signals, but may also includes additional `noise' directions that are not needed).  One way of viewing this is that the initial application acts more like a screening step, removing some of the noise directions and preserving the signal directions.   A second application of our algorithm (see Algorithm~\ref{alg:doubleDR}) applied to the projected data, can then refine the projection by combining signals and removing irrelevant directions. 

We demonstrate the performance of our algorithms through a large numerical study using simulated data in Section~\ref{sec:numerical}, and real data in Section~\ref{sec:realdata}.  Our proposal is very competitive in terms of empirical performance across a wide range of different settings. Code to implement our algorithms in \texttt{R} is available via GitHub\footnote{\url{https://github.com/Wenxing99/RPEDR}}.  

Work on sufficient dimension reduction dates back to the idea of \emph{inverse regression}, where in contrast to regressing $Y$ on $X$, we target the expectation of $X$ given $Y$. The earliest work in this area is \emph{sliced inverse regression} (SIR) \citep{SIR}, which uses `slices' of the response $Y$ to estimate $\mathrm{Cov}(X \mid Y)$, then the eigenvectors of this estimate are used as the basis of the dimension reduction space. This idea led to a long line of work on inverse regression, including, among many others, \emph{sliced average variance estimation} (SAVE) \citep{cook1991sliced}, \emph{principal Hessian directions} (PHD) \citep{li1992principal}, \emph{sparse SIR} \citep{li2006sparseSIR} and \emph{directional regression} (DR) \citep{li2007directional}. In another line of work, the nonparametric \emph{minimum average variance estimation} (MAVE) \citep{MAVE} algorithm estimates the dimension reduction space by constructing a minimization problem that quantifies how the dimension reduction model fits the observations.  \citet{ma2012semiparametric} propose a semiparametric based approach,  which take advantage of influence functions \citep{hampel1968contributions} to construct a new kind of estimator. An excellent overview of these methods, as well as many other early works, is provided in the review article by \citet{ma2013review}.  More recent proposals include, \emph{gradient based kernel dimension reduction} (gKDR) \citep{fukumizu2014gradient}, \emph{deep SDR} \citep{banijamali2018deep, kapla2022fusing}, \emph{conditional variance estimator} (CVE) \citep{fertl2022conditional} and \emph{dimension reduction and MARS} (drMARS) \citep{drMARS}; see also the recent review article by \citet{ghojogh2021sufficient}.

Random projections offer a natural approach to dimension reduction, and the celebrated Johnson--Lindenstrauss Lemma \citep{lindenstrauss1984extensions,dasgupta2003elementary} shows that applying a random projection to the data may approximately preserve the pairwise distances between the observations. However, as illustrated in Figure~\ref{fig:intro}, naively applying a single random projection typically fails.  Our proposal is motivated by the work of \citet{RPEnsembleClassification2017}, who propose a random projection ensemble method for binary classification based on aggregating the results of applying a base classifier to many carefully selected random projections of the data.  Related ideas have been used more recently in the context of sparse principal component analysis \citep{gataric2020sparse}, sparse sliced inverse regression \citep{zhang2023sparse} and semi-supervised learning \citep{wang2024sharp}.  Other works on (unsupervised) dimension reduction using random projection techniques include \citet{bingham2001random} and \citet{reeve2024heterogeneous}. Random projection based techniques have been used in a wide range of other statistical applications, including precision matrix estimation \citep{marzetta2011}, two-sample testing \citep{lopes2011}, clustering \citep{fern2003RPclustering,anderlucci2022high}, high-dimensional regression \citep{thanei2017RPregression,slawski2018principal,dobriban2019asymptotics,ahfock2021statistical} and differentially private learning \citep{huang2024efficient}.  See also the survey papers by \citet{xie2017RPRegressionSurvey} and \citet{cannings2021random}.

The remainder of our paper is as follows: in Section~\ref{sec:settingandmethods} we formally introduce our statistical setting and present our main algorithm, Section~\ref{sec:practicalconsiderations} focuses on various practical considerations and provides recommendations on how to choose the inputs to our algorithm.  In Section~\ref{sec:doubleRPEDR} we investigate the potential improvement gained by a second application of our main algorithm. The results of our numerical experiments with simulated and real data are presented in Sections~\ref{sec:numerical} and~\ref{sec:realdata}, respectively.  We conclude our paper with a discussion in Section~\ref{sec:conclusions}.  The Appendix contains the proofs of our theoretical results (Section~\ref{sec:proofs}), as well as the full results of our large simulation study (Section~\ref{sec:fullsims}).

\section{Statistical setting and methodology}
\label{sec:settingandmethods}
Let $P$ denote a distribution on $\mathbb{R}^p \times \mathbb{R}$ and  suppose that the covariate-response pair $(X, Y) \sim P$. Let $\eta(x) := \mathbb{E}(Y \mid X =x)$ be the regression function. We seek to find a dimension $d \in \{1, \ldots, p-1\}$ and a corresponding projection matrix $A \in \mathcal{A}_{p \times d} := \{A \in \mathbb{R}^{p \times d}: A^T A = I_{d \times d}\}$ for which there exists a function $g: \mathbb{R}^d \rightarrow \mathbb{R}$ satisfying 
\begin{equation}
\label{eq:SDR}
\eta(x) = g(A^T x). 
\end{equation}
Here $A^T$ maps $X$ onto a $d$-dimensional subspace, and when $d < p$ the lower-dimensional representation $A^T X$ contains all the information available in $X$ about the conditional mean of $Y$ given $X$. 

Equation~\eqref{eq:SDR} is of course satisfied by taking $d = p$, $A = I_{p\times p}$ (the $p$-dimensional identity matrix) and $g = \eta$. The main interest, however, is in finding a solution to~\eqref{eq:SDR} for a minimal choice of $d$, which we will denote by $d_0$.  The corresponding projection matrix and link function will be denoted by $A_0$ and $g_0$, respectively. While $d_0$ is unique, the associated $A_0$ is not.  Indeed, if $A_0 \in \mathcal{A}_{p \times d_0}$ and $g_0$ satisfy~\eqref{eq:SDR}, then so do $A_1 = A_0 B$ and $g_1(\cdot) = g_0(B \cdot)$, for any $B \in \mathcal{A}_{d_0 \times d_0}$.

Since $A_0$ is not uniquely identifiable, our focus is on  the space spanned by the columns of $A_0$. For a projection matrix $A \in \mathcal{A}_{p\times d}$, we write $\mathcal{S}(A) = \mathrm{span}(A)$ for the $d$-dimensional subspace spanned by the columns of $A$. If $A$ satisfies \eqref{eq:SDR}, then $\mathcal{S}(A)$ is called a \emph{mean dimension reduction subspace} \citep{cook2002dimension}.
The space spanned by the columns of the $d_0$-dimensional projection $A_0$, $\mathcal{S}(A_0)$, is called the \emph{central mean dimension reduction subspace}, or simply \emph{central mean subspace (CMS)}. 

For the purposes of this paper, we assume the existence and uniqueness of the CMS. This assumptions allows us to focus on our main goal of estimating the space $\mathcal{S}(A_0)$ based on a dataset of $n$ independent and identically distributed pairs $(X_1, Y_1), \ldots, (X_n, Y_n) \sim P$.  Our estimate in this problem will be the space spanned by the columns of a projection $\hat{A}_0 \in \mathcal{A}_{p\times \hat{d}_0}$ for some $\hat{d}_0 \in [p]$, say, where $\hat{d}_0$ and $\hat{A}_0$ are (possibly randomised) functions of the data $\mathcal{D} = \{(X_1, Y_1), \ldots, (X_n, Y_n)\}$.  There are many ways of measuring the distance between the spaces $\mathcal{S}(\hat{A}_0)$ and $\mathcal{S}(A_0)$. We will primarily focus on the measure $d_{\mathrm{F}}(\mathcal{S}(\hat{A}_0), \mathcal{S}(A_0) \bigr)$, where
\begin{equation}
\label{eq:dist}
  d_{\mathrm{F}}^2(\mathcal{S}(\hat{A}_0), \mathcal{S}(A_0) \bigr) := \frac{1}{2} \bigl \| \hat{A}_0 \hat{A}_0^T - A_0 A_0^T \bigr \|_{F}^2. 
\end{equation}
The scaling in \eqref{eq:dist} ensures that, if $d_0 = \hat{d}_0$, then the measure $d_{\mathrm{F}}$ is equivalent to the sin-theta distance \citep{davis1970rotation, yu2015useful}. 

\subsection{The random projection ensemble dimension reduction  algorithm}

In this section, we formally introduce our new procedure for dimension reduction, which is presented in Algorithm~\ref{alg:RPEDR estimator} below.  First, we outline some of the notation used, starting with the notion of a base regression method.  Given a data set of $n_1$ covariate-response pairs $((z_1,y_1),...,(z_{n_1},y_{n_1})) \in (\mathbb{R}^d \times \mathbb{R})^{n_1}$, a $d$-dimensional data-dependent regression method is a measurable function $\hat{g} : \mathbb{R}^d \times (\mathbb{R}^d \times \mathbb{R})^{n_1} \rightarrow \mathbb{R}$. The function $\hat{g}$ uses the data $(z_1,y_1),...,(z_{n_1},y_{n_1})$ to estimate a regression function $g : \mathbb{R}^d \rightarrow \mathbb{R}$. 
We write $\mathcal{G}_{d, n_1}$ for the set of all such regression methods.  In Algorithm~\ref{alg:RPEDR estimator}, $z_1, \ldots, z_{n_1}$ are different (randomly chosen) projections of a subset of the $p$-dimensional covariate observations $x_1, \ldots, x_n$.   

Our algorithm also takes as inputs a projection dimension $d \in [p]$ alongside a corresponding projection distribution $Q$ on $\mathcal{Q}_{p\times d} : = \{A \in \mathbb{R}^{p \times d} : \mathrm{diag}(A^TA) = (1,\ldots, 1)^T\}$, a number of groups of projections $L \in \mathbb{N}$, a group size $M \in \mathbb{N}$ and a subsample size $n_1 \in [n]$, as well as a base regression method $\hat{g} \in \mathcal{G}_{d, n_1}$. A full investigation, alongside recommendations of how to choose these inputs in practice is given in Section~\ref{sec:practicalconsiderations}. 

\begin{algorithm}[!ht]
\caption{Random projection ensemble dimension reduction}
\label{alg:RPEDR estimator}

\textbf{Input}: Data $((x_1,y_1),...,(x_n,y_n)) \in (\mathbb{R}^p \times \mathbb{R})^n$, projection dimension $d \in [p]$, projection distribution $Q$ on $\mathcal{Q}_{p\times d}$, number of groups $L$, group size $M$, sample-split size $n_1 \in [n-1]$, base regression method $\hat{g} \in \mathcal{G}_{d, n_1}$. 

Let $\mathbf{P}_{1,1},...,\mathbf{P}_{L,M} \stackrel{\mathrm{i.i.d}}{\sim} Q$. 

\For{$\ell \in [L]$} {
    Let $\mathcal{N}_1$ be a random sample of size $n_1$ chosen without replacement from $\{1,\ldots, n\}$. 
    
    \For{$m \in [M]$} {
    for $i \in \mathcal{N}_1^c$, let $\hat{h}_{\ell,m}(x_i) = \hat{g}\bigl(\mathbf{P}_{\ell,m}^Tx_i ; (\mathbf{P}_{\ell,m}^T x_j,y_j)_{j\in \mathcal{N}_1} \bigr)$.  
        
     let $\hat{R}_{\ell,m} = \sum_{i \in \mathcal{N}_1^c} \bigl\{ y_i - \hat{h}_{\ell,m}(x_i) \bigr\}^2$.
    }

    Choose the projection that minimizes the mean squared error within the group, $\mathbf{P}_{\ell, *} := \sargmin_{m \in [M]} \hat{R}_{\ell, m}$
}

Set $\hat{\Pi} := \frac{1}{L} \sum_{\ell = 1}^{L} \mathbf{P}_{\ell, *} \mathbf{P}_{\ell, *}^T$

Calculate the singular value decomposition of $\hat{\Pi}= U D U^T$

\textbf{Output}: The matrix $U \in \mathbb{R}^{p \times p}$ and  $\mathrm{diag}(D) \in \mathbb{R}^p$. 
\end{algorithm}

Algorithm~\ref{alg:RPEDR estimator} starts by generating $L \cdot M$ random projection matrices $\mathbf{P}_{\ell, m} \in \mathbb{R}^{p \times d}$, where each is sampled from the specified distribution $Q$. These projections are partitioned into $L$ disjoint groups of size $M$. For each projection  $\mathbf{P}_{\ell,m}$, we apply it to the covariates and fit the base regression model $\hat{g}$ using a subsample of the projected data. The performance of each projection is measured by the mean squared error on the complement of the training subsample and within each group, we retain the projection with the best performance.  These selected projections, denoted $\mathbf{P}_{\ell,*}$, are then aggregated by taking the average of their outer products to give $\hat{\Pi}$. Finally, the algorithm outputs the $p\times p$ orthonormal matrix $U = (U_1, \ldots, U_p)$ alongside the vector $\mathrm{diag}(D)$, provided by the singular value decomposition of $\hat{\Pi}$. 

The idea is that the columns of $U$ contain the estimated dimension reduction directions, with the singular values $D$ providing information about the relative importance of each direction. If the desired projection dimension $\hat{d}_0$ is predetermined, then we retain the first $\hat{d}_0$ columns and set $\hat{A}_0 = (U_1, \ldots, U_{\hat{d}_0})$ as our estimate of $A_0$.  On the other hand, if $\hat{d}_0$ is not determined in advance, then the singular values provide a useful guide.  In some cases, there may be a clear candidate choice of $\hat{d}_0$, but when an appropriate choice is less obvious, we propose to choose $\hat{d}_0$ by comparing the output of Algorithm~\ref{alg:RPEDR estimator} with the corresponding output obtained when $M=1$, i.e.~compared with random matrices without selection (see Algorithm~\ref{alg:DimensionEstimator} in Section~\ref{sec:choiceofd0}).

\begin{figure*}[!ht]
    \centering
        \includegraphics[width=0.44\textwidth]{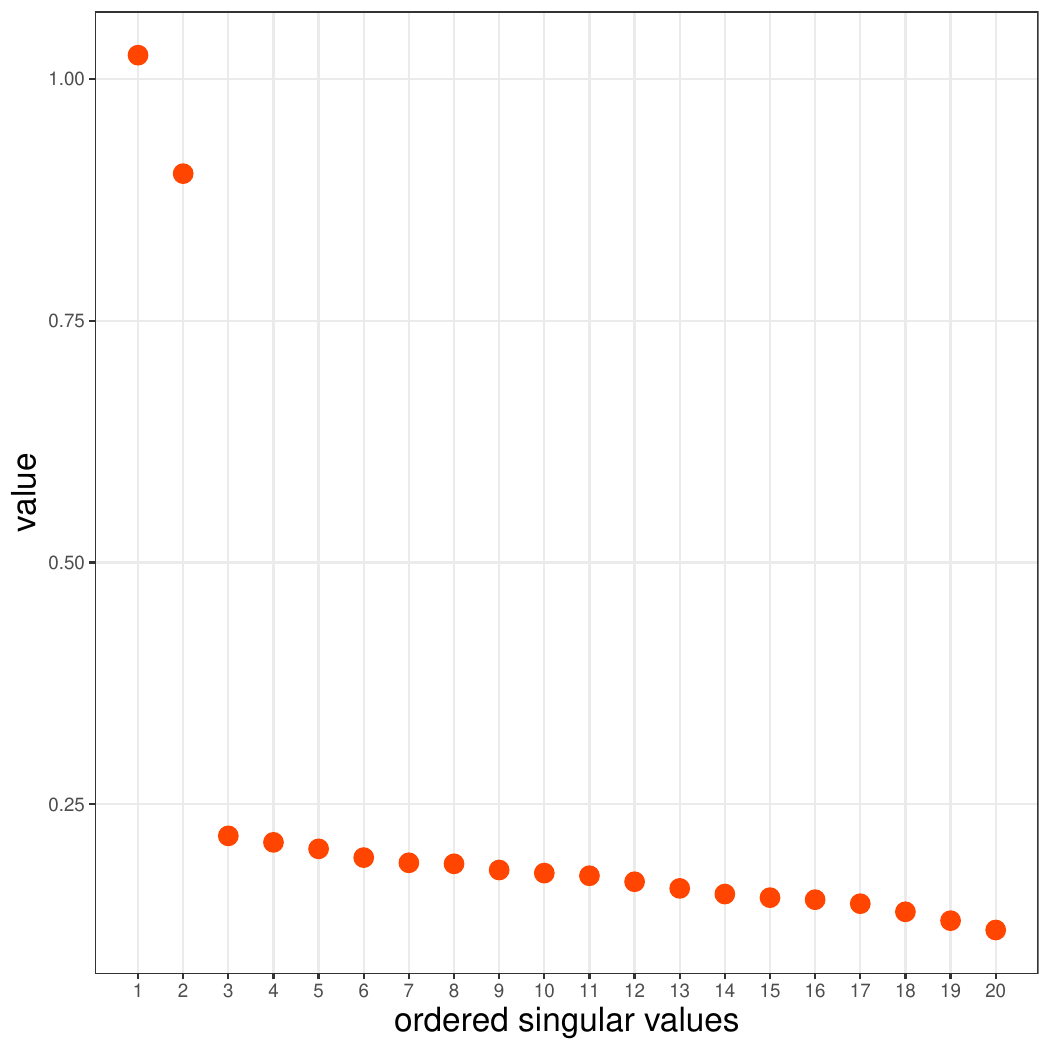}
          \includegraphics[width=0.44\textwidth]{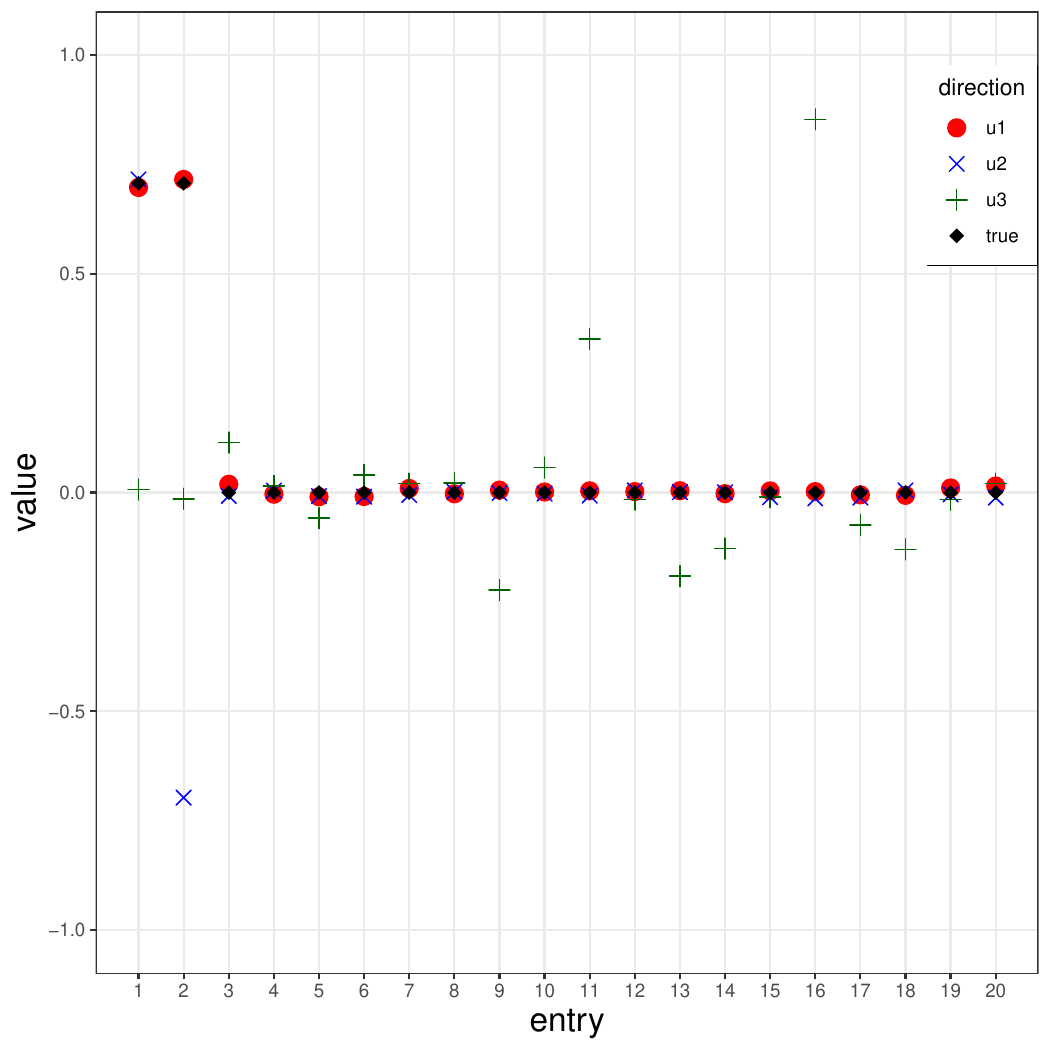}
     \caption {\small An example of some of the output from Algorithm~\ref{alg:RPEDR estimator} applied to the data used in Figure~\ref{fig:intro} with $p=20$ and $n = 100$. Here we take $d=5$, $L=200$, $M=200$, $n_1 = 66$, the random projections $Q$ are taken to have independent standard Cauchy random variables, and $\hat{g}$ the quadratic least-squares regressor. We present the entries of the diagonal matrix $D$ (left plot), alongside the entries of the first (red filled circles), second (blue $\times$) and third (green $+$) columns of $U$ (right plot). We also include the entries of the true dimension reduction direction (black filled diamonds).}  
    \label{fig:Alg1demo}
\end{figure*}

To demonstrate how our approach works, we applied Algorithm~\ref{alg:RPEDR estimator} to the data used in Figure~\ref{fig:intro}. In Figure~\ref{fig:Alg1demo} we observe that the first two entries of $D$ are relatively large, suggesting that the singular vectors $U_1$ and $U_2$ contain a large proportion of the signal.  In this case, the true sufficient dimension reduction direction $A_0$ is one-dimensional and  nearly all of the signal is found in the first singular vector, with $|U_1^TA_0| = 0.996$.  The second singular vector here contains a small amount of the signal with $|U_2^TA_0| = 0.048$.   The remaining entries of $D$ are relatively small, suggesting that $U_3,\ldots, U_p$ can be discarded (indeed $\max_{j=3, \ldots, p} |U_j^TA_0| < 0.01)$.

\section{Practical considerations}
\label{sec:practicalconsiderations}

\subsection{Random projection generating distribution}
\label{sec:RPchoice}
In this section, we investigate the effect of the random projection distribution in Algorithm~\ref{alg:RPEDR estimator}. For brevity, we focus on two distributions:
\begin{enumerate} 
\item \emph{Gaussian rows:} Let $Q_\mathrm{N}$ denote the distribution $\mathcal{Q}_{p\times 1}$ formed by first generating $p$ independent standard  Gaussian random variables and then normalising the row to have unit norm.  In other words, we say $\mathbf{Q} \sim Q_\mathrm{N}$ if $\mathbf{Q} \stackrel{d}{=} Z/\|Z\|$, where $Z \sim N_p(0, I)$. 
\item \emph{Cauchy rows:} We write $\mathbf{Q} \sim Q_\mathrm{C}$ if $\mathbf{Q} \stackrel{d}{=} W/\|W\|$, where $W = (W_1, \ldots, W_p)$ and $W_1, \ldots, W_p \stackrel{\mathrm{i.i.d.}}{\sim} \mathrm{Cauchy}(0,1)$. 
\end{enumerate} 

Projections with independent Gaussian entries are widely used in the literature and have desirable rotational invariance properties \citep[see, for example,][Proposition~3.3.2]{Vershynin_2018}. On the other hand, several works have explored the use of Cauchy random projections in various contexts, offering distinct advantages for certain applications. For example, \citet{Li2007Nonlinear} extended the Johnson-Lindenstrauss Lemma to the 
$l_1$-norm using Cauchy random projections, while \citet{Ramirez2012Reconstruction} used them for sparse signal reconstruction.

 We conduct a numerical experiment using three different versions of the following model:
 \begin{itemize}
  \item Model 1: Let $X \sim N_p(0, I)$ and $Y = 2(A_0^T X)^2 + \epsilon$, where $\epsilon \sim N(0,1/4)$ and  $A_0^T = \frac{1}{\sqrt{q}}(\mathbf{1}_{q}, \mathbf{0}_{p-q})$, for $q \in [p]$. Here $\mathbf{1}_{q}$ and $\mathbf{0}_{p-q}$ denote a $q$-dimensional vector of ones and a ($p-q$)-dimensional vector of zeros, respectively.
  \end{itemize}
 We set $p=20$ and vary $q$ to simulate different sparsity levels, with $q = 2$ (Model 1a), $q = 10$ (Model 1b) and $q = 20$ (Model 1c). We consider three different projection distributions, namely $Q_\mathrm{N}^{\otimes d}$ (Gaussian),  $Q_\mathrm{C}^{\otimes d}$ (Cauchy) and $\frac{1}{2} Q_\mathrm{N}^{\otimes d} + \frac{1}{2} Q_\mathrm{C}^{\otimes d}$ (mixture). The idea is to show the effect of the choice of random projection distribution for the different levels of sparsity.  The experiment is designed in such a way that the signal strength is the same in each setting. The other inputs to Algorithm~\ref{alg:RPEDR estimator} are the same in each case. Indeed, in all the experiments in this subsection, we set $n = 200$, and the inputs to Algorithm~\ref{alg:RPEDR estimator} are $d = 5$, $L = 200$, $M = 200$, $n_1 = \lceil 2n/3 \rceil$, the base method $\hat{g}$ is taken to be the global quadratic least squares estimator (see equation \eqref{eq:globalQLS} for a formal definition).  

\begin{figure}[!ht]
    \centering
        \includegraphics[width=\textwidth]{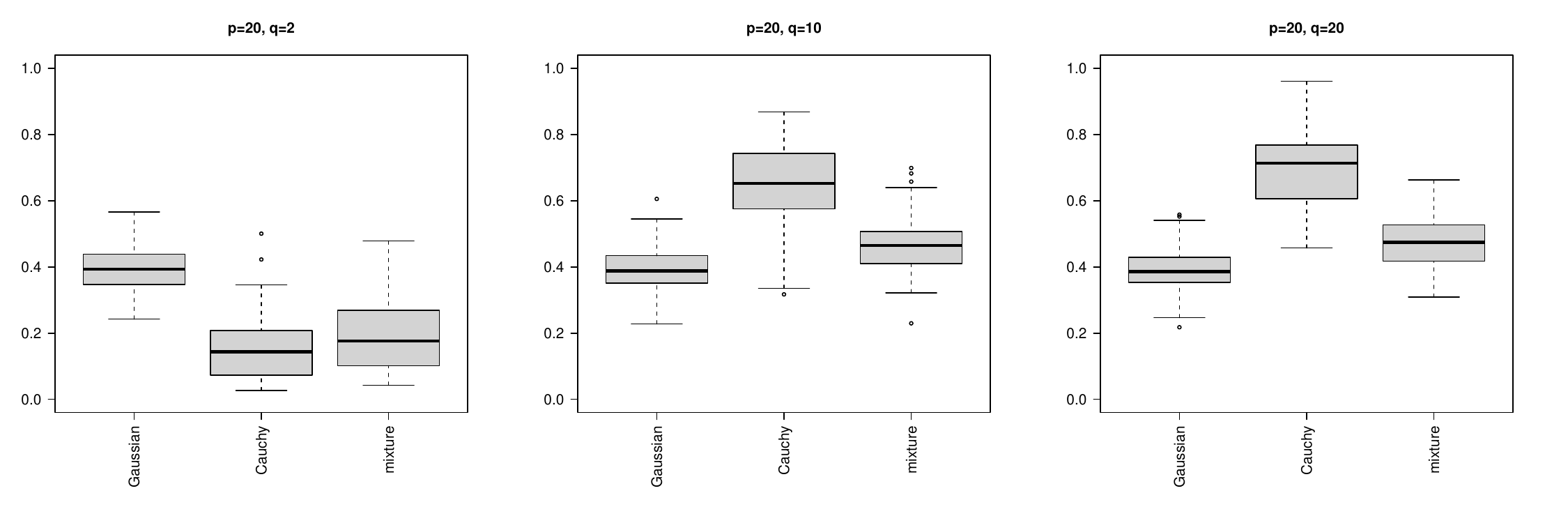}
    
    \caption{{\small Boxplots of the $\sin$-theta distance (see~\eqref{eq:dist}) between $U_1$ from the output of Algorithm~\ref{alg:RPEDR estimator} and the true direction $A_0 = (\mathbf{1}_q, \mathbf{0}_{p-q})$ over 100 repeats of the simulation. We present the results for three different instances of Model~1, namely when $q= 2$ (left), $q=10$ (middle) and $q= 20$ (right). We compare the results when the projection distribution $Q$ is $Q_{\mathrm{N}}^{\otimes d}$ (Gaussian), $Q_{\mathrm{C}}^{\otimes d}$ (Cauchy), and $\frac{1}{2} Q_{\mathrm{C}}^{\otimes d} + \frac{1}{2} Q_{\mathrm{N}}^{\otimes d}$ (mixture). \label{fig:RPchoice}}}  
\end{figure}

Figure~\ref{fig:RPchoice} shows that in sparse models Cauchy projections do indeed outperform Gaussian projections, but the heavy tailed Cauchy distribution is not well-suited to cases where there is no sparsity.  As expected, due to the rotational invariance in the problem, Gaussian random projections exhibit the same performance regardless of the sparsity level.  We see that using a fifty-fifty mixture of Gaussian and Cauchy projections in our algorithm leads to good performance regardless of the sparsity level, and if the practitioner has no knowledge of the sparsity level of $A_0$, then we recommend using this mixture as the default option and indeed this strategy will be employed in all of the numerical experiments presented later in the paper.

\subsection{Choice of base regression method}
\label{sec:basechoice}
The choice of base regression method 
$\hat{g} \in \mathcal{G}_{d,n_1}$ used in Algorithm~\ref{alg:RPEDR estimator} depends on two primary factors.  First, the method should be computationally efficient, as it will be applied $L\cdot M$ times to different randomly projected datasets. We delay our discussion of computational aspects of the problem to Section~\ref{sec:comp} and focus here on the second primary factor: the base regression method must capture at least some of the signal after a \emph{good} projection has been applied.  We will demonstrate that a global polynomial based method is often effective, even when the true signal is not exactly a polynomial. However, there are scenarios in which a more flexible nonparametric approach may be needed.   

We investigate four options in detail. In each case, we describe how $\hat{g} \in \mathcal{G}_{d,m}$ is defined based on a dataset $(z_1, y_1), \ldots, (z_m, y_m) \in (\mathbb{R}^d \times \mathbb{R})^m$ of size\footnote{this sample size requirement is mild and ensures that we can apply the global quadratic method below -- if it is not satisfied, then either $d$ should be reduced, or we may be limited to using the linear least squares method which only requires $m > d$.} $m > 1 + d(d+3)/2$ and a test point $z \in \mathbb{R}^d$.  
\begin{enumerate}
    \item Global linear least squares (LLS): $\hat{g}_{\mathrm{GL}}(z) = \hat{\alpha} + \hat{\beta}^T z$, where 
    \[
    (\hat{\alpha}, \hat{\beta}) := \argmin_{(\alpha, \beta) \in \mathbb{R}\times \mathbb{R}^d} \Bigl\{\sum_{i=1}^m (y_i - \alpha - \beta^Tz_i)^2\Bigr\}.
    \]
    \item Global quadratic least squares (QLS): $\hat{g}_{\mathrm{GQ}}(z) := \hat{a} + \hat{b}^T z + z^T\hat{C} z $, where 
    \begin{equation}
    \label{eq:globalQLS}
    (\hat{a}, \hat{b}, \hat{C}) := \argmin_{(a, b, C) \in \mathbb{R}\times \mathbb{R}^d \times \mathbb{S}_{d\times d}} \Bigl\{\sum_{i=1}^m (y_i - a - b^Tz_i - z_i^TCz_i)^2\Bigr\},
    \end{equation}
    and $\mathbb{S}_{d\times d}$ denotes the set of $d \times d$ symmetric matrices. 
     \item Nadaraya--Watson: Fix a kernel $K$ and a bandwidth $h >0$, then let 
     \[
     \hat{g}_{\mathrm{NW}}(z) = \frac{\sum_{i=1}^m y_i K(h^{-1} \|z_i-z\|) }{\sum_{i=1}^m K(h^{-1}\|z_i-z\|)}.
     \] 
     In our experiments below, we take $K(t) = \frac{e^{-t^2/2}}{\sqrt{2\pi}}$ and $h = 0.1$.
    \item Multivariate Adaptive Regression Splines (MARS):  Let $\hat{g}_{\mathrm{MARS}}$ be Algorithms~2 and~3 of \citet{Friedman1991MARS}, which fits piecewise linear models to the data and automatically selects interactions between variables. We use the implementation of this method from the \texttt{R}-package \texttt{earth} \citep{package2024earth}. The maximum degree of interaction is taken to be $3$.
\end{enumerate}
While many alternative base regression methods are available, we focus on these four methods for their balance of flexibility and computational feasibility.

To demonstrate how these different options perform in practice, we carry out a further set of experiments using three different models:
\begin{itemize}
   \item Model 1a: as used in the previous subsection with $A_0 = (\mathbf{1}_2, \mathbf{0}_{p-2})^T$.
    \item Model 2: Let $X = (X_1,\ldots, X_p)^T \sim \mathrm{Unif}([-1,1]^p)$ and $Y = 2 \sin(2\pi X_3 ) + \epsilon$ with $\epsilon \sim N(0,1/4)$ independent of $X$.
     \item Model 3\footnote{Model 3 was used in the simulation studies in \citet{SIR,MAVE,drMARS}.}: Let $X = (X_1, \ldots, X_p)^T \sim N_p(0,I)$ and 
     \[
     Y =  \frac{X_6}{1/2 + (X_7+3/2)^2} + \epsilon
     \] 
     with $\epsilon \sim N(0,1/4)$ independent of $X$.  Here $A_0 = (e_6,e_7)$ is $2$-dimensional, where $e_j$ denotes the $j$th canonical Euclidean basis vector in $\mathbb{R}^p$.
\end{itemize}
In each case, we set $p =20$ and $n=200$ and apply Algorithm~\ref{alg:RPEDR estimator} with  $d = 5$, $L = 200$, $M = 200$, $n_1 = \lceil 2n/3 \rceil$, and $Q = \frac{1}{2} Q_{\mathrm{N}}^{\otimes d} + \frac{1}{2} Q_{\mathrm{C}}^{\otimes d}$.

\begin{figure}[!ht]
    \centering
        \includegraphics[width=0.9\textwidth]{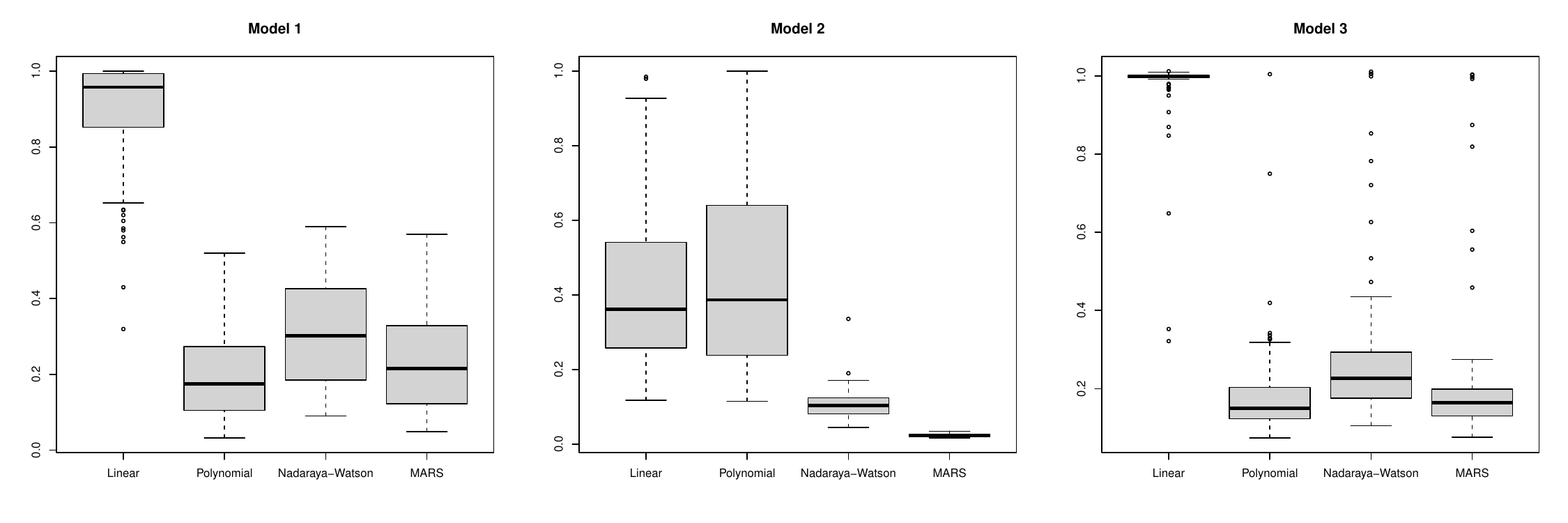}
    \caption{{\small Boxplots of the $\sin$-theta distance (see~\eqref{eq:dist}) between $U_{1:d_0}$ from the output of Algorithm~\ref{alg:RPEDR estimator} and the true direction $A_0$ over 100 repeats of the simulation. We present the results for Model~1a (left), Model 2 (middle) and Model 3 (right), for $p=20$ and $n = 200$. The base method is linear least squares, quadratic least squares, Nadaraya-Watson or the MARS algorithm. }\label{fig:basechoice}}  
\end{figure}

Figure~\ref{fig:basechoice} exhibits some interesting behaviour.  In Model 1a,  the linear regression base method is not effective because no linear signal exists in any $d$-dimensional projection of the data. For a similar reason,  the linear and quadratic least squares methods do not perform well in Model 2. However, in Model 1a, the quadratic base method is very effective, as this model is correctly specified after applying the oracle projection $A_0$. MARS also performs well here, although at a higher computational cost.  For Model 2, the flexibility of the nonparametric methods is required to find the signal, and MARS outperforms the Nadarya-Watson base method.  The results for Model 3 show that, even when the base model is incorrectly specified, good performance is still achieved. Here the true regression function after the oracle projection is not a polynomial. The linear base method is effective in finding the (approximately) linear signal in $X_6$, but does not find any of the signal in $X_7$. However, the signal is sufficiently well-approximated by a quadratic function that we are able to find a lot of the signal with a the quadratic base method.  Here the slightly more computationally expensive MARS algorithm performs similarly to the quadratic least squares approach.  

Based on results in this section, MARS is recommended as a suitable default base regression method. It performs well across all examples presented here, as well as in a broad range of settings in our large simulation study in Section~\ref{sec:numerical}. 

\subsection{Choice of \texorpdfstring{$L$}{L}, \texorpdfstring{$M$}{M} and \texorpdfstring{$d$}{d}}
\label{sec:choieofLM}
Our Algorithm~\ref{alg:RPEDR estimator} combines the results of assessing the empirical performance of the base regression method in a total of $L \cdot M \cdot d$ projection directions.  More precisely, we aggregate the results of $L$ ($d$-dimensional) projections, where each projection is selected as the best out of a (disjoint) block of size $M$.  In this section, we explore  how varying these parameters affects performance. As we will show, the performance typically improves as $L$ increases, in the sense that the sin-theta distance decreases at rate $L^{-1/2}$ (assuming the other quantities are kept fixed); see Theorem~\ref{thm:choiceofL} and our numerical results in Figure~\ref{fig:increasingL}. Overall, the results in this section lead to our recommendation of a default choice of $L = 200$.  The effect of the parameters $M$ and $d$ is slightly less straightforward, and their choices are influenced by the ambient dimension $p$. Indeed, ultimately we recommend taking $M = 10p$ and $d = \lceil p^{1/2} \rceil$. 

We first focus on the choice of $L$. Theorem~\ref{thm:choiceofL} below compares the sin-theta distance between the first $d_0$ columns of $U$ from the output of Algorithm~\ref{alg:RPEDR estimator} applied with $L$ projections, with the corresponding output when considering an infinite simulation version of our algorithm which takes ``$L = \infty$".  More precisely, given data $\{(x_1,y_1), \ldots, (x_n, y_n)\}$ (here considered to be fixed pairs in $\mathbb{R}^p \times \mathbb{R}$), $d\in [p]$, $Q \in \mathcal{Q}_{d\times p}$, $L \in \mathbb{N}$, $M \in \mathbb{N}$, $n_1 \in [n-1]$ and $\hat{g} \in \mathcal{G}_{d,n_1}$, let $\hat{A}_0^{L} := \hat{A}_0 = (U_1, \ldots, U_{d_0})$ denote the first $d_0$ columns of the output of Algorithm~\ref{alg:RPEDR estimator}.  Here we assume knowledge of the true $d_0$ and explicitly emphasise the dependence on $L$.  Further define $\Pi^{\infty} := \mathbb{E} (\mathbf{P}_{\ell, *}\mathbf{P}_{\ell, *}^T) \in \mathbb{S}_{p\times p}$ to be the expected value of $\hat{\Pi}$ in line 9 of Algorithm~\ref{alg:RPEDR estimator}. Here the expected value is taken only over the randomness in the projections and in the sample-split in line 4 of the algorithm, but not the data.  Let $\hat{A}_0^{\infty} := (U^{\infty}_1, \ldots, U^{\infty}_{d_0})$, where $U^{\infty}_1, \ldots, U^{\infty}_p$ are the singular vectors of $\Pi^{\infty}$.

\begin{theorem}    
\label{thm:choiceofL}
We have 
\[
   \mathbb{E} d_{\mathrm{F}}\bigl(\mathcal{S}(\hat{A}_0^{L}), \mathcal{S}({A}_0)\bigr)  \leq  2d_0^{1/2}\|\Pi^{\infty} - A_0A_0^T\|_{\mathrm{op}} +  \frac{2d_0^{1/2} (2\pi)^{1/2}p}{L^{1/2}}.
    \]
\end{theorem}
The proof of Theorem~\ref{thm:choiceofL} is given in Section~\ref{sec:proofs}. The second term in the bound given by Theorem~\ref{thm:choiceofL} suggests that the error of our algorithm decreases at a rate of $L^{-1/2}$ as $L$ increases. Our numerical results below suggest that this is typically the case in practice.  
The first term in the bound does not depend on $L$, and can be seem as the infinite simulation error of our method.  This term depends on the choice of $M$ and $d$ (as well as the base regression method and projection distribution) and would be small when, on average, the projections selected are close to $A_0$.  Since $A_0$ is unknown, determining a theoretically optimal choice of $M$ and $d$ seems intractable and we will instead provide recommendations based on the empirical experiments below -- see Figure~\ref{fig:choiceofMd}.

To further elucidate the effect of $L$ in practice, we conduct a numerical experiment using versions of Models 1a, 2 and 3 from the previous subsection with $p = 20$ and $p=100$; see Figure~\ref{fig:increasingL}.  We set $n = 200$, with the other inputs to Algorithm~1 being $d = 5$, $M = 10p$, $n_1 = \lceil 2n/3 \rceil$,  $Q = \frac{1}{2} Q_{\mathrm{C}}^{\otimes d} + \frac{1}{2} Q_{\mathrm{N}}^{\otimes d}$ for the projection distribution, and $\hat{g}$ being the MARS algorithm.  The number of projections $L$ varies from $10$ to $1000$.

\begin{figure}[!ht]
    \centering
   \includegraphics[width=0.3\textwidth]{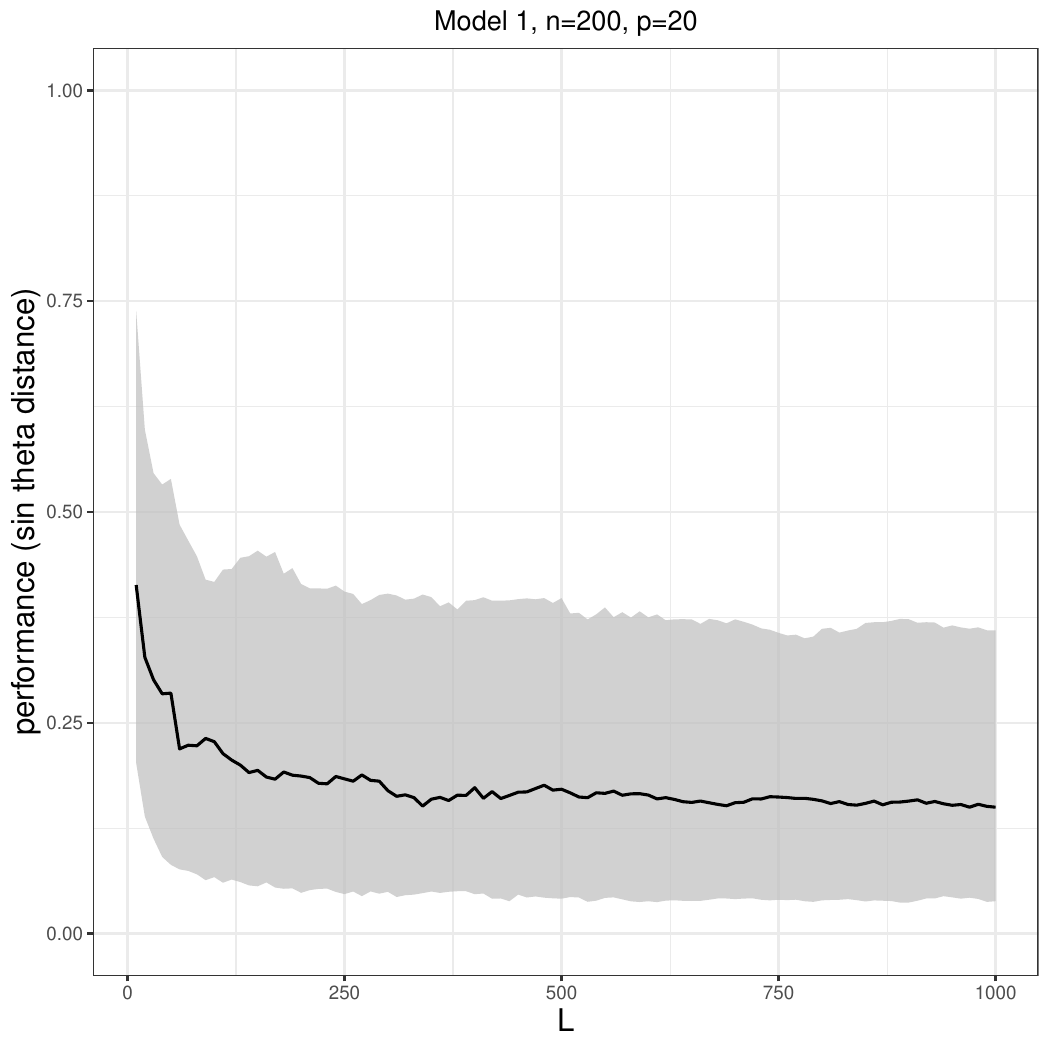}
   \includegraphics[width=0.3\textwidth]{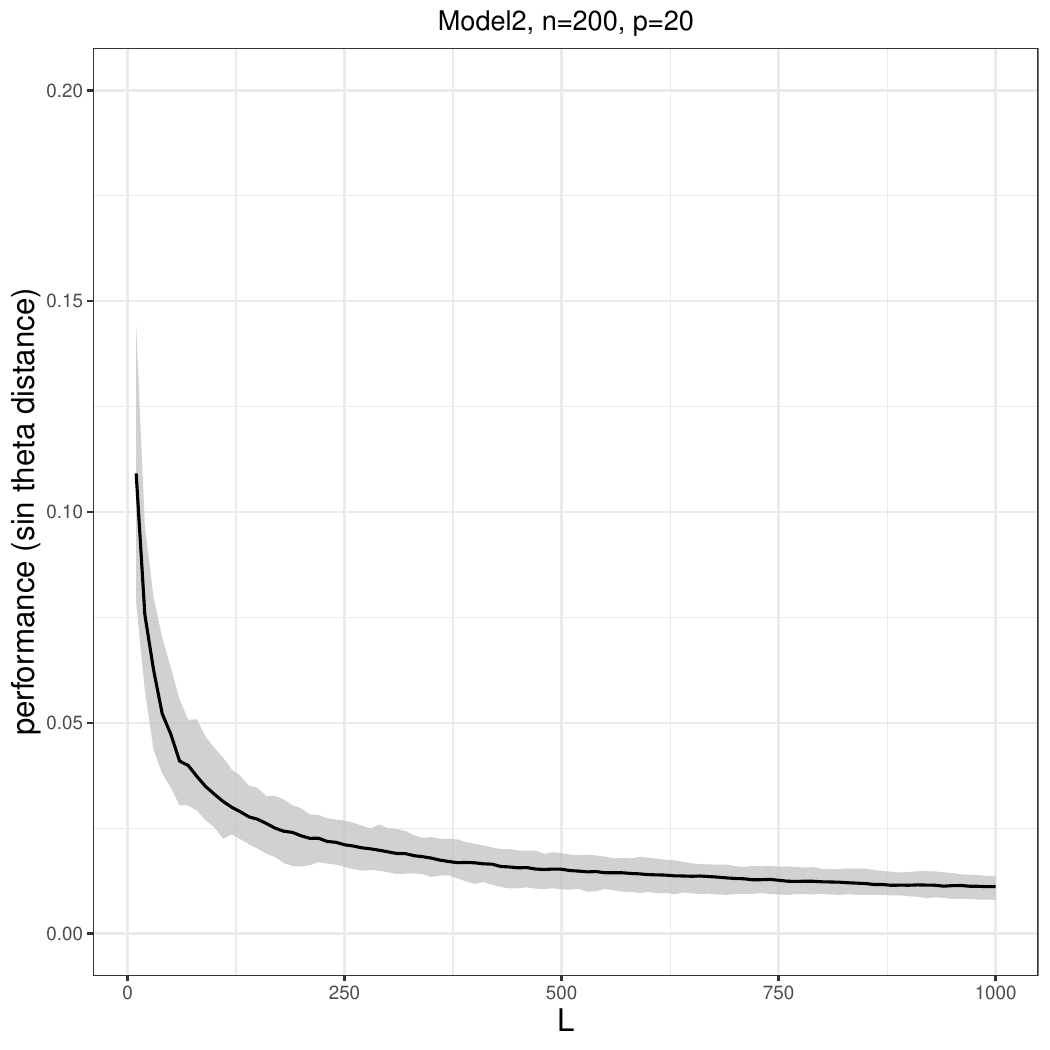}
   \includegraphics[width=0.3\textwidth]{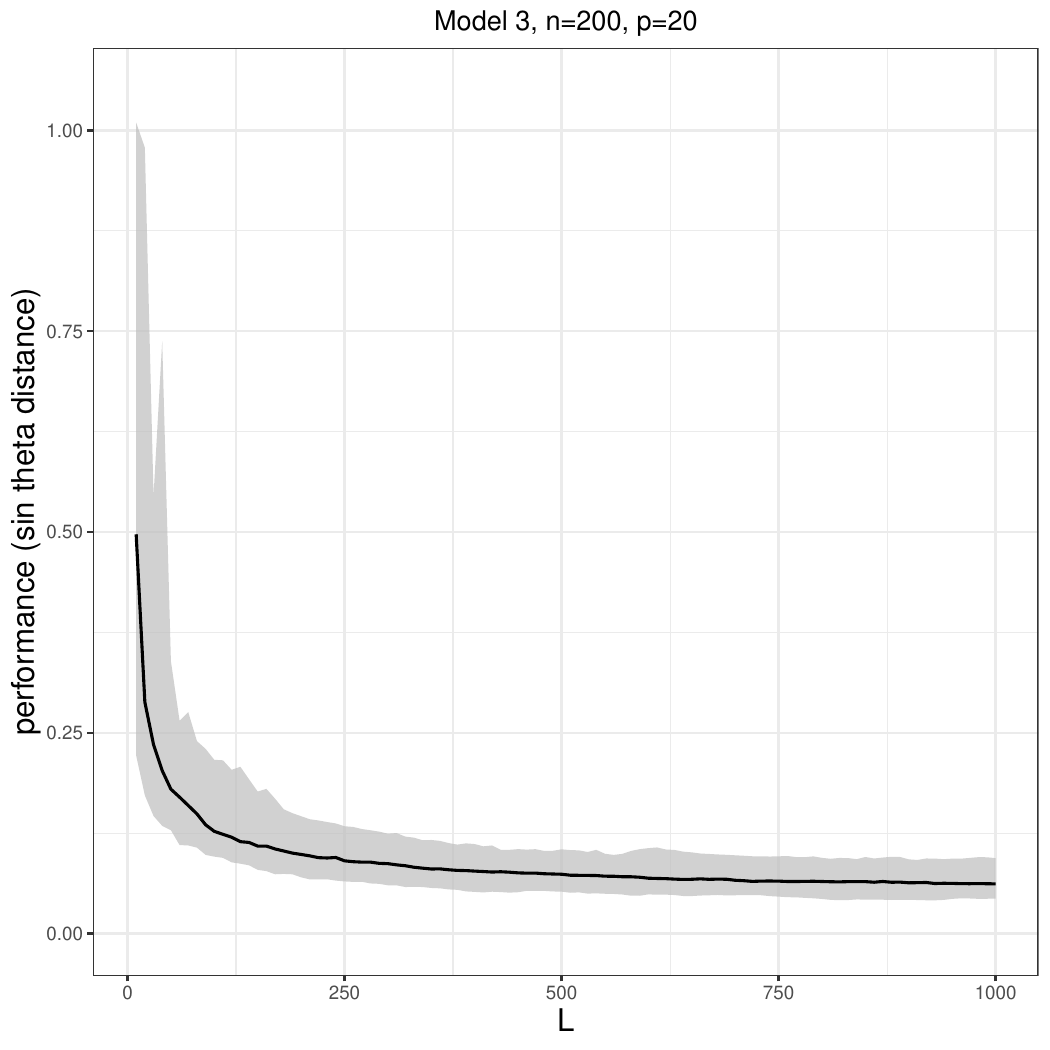}

    \includegraphics[width=0.3\textwidth]{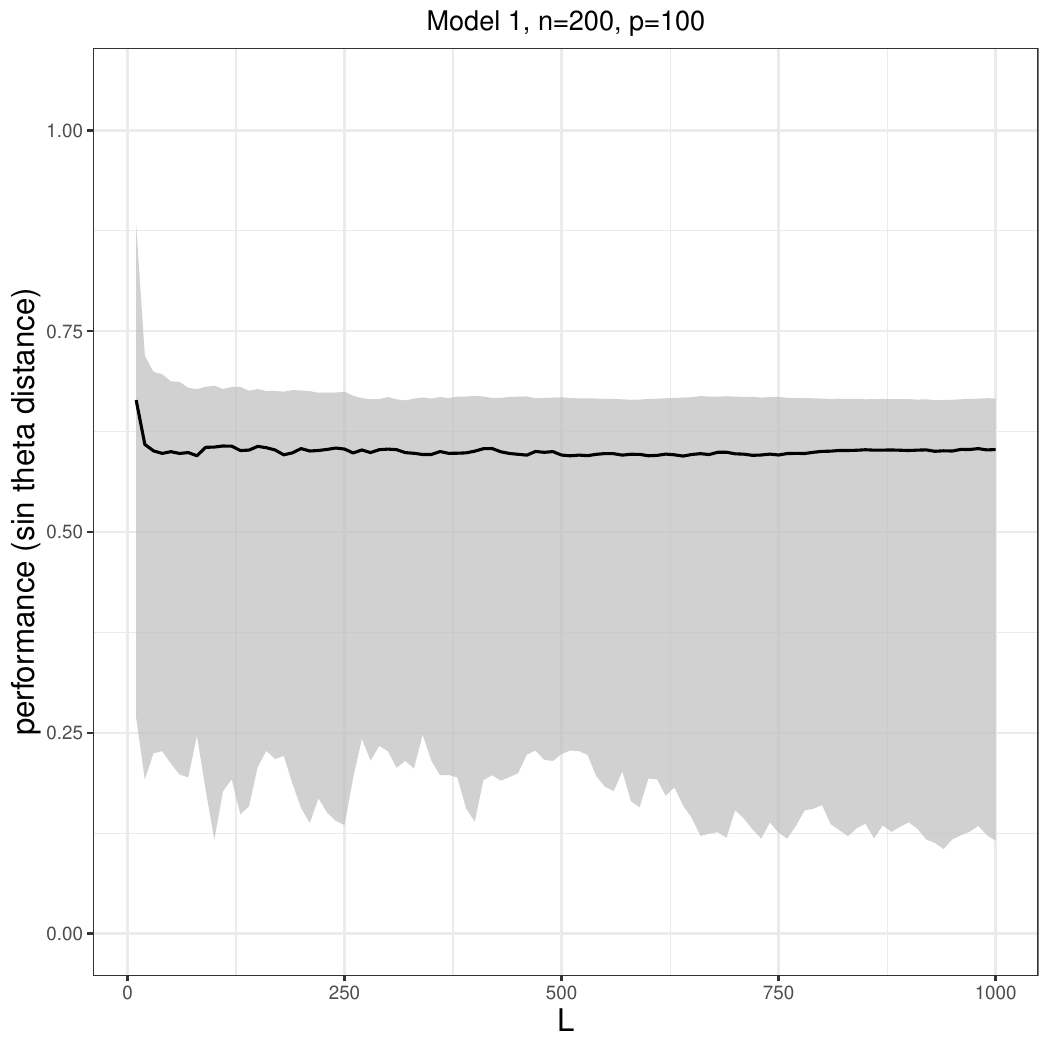} 
    \includegraphics[width=0.3\textwidth]{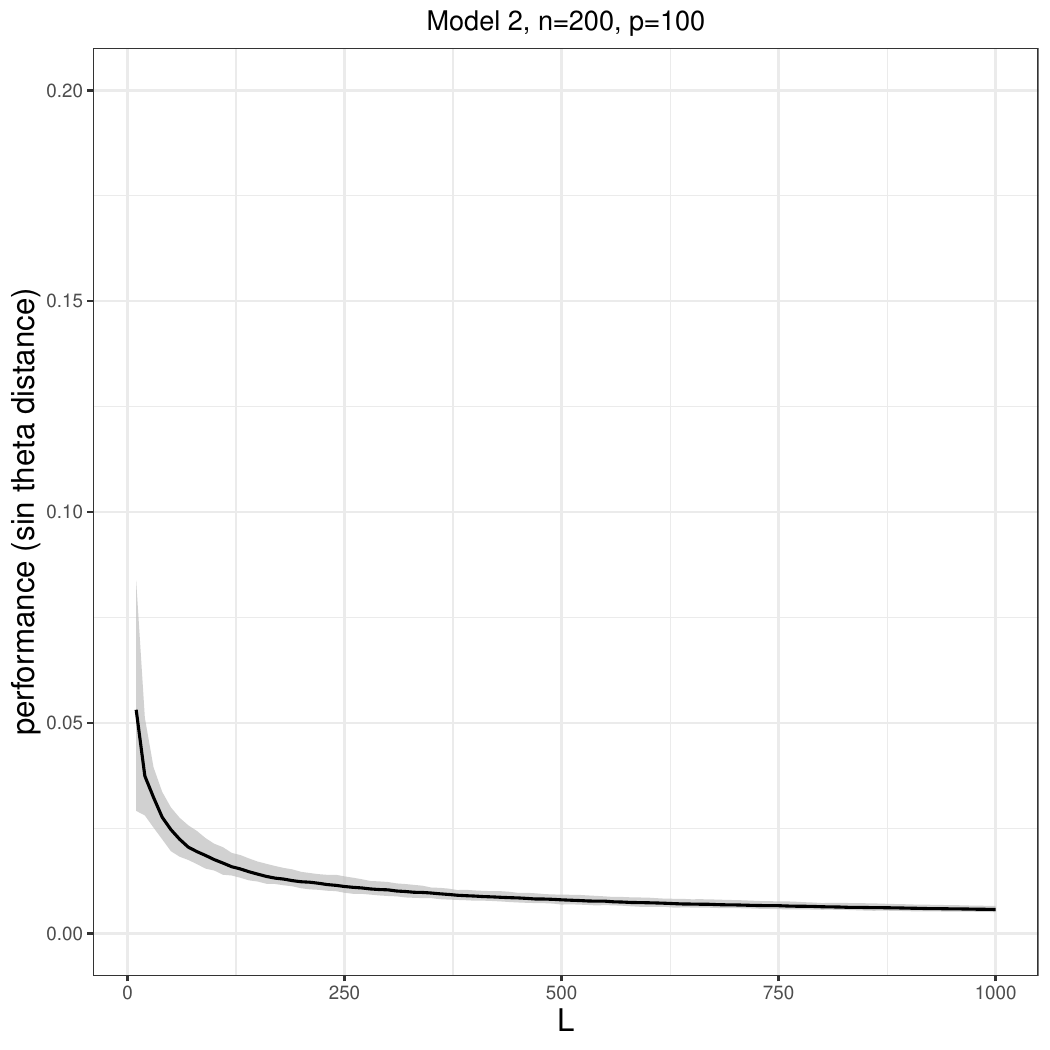}
   \includegraphics[width=0.3\textwidth]{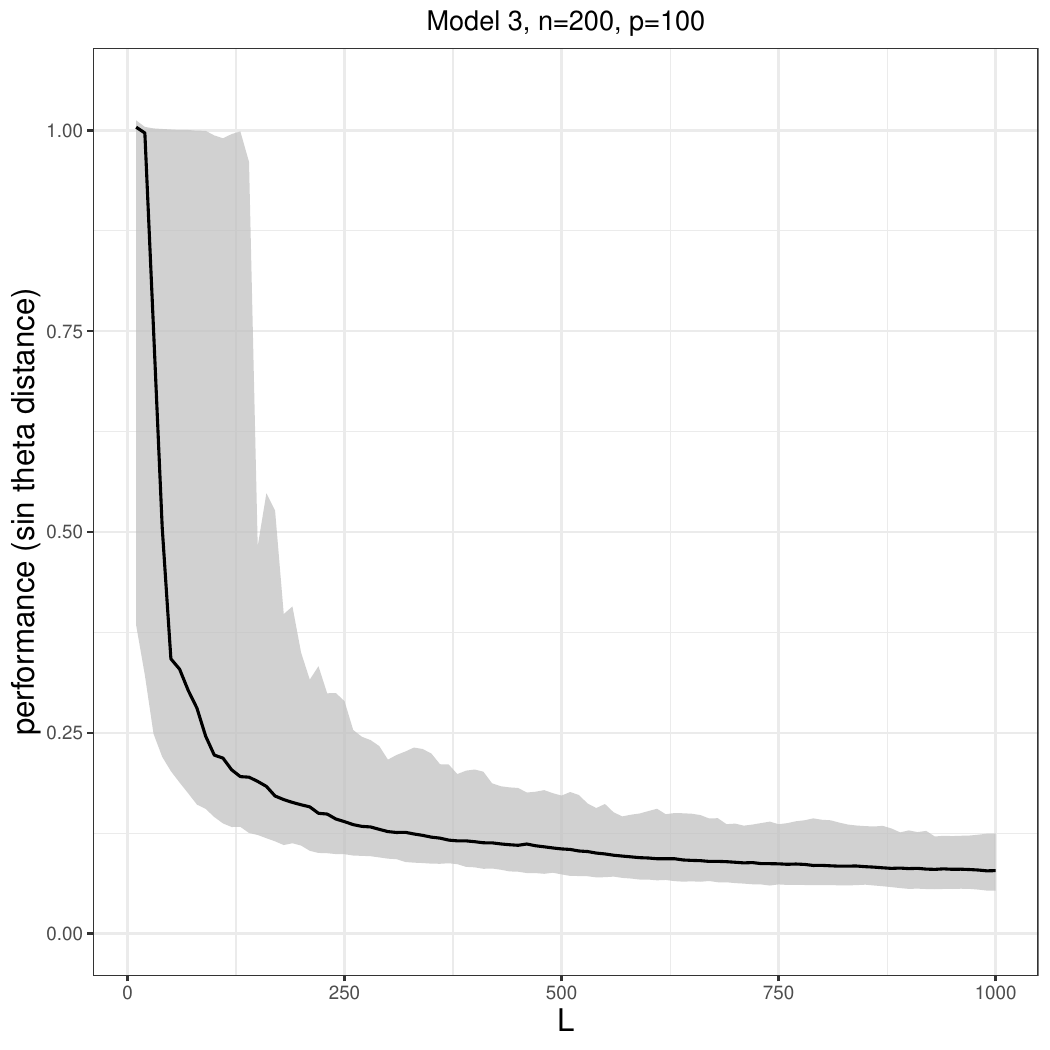}
    \caption{{\small The median (solid line) of the $\sin$-theta distance (see~\eqref{eq:dist}) between $U_{1:d_0}$ from the output of Algorithm~\ref{alg:RPEDR estimator} and the true direction $A_0$ over 100 repeats of the simulation as $L$ varies. We present the results for Model~1a (left), Model 2 (middle) and Model 3 (right), for $p=20$ (top row) and $p = 100$ (bottom row). The grey shaded region shows area between the 5\%-95\% quantiles for each value of $L$. \label{fig:increasingL}}}
\end{figure}

The empirical results in Figure~\ref{fig:increasingL} confirm that in these examples, we can expect the performance of our algorithm to improve as $L$ increases as indicated by Theorem~\ref{thm:choiceofL}.  In the results for Model 1, we observe a larger variance across different repeats of the experiment -- this is due in part to the fact that, in this example, our algorithm captures only part of the true signal in the leading order singular vector in $U$.  As discussed later in Section~\ref{sec:numerical}, our method performs well here in the sense that nearly all of the signal is found in the first two singular vectors (but this is not captured by the sin theta distance measure presented here). In fact, an extension of our method given in Section~\ref{sec:doubleRPEDR} is able to find the majority of the signal in a one dimensional projection using a double application of our procedure.

Turning now to the choice of $M$ and $d$, which together determine the total number of projection directions considered within each of the $L$ groups in Algorithm~\ref{alg:RPEDR estimator}. In slight contrast with the choice of $L$, there is a trade-off when choosing $d$ and $M$.  For instance, if $M$ is taken to be small (e.g.~less than 20 when $p = 20$) then it's unlikely that we have a good projection within a block of size $M$. On the other hand, if $M$ is taken to be very large (e.g.~greater than 1000 when $p =20$) we may start to overfit, although this effect is often negligible due to the data resampling strategy taken in Algorithm~\ref{alg:RPEDR estimator}. In our experiments we see that the performance of our algorithm levels off as opposed to deteriorating as $M$ increases.  Regarding the projection dimension $d$, if $d$ is taken to be too small, in particular less than (the unknown) $d_0$, then we will always miss one or more projection directions, as our algorithm will consistently choose only a more dominant dimension reduction direction. However, we aim to keep $d$ to be relatively small both for computational considerations and the fact that the base regression methods may suffer from the curse of dimensionality.   

To explore these aspects in practice, we repeat the experiments presented in Figure~\ref{fig:increasingL}, but now with $L$ fixed at the default recommendation of $200$, and vary $d \in \{1,2,5,10,20\}$ and $M \in [1000]$.  The results are presented in Figure~\ref{fig:choiceofMd}.  When $p=20$, we see that taking $d = 1$ is not effective across all three models considered, and there is an advantage in taking $d > d_0$.  To elucidate this point, consider an example where $d_0 = 1$. Setting $d=1$ requires one of the $M$ projections in each group to be relatively highly correlated with the true population direction, which is  rare unless $M$ is very large. With a larger $d$, it is more likely that we find a projection with some signal, which then leads to good performance in the final aggregation step of our algorithm.  However, if $d$ is too large, the base method may start to overfit, as seen in, for example, Model 1 when $p=20$ and $d = 10$.   

\begin{figure}[!ht]
    \centering
    \includegraphics[width=0.32\textwidth]{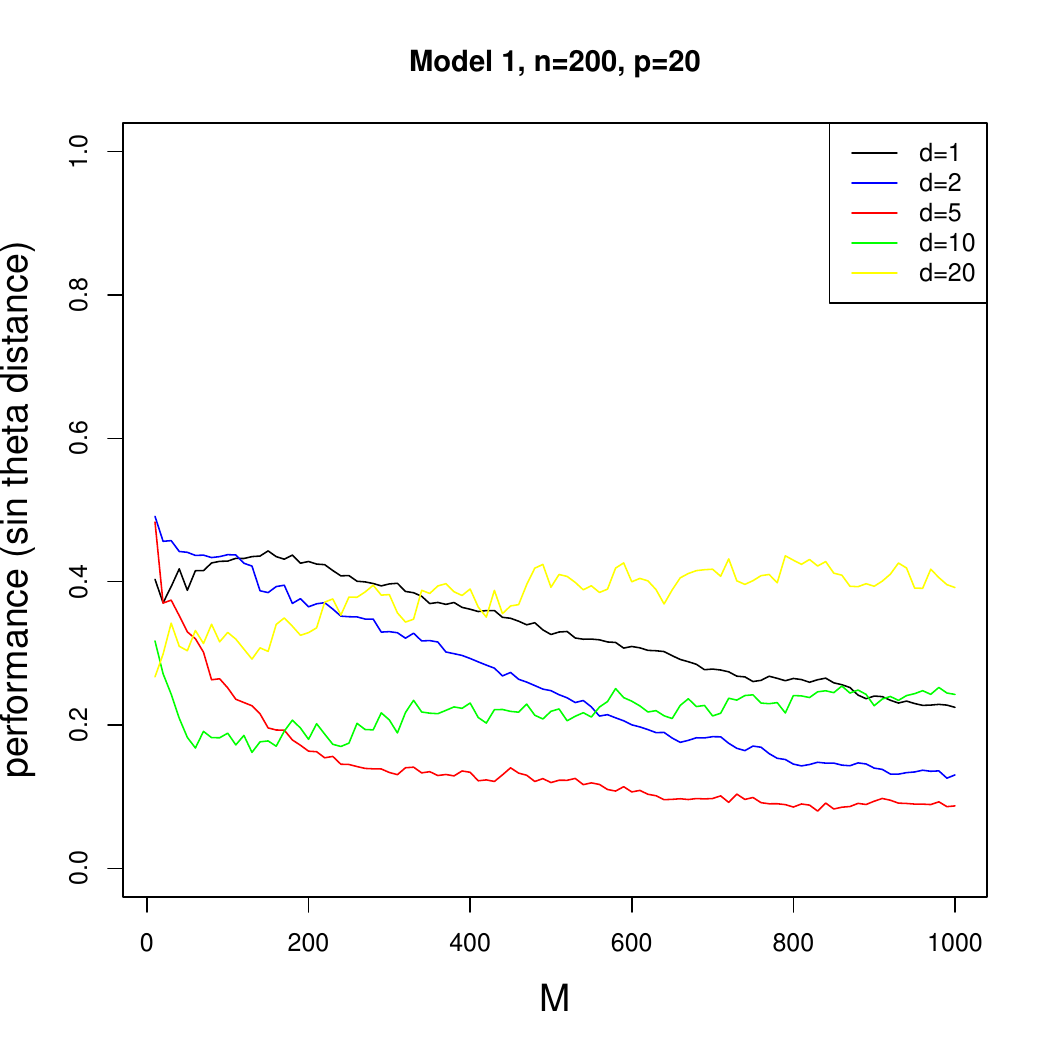}
        \includegraphics[width=0.32\textwidth]{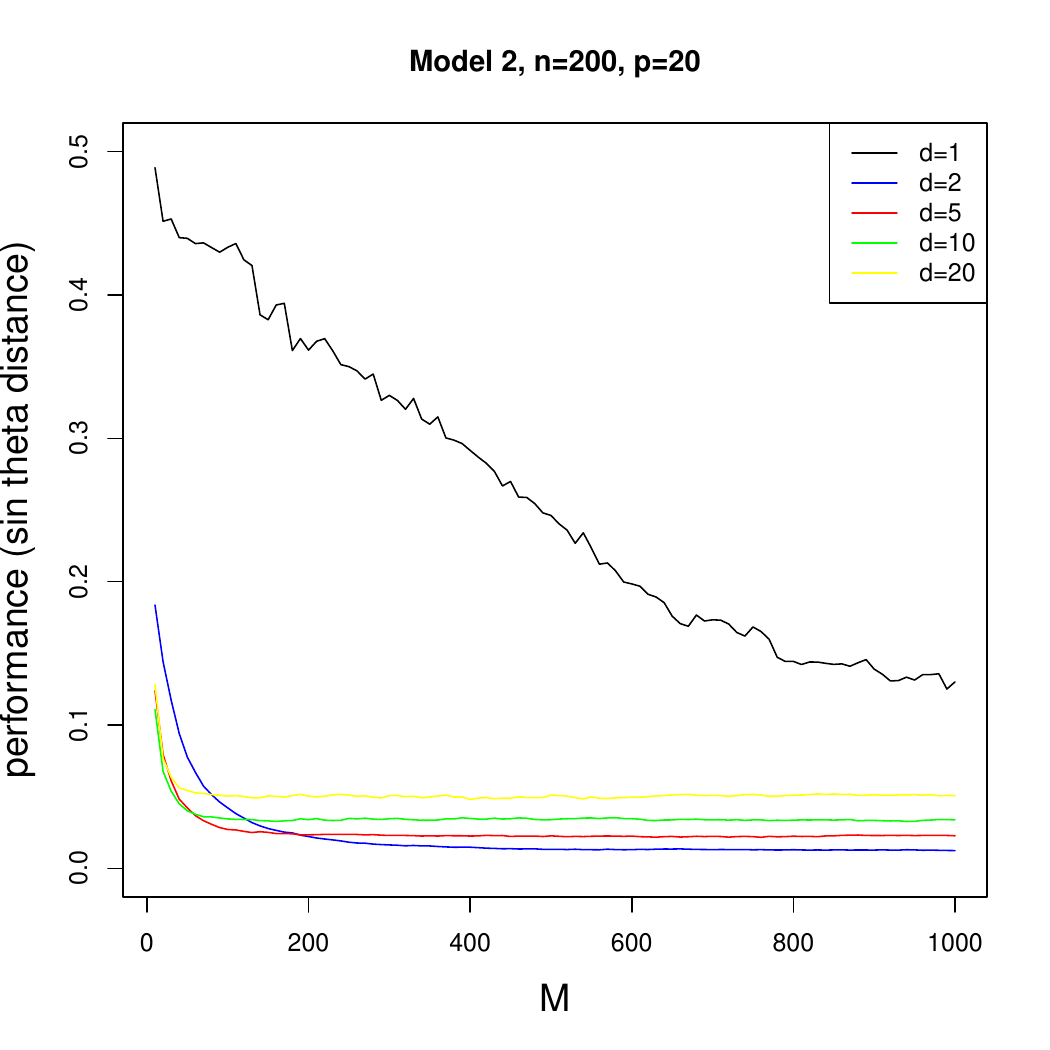}
   \includegraphics[width=0.32\textwidth]{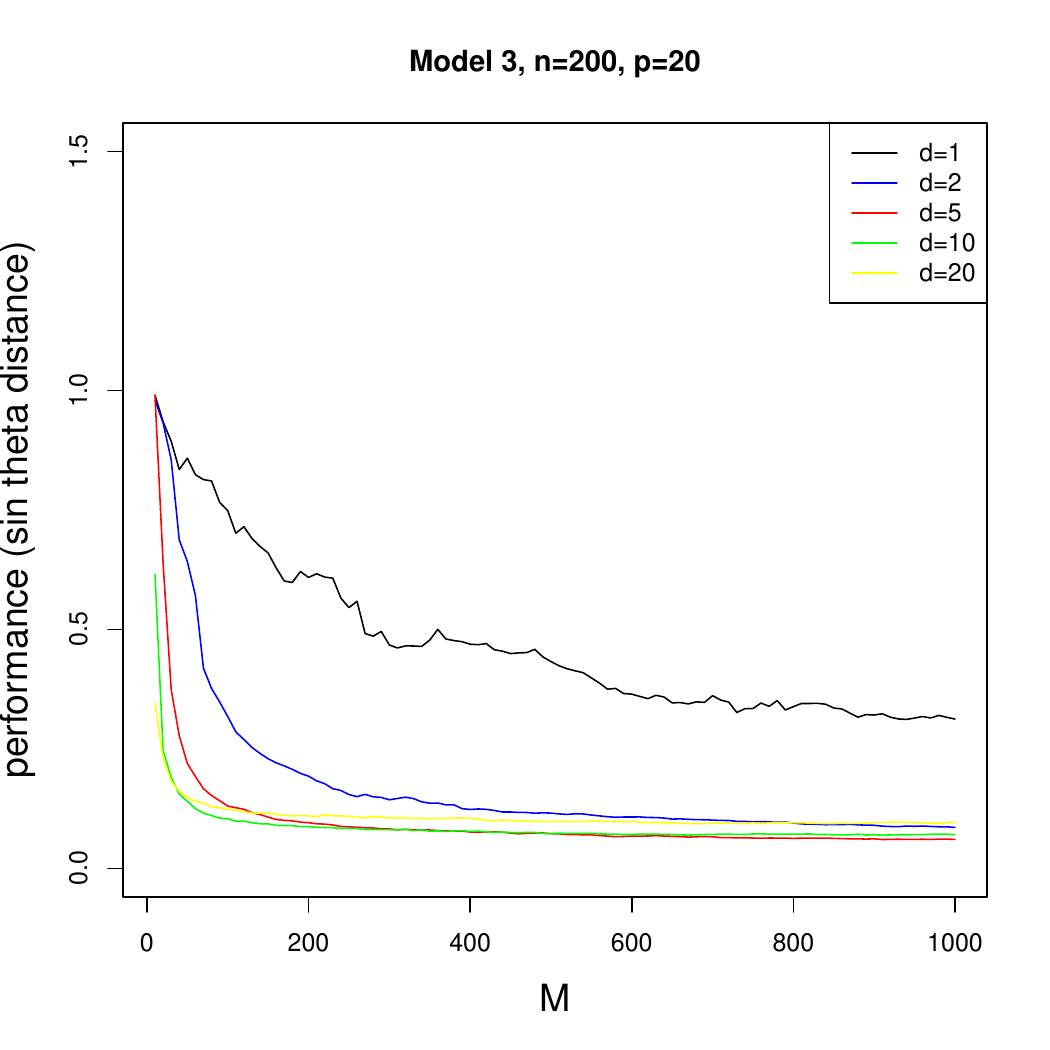}

    \includegraphics[width=0.32\textwidth]{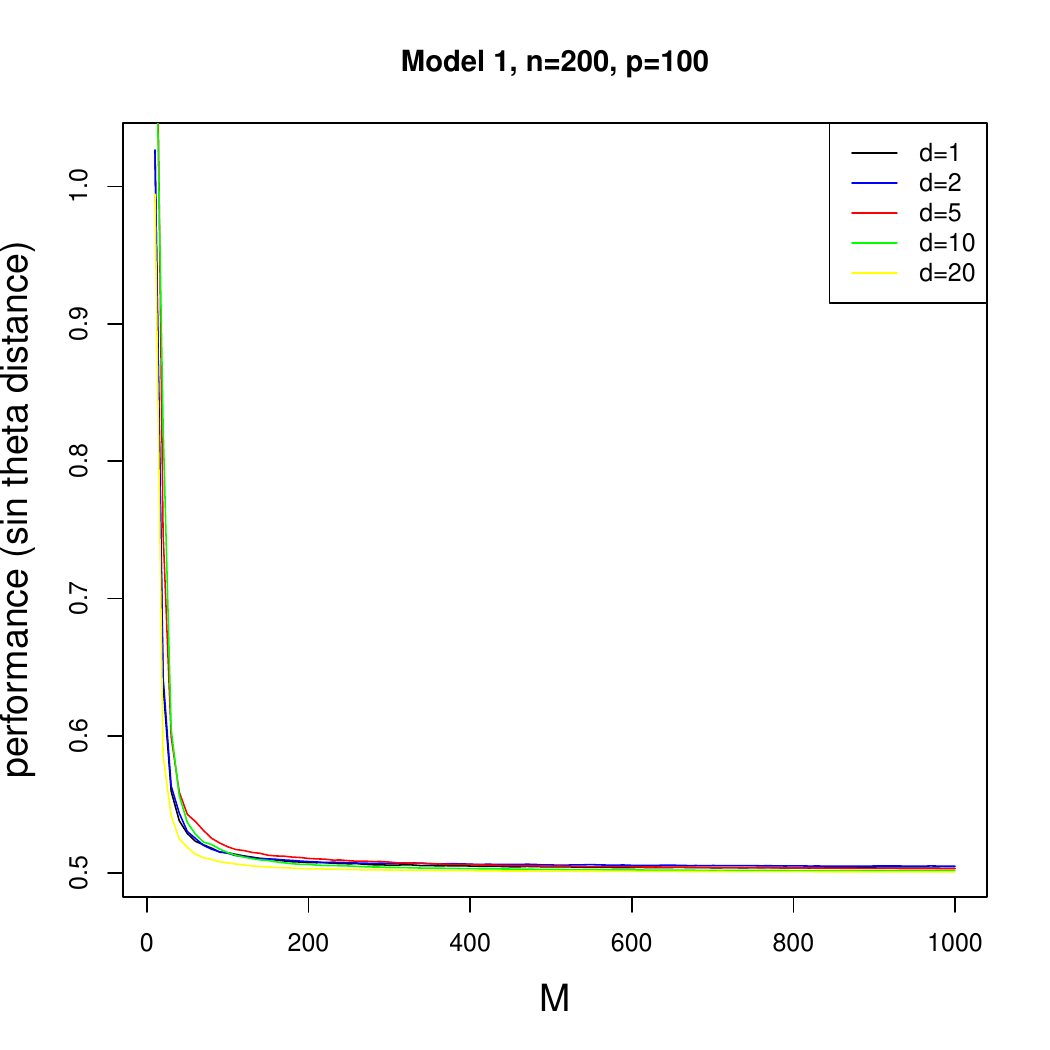}
    \includegraphics[width=0.32\textwidth]{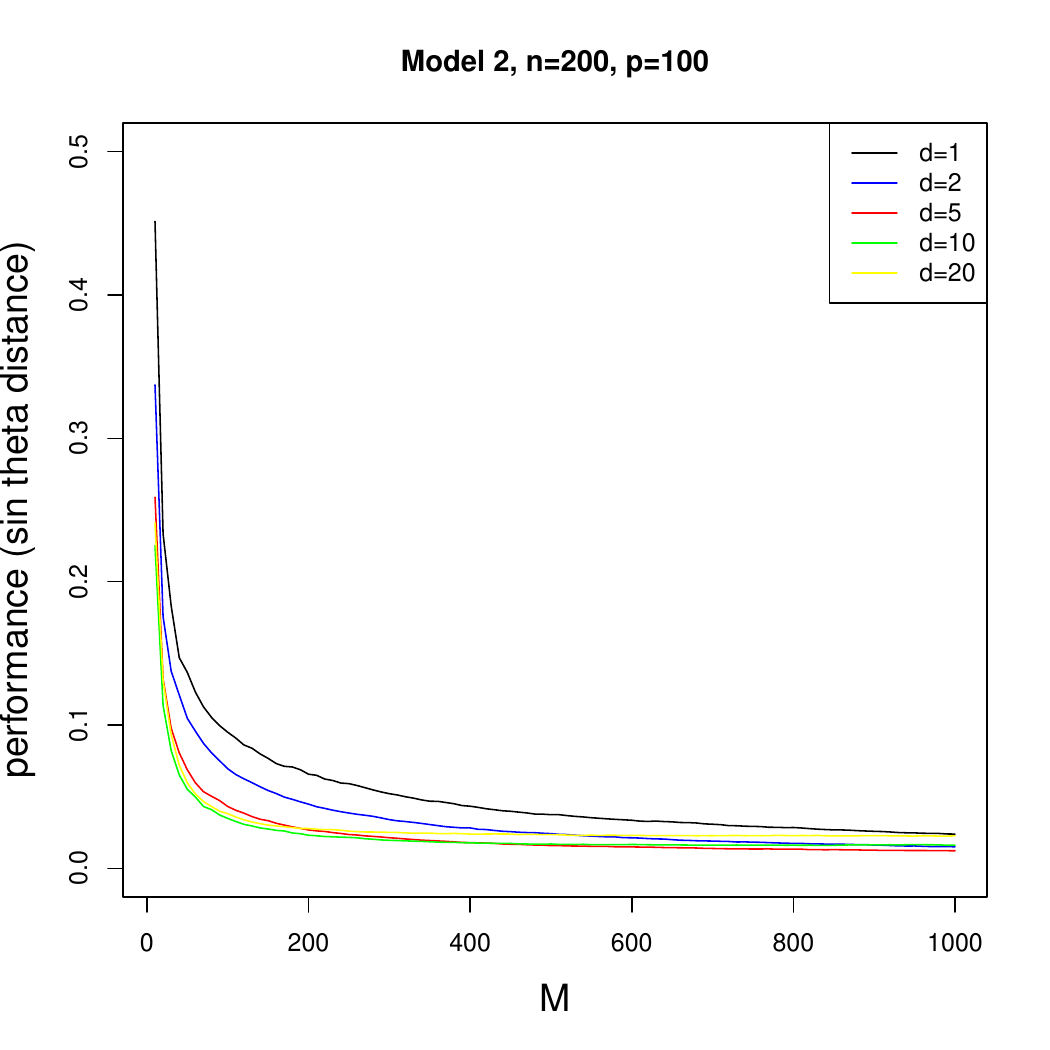}
   \includegraphics[width=0.32\textwidth]{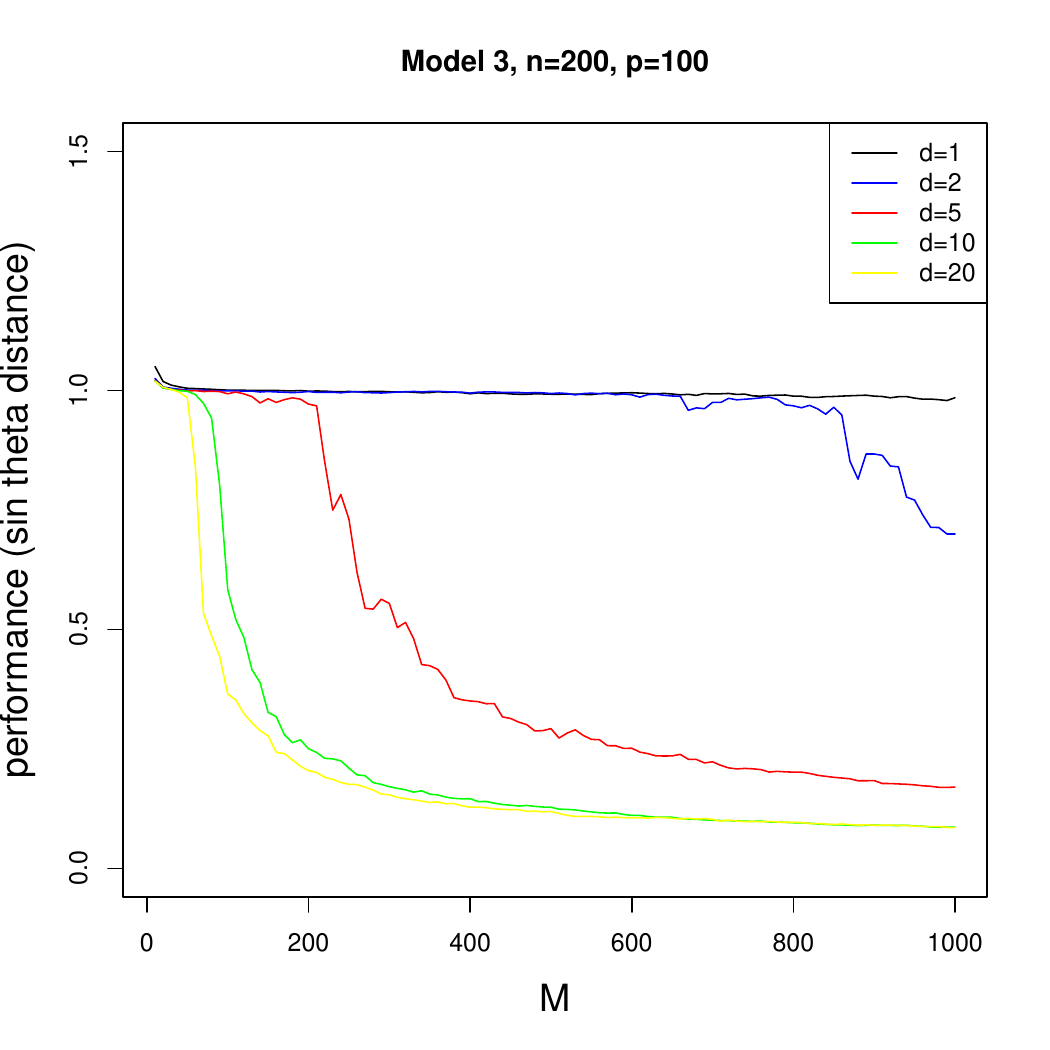}
    \caption{{\small The median of the $\sin$-theta distance (see~\eqref{eq:dist}) between $U_{1:d_0}$ from the output of Algorithm~\ref{alg:RPEDR estimator} and the true direction $A_0$ over 100 repeats of the simulation as $M$ varies from 10 to 2000, with $d=1$ (black), $d=2$ (blue), $d=5$ (red), $d=10$ (green) and $d = 20$ (yellow).  We present the results for Model~1a (left), Model 2 (middle) and Model 3 (right), for $p=20$ (top row) and $p = 100$ (bottom row). }\label{fig:choiceofMd}}
\end{figure}

Regarding $M$, we see that when $p = 20$ and $d = 5$, the performance of our proposal improves rapidly as $M$ increases, but then levels off for $M$ between 100 and 200. When $p = 100$ in Model~1, $M$ needs to be very large in order to capture all signal in the first singular vector of $U$. In fact, nearly all of the signal is contained in the first two singular vectors of $U$, even for moderately small values of $M$.  For Model~3, with $p=100$, the algorithm finds one of the $d_0 = 2$ directions for small values of $M$ (for every $d$), but larger $d$ and $M$ are needed to capture the second direction. This explains the two plateaux seen in the curves in bottom right plot.  Overall, we recommend taking $d = \lceil \sqrt{p} \rceil$ and $M = 10p$ as a sensible default choices.

\subsection{Choice of \texorpdfstring{$\hat{d}_0$}{d0 hat}}
\label{sec:choiceofd0}
In some scenarios, the user may have a predetermined lower dimension, $\hat{d}_0$ say, in mind. In which cases, as mentioned above, they should simply output the first $\hat{d}_0$ columns of $U$ from the output of Algorithm~\ref{alg:RPEDR estimator} and set $\hat{A}_0 = (U_1, \ldots, U_{\hat{d}_0})$.  Indeed, the singular vectors of $U$, in some sense, represent the most frequently chosen dimensions (in decreasing order of importance) across the groups of projections in the algorithm. 

In situations where the desired dimension is unknown, we propose to use the information contained in the singular values $D$ from the output of Algorithm~\ref{alg:RPEDR estimator}.  The main idea is to compare these singular values with median of those obtained from applying the corresponding procedure to random projections with the same distribution (i.e.~without selection within each group of size $M$) -- see Algorithm~\ref{alg:DimensionEstimator} below. 

\begin{algorithm}[!ht]
\caption{EDR Dimension Estimator}
\label{alg:DimensionEstimator}
\textbf{Input}: Projection dimension $d \in [p]$, projection distribution $Q$ on $\mathcal{Q}_{p\times d}$, number of groups $L$, and the vector $\mathrm{diag}(D) = (D_1, \ldots, D_p)$ of singular values from the output of Algorithm~\ref{alg:RPEDR estimator}, as well as a random projection resample size $R\in \mathbb{N}$. \\

Let $\mathbf{P}_{1,1},...,\mathbf{P}_{L,R} \stackrel{\mathrm{i.i.d}}{\sim} Q$.

Set $\hat{\Pi}^{(r)} = \frac{1}{L} \sum_{\ell=1}^{L} \mathbf{P}_{\ell, r} \mathbf{P}_{\ell, r}^T$, for $r \in [R]$.

Calculate the singular value decomposition of $\hat{\Pi}^{(r)} = U^{(r)} D^{(r)} (U^{(r)})^T$ and let $\mathrm{diag}(D^{(r)}) =  (D_1^{(r)}, \ldots,  D_p^{(r)})$. 

Let $T_j = \frac{1}{R} \sum_{r=1}^R \mathbbm{1}_{\{D^{(r)}_j \leq D_j\}}$

\textbf{Output}: The dimension $\hat{d}_0 := \max\{j \in [p] : T_{\ell}  > 1/2,   \text{ for all } \ell \leq j\}$.\footnote{Here the maximum over the emptyset is taken to be zero, with the interpretation being that this is evidence that there is no correlation between $X$ and $Y$.}  
\end{algorithm}

 The  motivation behind our approach is that, for a particular direction to be selected, there should be evidence that it contains more signal than expected at random. For the first singular value, this asks that there is a dimension which is more likely to be selected by our algorithm than what happens at random.  For later columns of $U$, since our projections are rescaled to have trace $d$, we have 
\[
\sum_{j = 1}^p D_j = \mathrm{trace}(\hat{\Pi}) = d.
\] 
Thus there is a natural penalty to selecting more dimensions in cases where there are relatively large singular values earlier on in $D$.

\begin{figure}[!ht]
    \centering
    \includegraphics[width=0.44\textwidth]{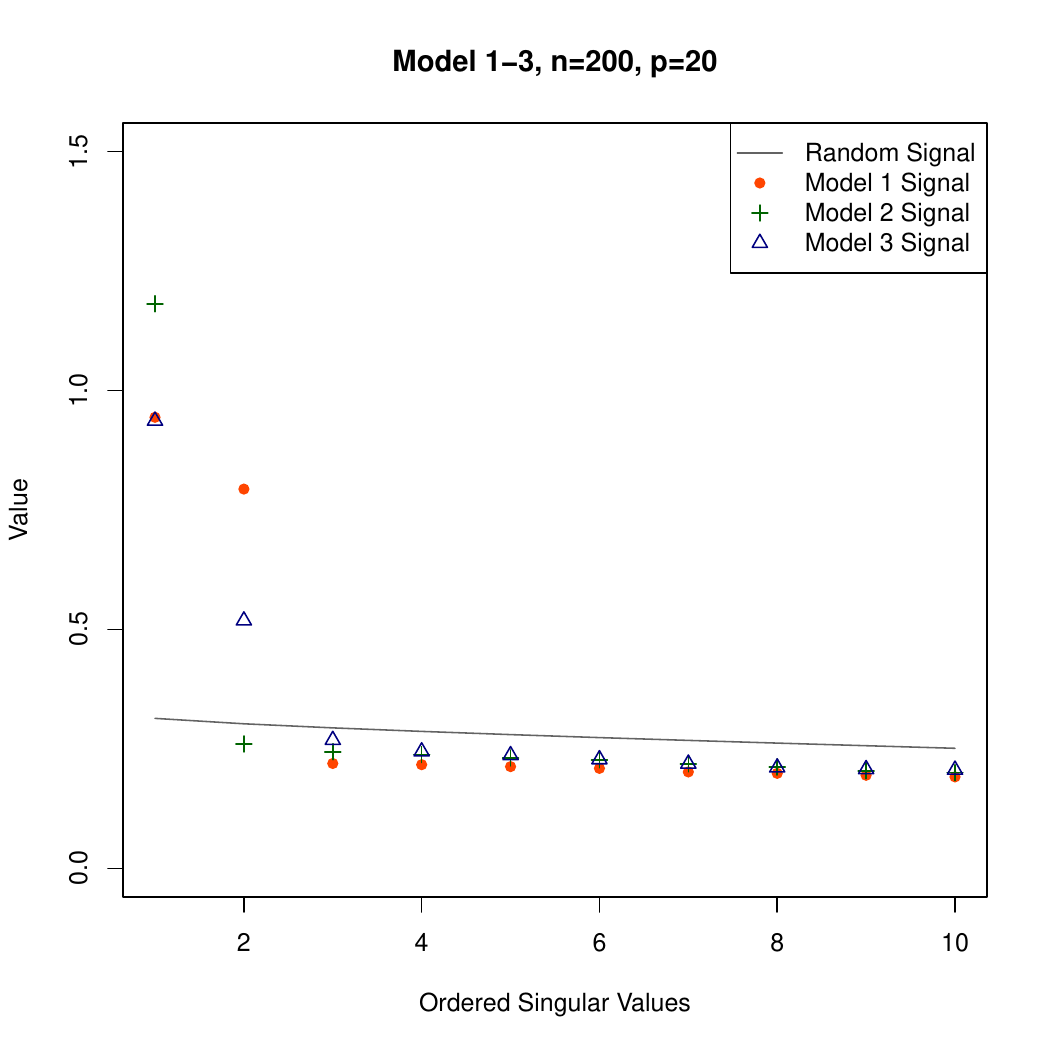}
    \includegraphics[width=0.44\textwidth]{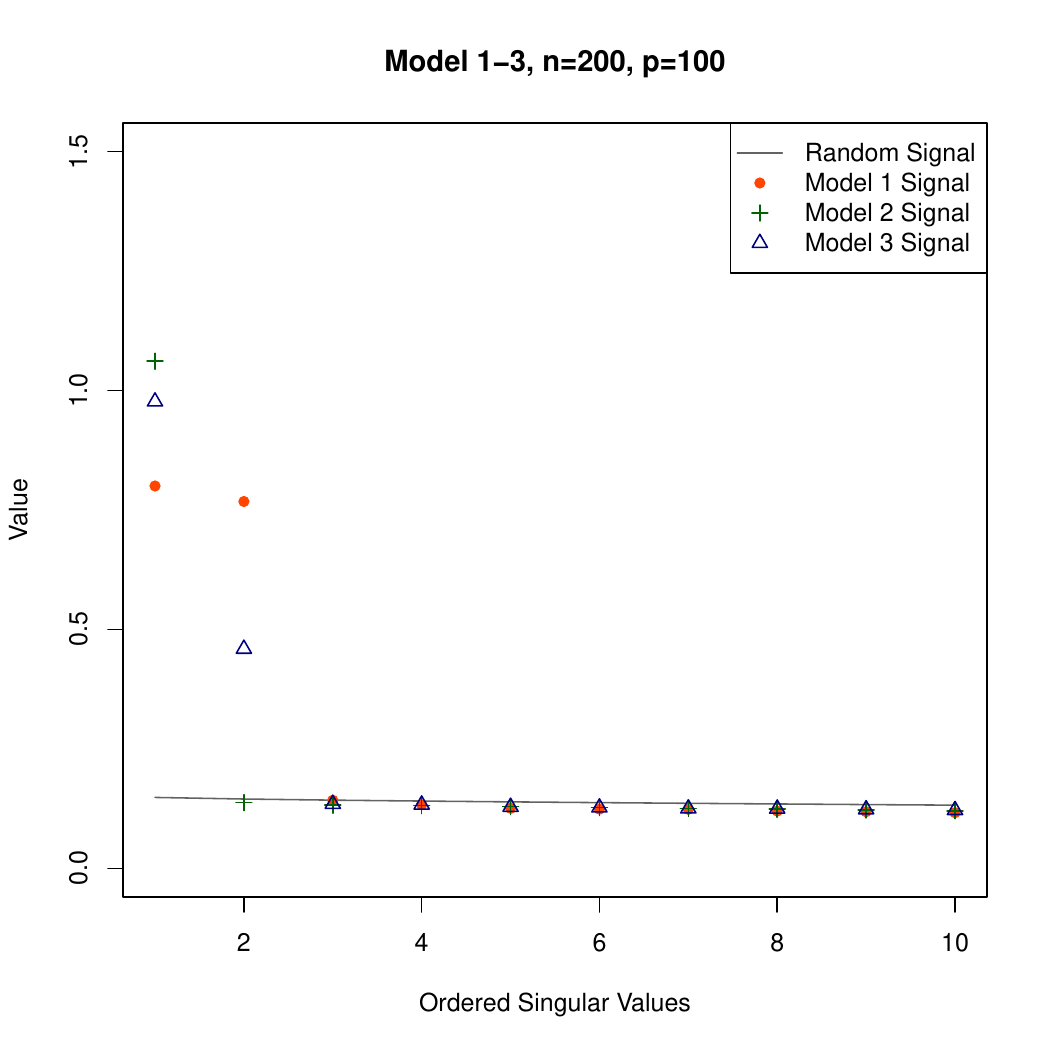}

    \caption{\small Comparison between the vector $D$ of singular values from the output of Algorithm~\ref{alg:RPEDR estimator} for Models 1a, 2 and 3, and the median of corresponding $D^{(r)}$ in Algorithm~\ref{alg:DimensionEstimator} for $p=20$ (left) and $p =100$ (right).  We present only the first 10 singular values in each case. \label{fig:Alg2Output}}
\end{figure}

Figure~\ref{fig:Alg2Output} demonstrates our procedure for selecting $\hat{d}_0$ in Models 1a, 2 and 3.  For Model~1a, both when $p=20$ and $p=100$, the output suggests retaining the first two singular vectors.  This is in part because, as also seen in Figure~\ref{fig:Alg1demo}, the single signal $\frac{1}{\sqrt{2}} (X_1+X_2)$ is often split into two. The reason in part due to the fact, even though the projection $(1/\sqrt{2}, -1/\sqrt{2}, 0, \ldots, 0)^T$ is orthogonal to the true signal direction, the response $Y$ is not (marginally) independent of $X_1 - X_2$ and therefore our algorithm suggests selecting (approximately) this direction.  We demonstrate in Section~\ref{sec:doubleRPEDR} below how applying our algorithm twice in this cases can accurately estimate only the true one-dimensional signal.   For the other two models, there is a clear indication that we should take $\hat{d}_0 = d_0$, namely $\hat{d}_0 = 1$ in Model~2 and $\hat{d}_0 = 2$ in Model~3 and in fact we recover most of the signal here; for example, for Model~2 when $p=20$ we have $|U_1^T e_3| > 0.999$ and for Model~3, we have $\min\{|U_1^Te_6|,|U_2^Te_7|\} \geq 0.995$.  We delay further evaluation of the this approach to selecting $\hat{d}_0$ to our full simulation study in Section~\ref{sec:numerical}.

\subsection{Computational considerations}
\label{sec:comp} 
The majority of the computational cost of our main algorithm arises from the fact that the base regression method is applied $L \cdot M$ times. Moreover, the cost of applying the base method generally depends on the projection dimension $d$ and the sample size $n_1$.  If the ambient dimension $p$ is relatively modest (i.e.~in the hundreds), the cost of projecting the data and performing the final singular value decomposition is relatively low.  Notably, the $L \cdot M$ applications of the base regression method can be parallelised, meaning that the overall computational burden can be significantly reduced if high-performance computing resources are available. 

Assessing the full trade off in computational resources and statistical performance in selecting the inputs to Algorithm~\ref{alg:RPEDR estimator} discussed in this section is beyond the scope of the work here. Nevertheless in Section~\ref{sec:runtime} we show that the runtime of our algorithm, while slower than the competitors considered, is manageable with our recommended settings even when $p$ and $n$ are large. 

\section{Double random projection ensemble dimension reduction}
\label{sec:doubleRPEDR}
In this section, we explore how applying our main algorithm a second time may help to further reduce the dimension in cases where the used is not satisfied with the suggested dimension from an initial application of Algorithms~\ref{alg:RPEDR estimator} and~\ref{alg:DimensionEstimator}. For example, we may have knowledge of the true dimension $d_0$, but Algorithm~\ref{alg:DimensionEstimator} suggests  $\hat{d}_0 > d_0$. In other cases, we may not know the true $d_0$ precisely, but may still have a specific target dimension in mind. 

We will demonstrate that this double application strategy is particularly effective when the number of  covariates contributing to $\mathcal{S}(A_0)$ outnumber the true dimension $d_0$. In particular, in such situations, a single application of our algorithm may favour the coordinate axes of the relevant covariates, but fail to combine these effectively into the smallest possible number of dimension reduction directions. One perspective of this approach, is that the first application of our algorithm acts as a screening step to select potentially relevant directions, while the second application then combines these directions to yield a lower dimensional projection that still retains a large amount of the signal.  The full procedure is presented in Algorithm~\ref{alg:doubleDR}. 

\begin{algorithm}[!ht]
\caption{Double RPE dimension reduction}
\label{alg:doubleDR}

\textbf{Input}: Data $((x_1,y_1),...,(x_n,y_n)) \in (\mathbb{R}^p \times \mathbb{R})^n$ and desired dimension $\check{d}_0 \in [p]$. \\

Let $U$ and $D$ be the output of Algorithm~\ref{alg:RPEDR estimator} applied to $(x_1,y_1),...,(x_n,y_n)$ with $L = 200$, $M=10p$, $d = \lceil \sqrt{p} \rceil$, $n_1 = \lceil 2n/3 \rceil$,  $Q = \frac{1}{2} Q_{\mathrm{C}}^{\otimes d} + \frac{1}{2} Q_{\mathrm{N}}^{\otimes d}$, and $\hat{g} = \hat{g}_{\mathrm{MARS}}$.
 \\

Let $\hat{d}_0$ be the output of Algorithm~\ref{alg:DimensionEstimator} applied to $D$,  with $L = 200$,  $d = \lceil \sqrt{p} \rceil$, $Q = \frac{1}{2} Q_{\mathrm{C}}^{\otimes d} + \frac{1}{2} Q_{\mathrm{N}}^{\otimes d}$. 

\eIf{$\hat{d}_0 > \check{d}_0$}{ 
Set $\hat{A}_0 := (U_1, \ldots, U_{\hat{d}_0})$ and $z_i := \hat{A}_0^T x_i$ for $i \in [n]$. 

Let $\check{U}$ and $\check{D}$ be the output of Algorithm~\ref{alg:RPEDR estimator} applied to $(z_1,y_1),...,(z_n,y_n)$ with $L = 200$, $M=10 \hat{d}_0$, $d = \lceil \sqrt{\hat{d}_0} \rceil$, $n_1 = \lceil 2n/3 \rceil$,  $Q = \frac{1}{2} Q_{\mathrm{C}}^{\otimes d} + \frac{1}{2} Q_{\mathrm{N}}^{\otimes d}$, and  $\hat{g} = \hat{g}_{\mathrm{MARS}}$.
 \\

Let $\check{A}_0 = \hat{A}_0 (\check{U}_1, \ldots, \check{U}_{\check{d}_0}) \in \mathbb{R}^{p \times \check{d}_0}$
}
{
Let $\check{A}_0 = (U_1, \ldots, U_{\check{d}_0})\in \mathbb{R}^{p \times \check{d}_0}$
}

\textbf{Output}: The projection $\check{A}_0$. 

\end{algorithm}

We demonstrate this procedure using Model 1a and a new model as follows: 
\begin{itemize} 
\item Model 4: Let $X = (X_1, \ldots, X_p)^T \sim N_p(0,I)$ and 
\[
 Y = \exp\Bigl( \frac{X_1 - X_2 + X_3}{3}\Bigr) + \epsilon
\] 
with $\epsilon \sim N(0,1/5)$ independent of $X$.  Here $A_0 = \frac{1} {\sqrt{3}}(1,-1,1,\mathbf{0}_{p-3})^T $ is $1$-dimensional. 
\end{itemize}

\begin{figure}[!ht]
    \setlength\tabcolsep{6pt}
    \adjustboxset{width=\linewidth, valign=c}
    \centering
    \begin{tabularx}{1.0\linewidth}{@{}
      l
      X @{\hspace{6pt}}
      X
    @{}} 
    \rotatebox[origin=c]{90}{\textbf{Model 1a}}
    & \includegraphics[scale=0.8]{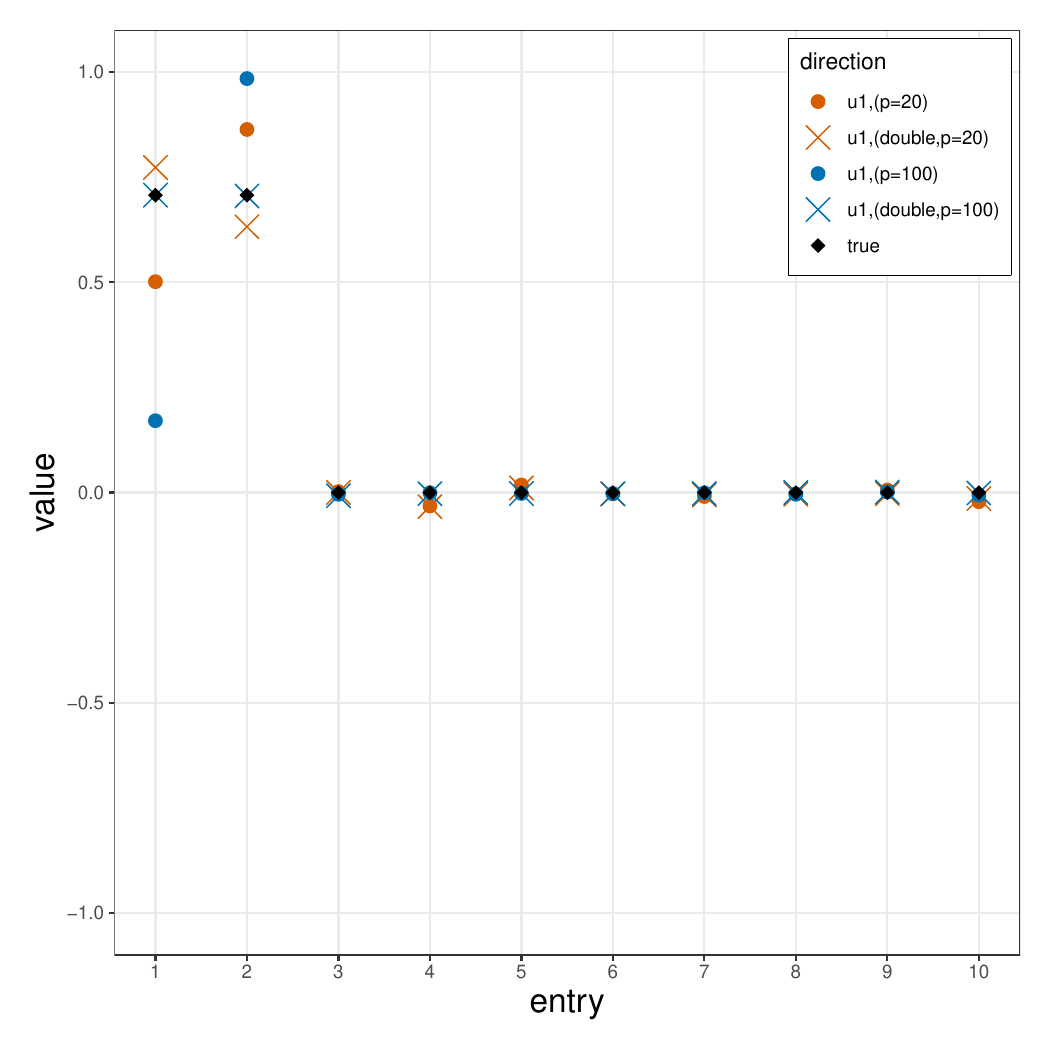}
    & \includegraphics[scale=0.8]{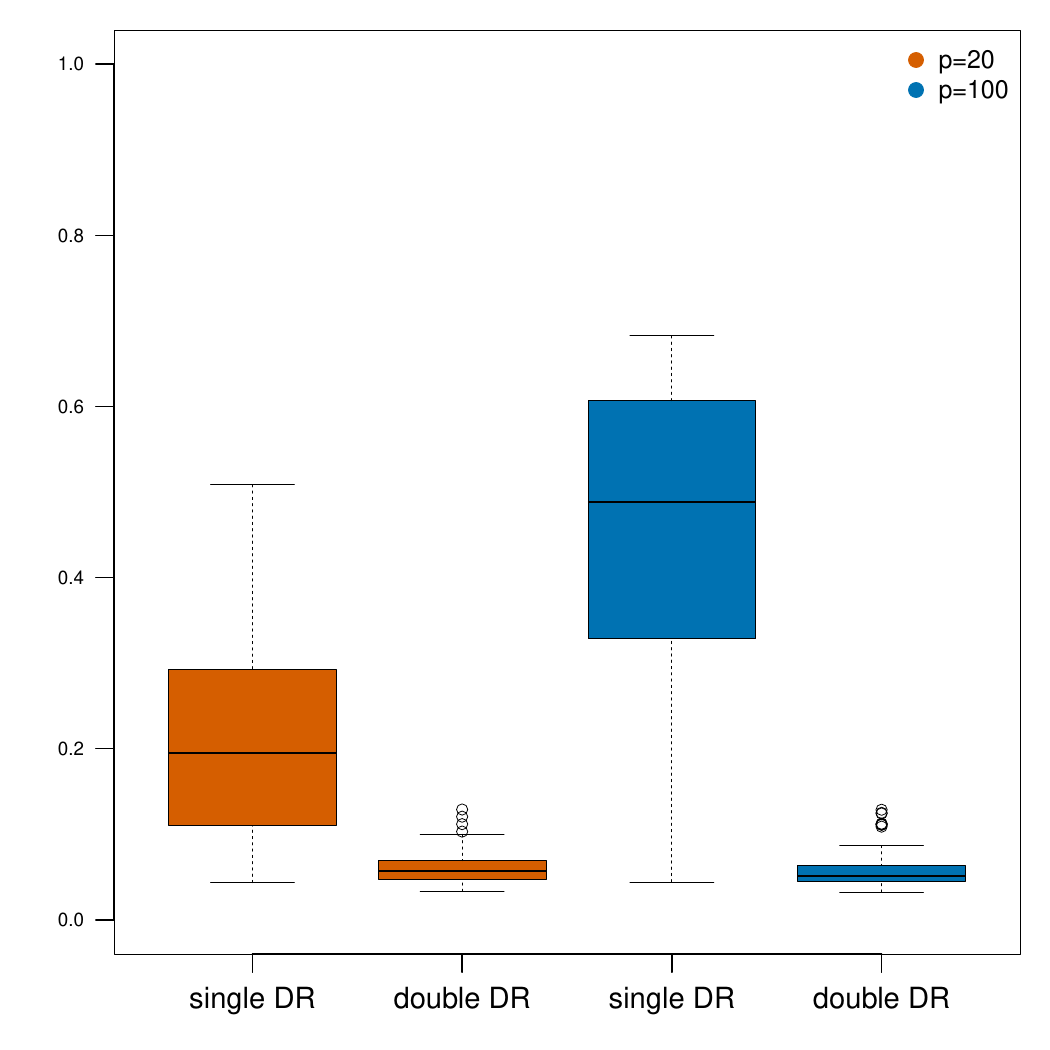} \\
    \rotatebox[origin=c]{90}{\textbf{Model 4}}
    & \includegraphics[scale=0.8]{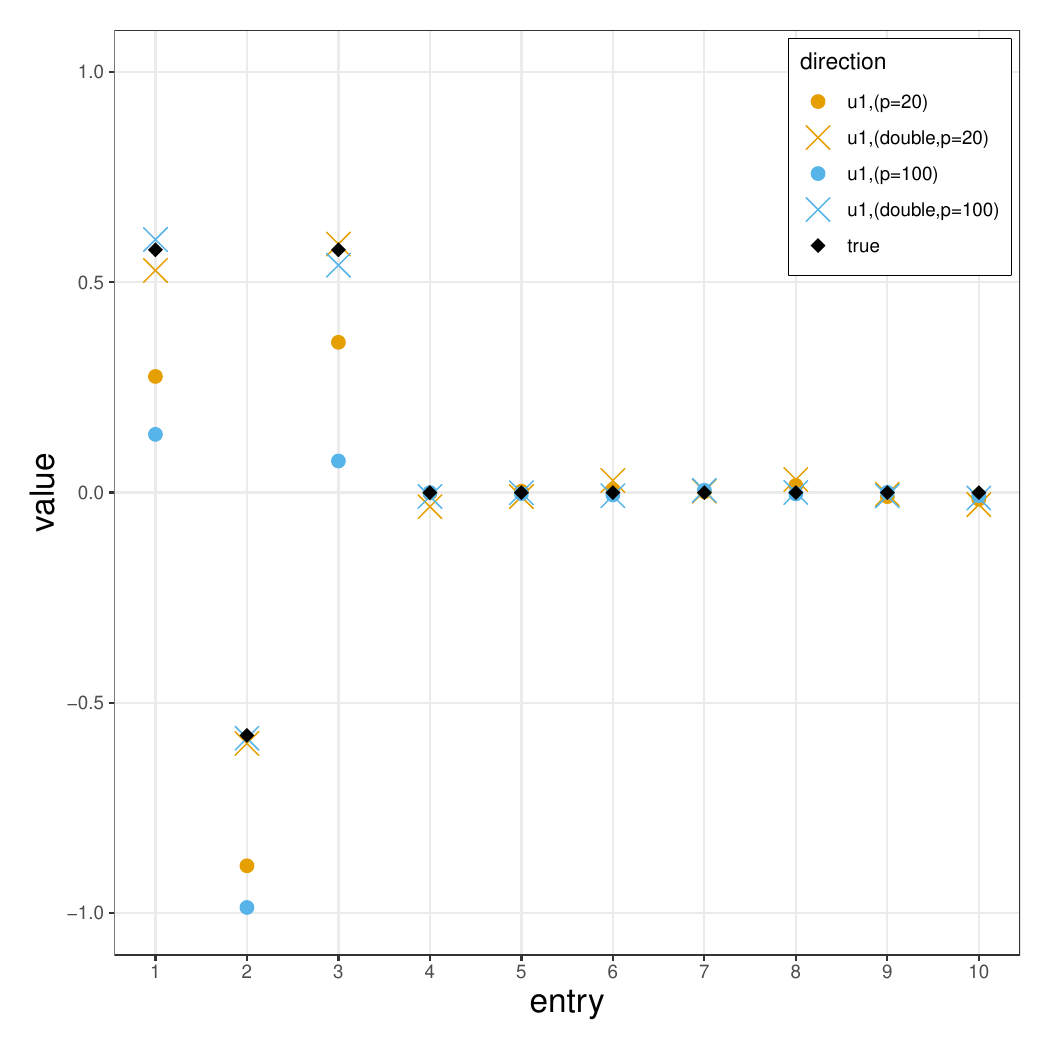}
    & \includegraphics[scale=0.8]{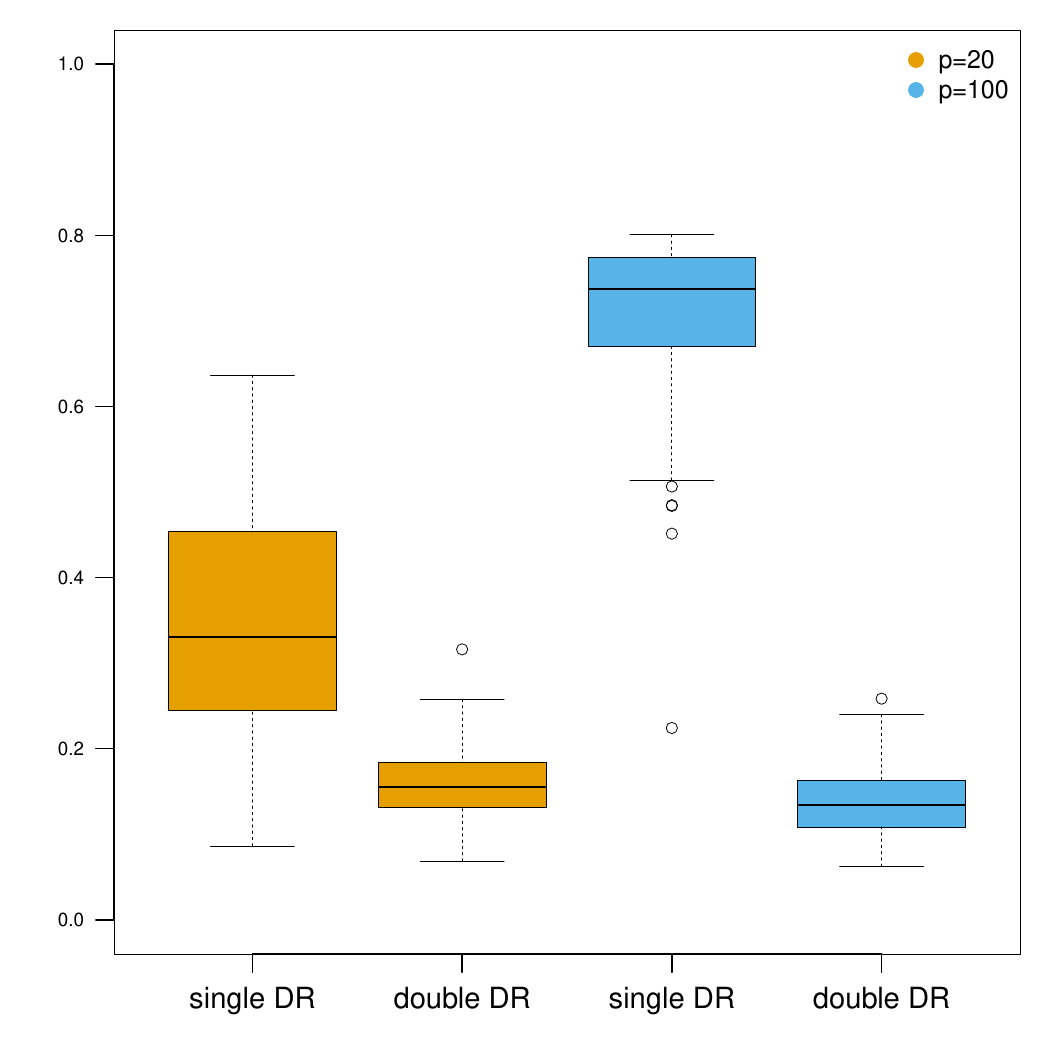}
    \end{tabularx}
    \caption{\small Demonstration of our Double RPE dimension reduction approach in Algorithm~\ref{alg:doubleDR}.   
We compare the results of a single application of our Algorithm~\ref{alg:RPEDR estimator} taking $\hat{d}_0 = 1$, with double application in Algorithm~\ref{alg:doubleDR} with $\check{d}_0 = 1$. We present the first 10 entries of the estimated $A_0$ for a single run of the experiment (left) alongside the sin-theta distance (right) over 100 repeats of the experiment. We have $p \in \{20,100\}$ and Model~1a (top) and Model~4 (bottom).}
\label{fig:Demo of DoubleDR}
\end{figure}

In both Models 1a and 4, the true projection $A_0$ is one-dimensional and therefore we seek to find the best one dimensional projection based on our algorithms.  Figure~\ref{fig:Demo of DoubleDR} shows the comparison between the leading singular vector from the output of Algorithm~\ref{alg:RPEDR estimator} and the output of Algorithm~\ref{alg:doubleDR} with default recommendation inputs and with $\check{d}_0 = 1$. More precisely, we apply Algorithm~\ref{alg:RPEDR estimator} with the recommended values from the previous section, that is $L = 200$, $M=10p$, $d = \lceil \sqrt{p} \rceil$, $n_1 = \lceil 2n/3 \rceil$,  $Q = \frac{1}{2} Q_{\mathrm{C}}^{\otimes d} + \frac{1}{2} Q_{\mathrm{N}}^{\otimes d}$ for the projection distribution, and $\hat{g} = \hat{g}_{\mathrm{MARS}}$.  In both models, a second application of our algorithm leads to significantly improved performance. In Model 1a, the first application of Algorithm~\ref{alg:RPEDR estimator} typically identifies that the first two variables are important, but splits this signal across the first two singular vectors. The second application then effectively combines these separate signals into a  better estimator of the one-dimensional $A_0$.  A similar scenario occurs in Model~4, though the improvement is slightly less pronounced.  

In cases where the initial application of Algorithms~\ref{alg:RPEDR estimator} and~\ref{alg:DimensionEstimator} correctly identify the true dimension $d_0$, then there is no benefit to be gained from applying Algorithm~\ref{alg:doubleDR}. However, as long as $\check{d}_0 \geq d_0$, there is typically no disadvantage either -- see, for example the results for Model~3 presented in Figure~\ref{fig:sin theta distance boxplots}.


\section{Numerical simulations}
\label{sec:numerical}

In this section, we present the results of a large simulation study comparing the performance of our proposal with several existing methods for estimating the central mean subspace $\mathcal{S}(A_0)$.  We evaluate the performance of three slightly different versions of our proposed method. The first variant applies Algorithm~\ref{alg:RPEDR estimator} on its own, and the second variant is the double dimension reduction approach presented in Algorithm~\ref{alg:doubleDR}. Both variants are suitable when the true projection dimension $d_0$ is known. The third variant, which combines Algorithm ~\ref{alg:RPEDR estimator} and~\ref{alg:DimensionEstimator}, is designed for situations where $d_0$ is unknown. For comparison, we use the default implementations of several existing methods, including SIR \citep{SIR}, pHd \citep{li1992principal}, MAVE \citep{MAVE},  DR \citep{li2007directional}, gKDR \citep{fukumizu2014gradient} and drMARS \citep{drMARS}.  
The models used in our experiments include Models 1a, 2, 3 and 4, introduced earlier in the paper, as well as five additional models as follows: let $X := (X_1, \ldots, X_p)^T \sim N_p(0,I)$ and $\epsilon \sim N(0,1/4)$, with $X$ and $\epsilon$ independent. The additional models are specified as:
\begin{itemize}
    \item Model 5: $Y = \frac{X_1+X_2 + 5e^{-2(X_1+X_2+X_3)^2}}{2} + \epsilon$, 
    \item Model 6: $Y=5 X_1 X_2 X_3 + \epsilon$
    \item Model 7: $Y=4(X_1-X_2+X_3)\sin\bigl(\frac{\pi}{2}(X_1+X_2)\bigr) + \epsilon$
    \item Model 8: $Y=X_1(X_1+X_2+1) + \epsilon$
    \item Model 9: $Y=10 \cos{(6X_1)} + e^{X_2+1} + \epsilon$
\end{itemize}
Models 5, 6 and 7 were used by \citep{drMARS} in their simulation study, and Model 8 was used in \citep{SIR,MAVE,drMARS}. 

In our full simulation study, the results of which are presented in Section~\ref{sec:fullsims} in the Appendix, we vary $n \in \{50, 200, 500\}$ and $p \in \{20, 50, 100\}$ for each model. In the main text,  we focus on the performance when $n= 200$ and $p = 50$, and note that the results are broadly similar across the different values of $n$ and $p$ considered.

\subsection{Dimension \texorpdfstring{$d_0$}{d0} known}
\label{sec:simsd0known}
We begin with the slightly simpler scenario where the dimension $d_0$ of the central mean subspace $\mathcal{S}(A_0)$ is known.  In this case, our first proposal, denoted as \textbf{RPE} in Figure~\ref{fig:sin theta distance boxplots} below, sets $\hat{A}_0 = (U_1,\ldots, U_{d_0})$, where $U = (U_1,\ldots, U_p)$ is the output of Algorithm~\ref{alg:RPEDR estimator} with the default inputs recommended in Section~\ref{sec:practicalconsiderations}. In particular, we set $d = \lceil p^{1/2} \rceil$, $L = 200$, $M = 10p$, $n_1 = \lceil 2n/3 \rceil$, $Q = \frac{1}{2} Q_{\mathrm{C}}^{\otimes d} + \frac{1}{2} Q_{\mathrm{N}}^{\otimes d}$, and $\hat{g} = \hat{g}_{\mathrm{MARS}}$.  We also present the performance of Algorithm~\ref{alg:doubleDR} with $\check{d}_0 = d_0$, which is denoted \textbf{RPE2} in Figure~\ref{fig:sin theta distance boxplots}.  

For the competing methods, we use the \texttt{R} packages \texttt{dr} \citep{weisberg2002dr} for SIR and pHd, \texttt{MAVE} \citep{yingcun2021mave} for MAVE, and the relevant code available via gitHub for DR (\url{https://github.com/JSongLab/SDR_HC}), gKDR and drMARS (\url{https://github.com/liuyu-star/drMARS}). In each case, we use the corresponding default recommendations to select appropriate tuning parameters.

\begin{figure}[!ht]
    \centering
    \includegraphics[width=0.32\textwidth]{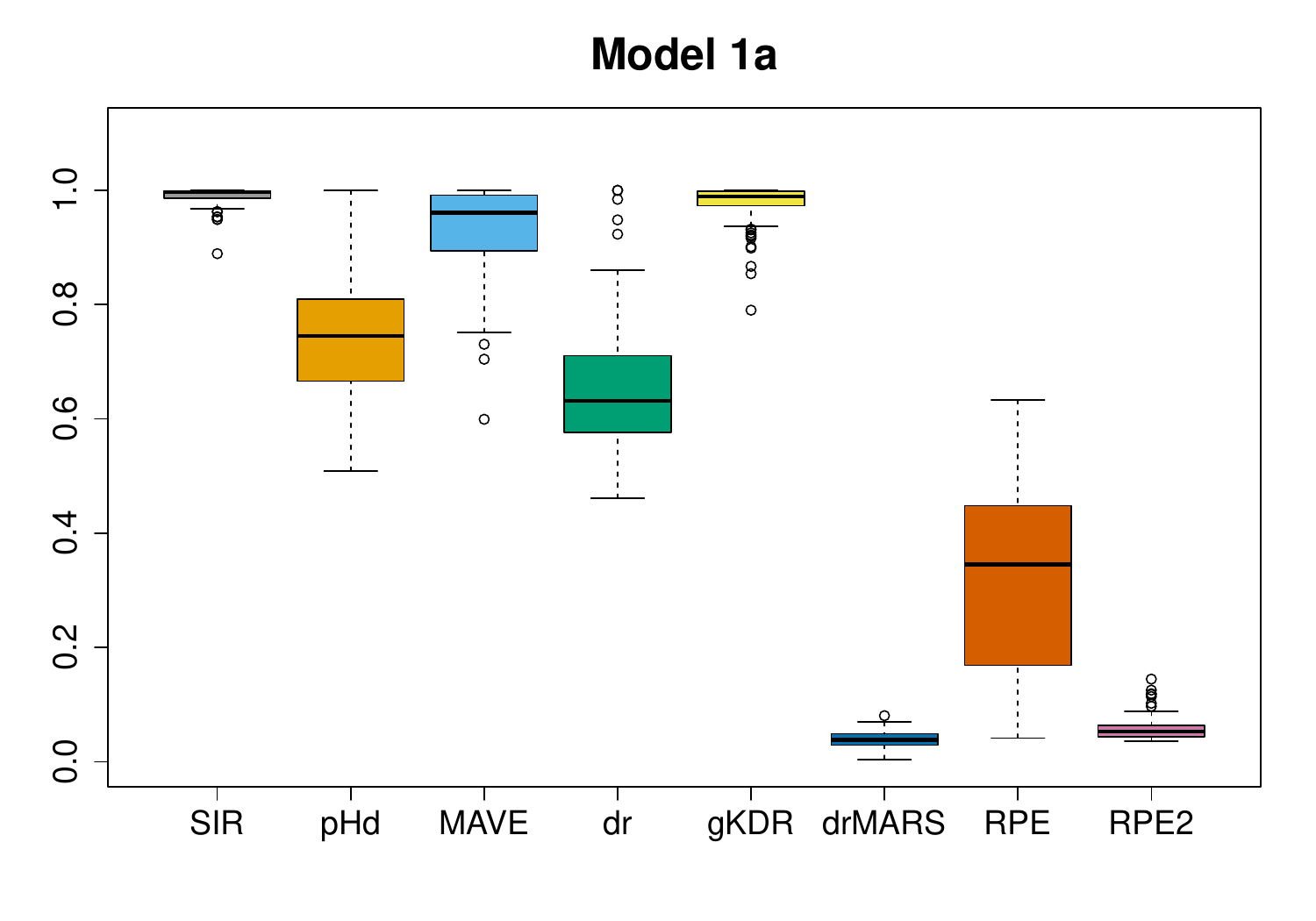}
    \includegraphics[width=0.32\textwidth]{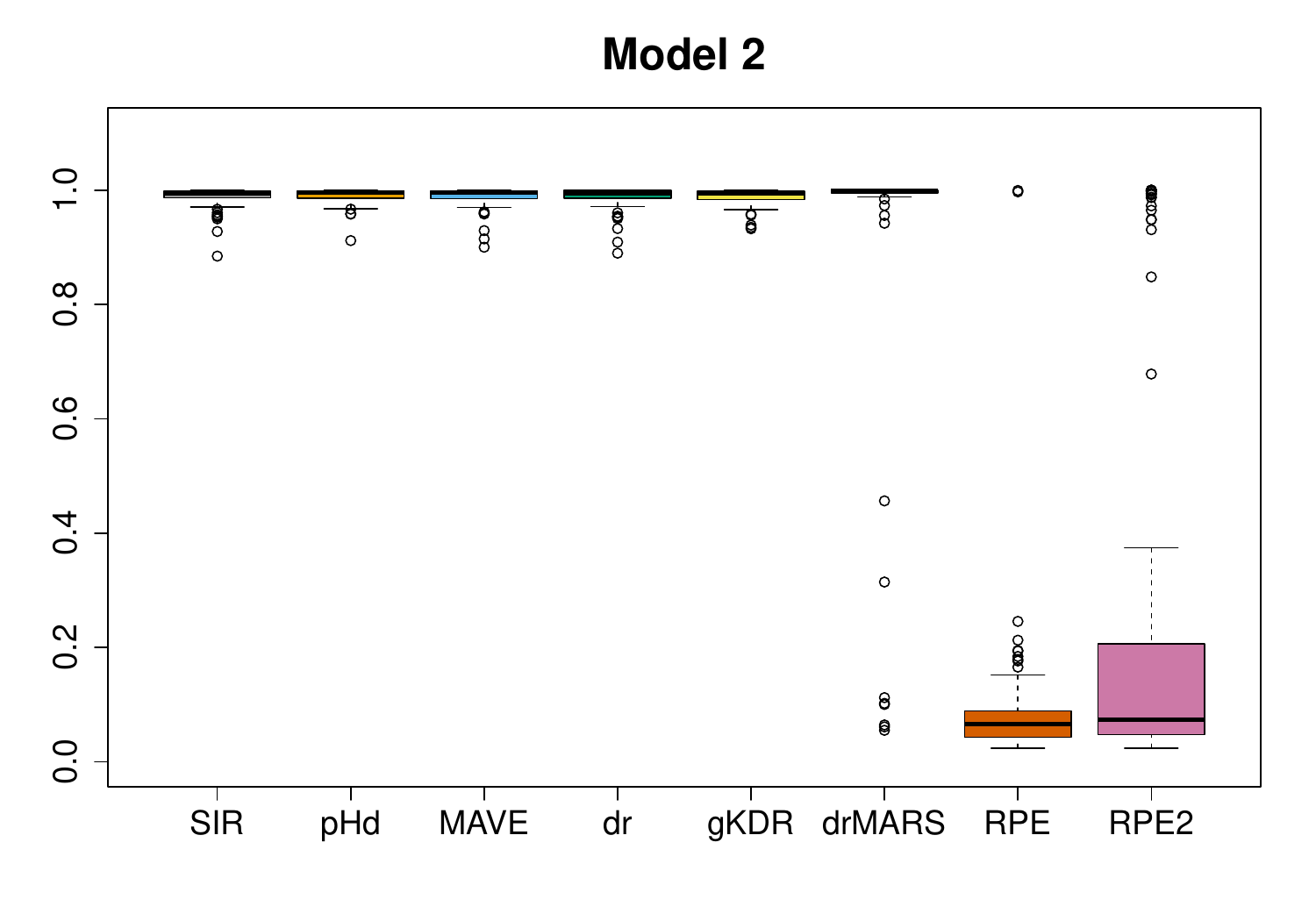}
    \includegraphics[width=0.32\textwidth]{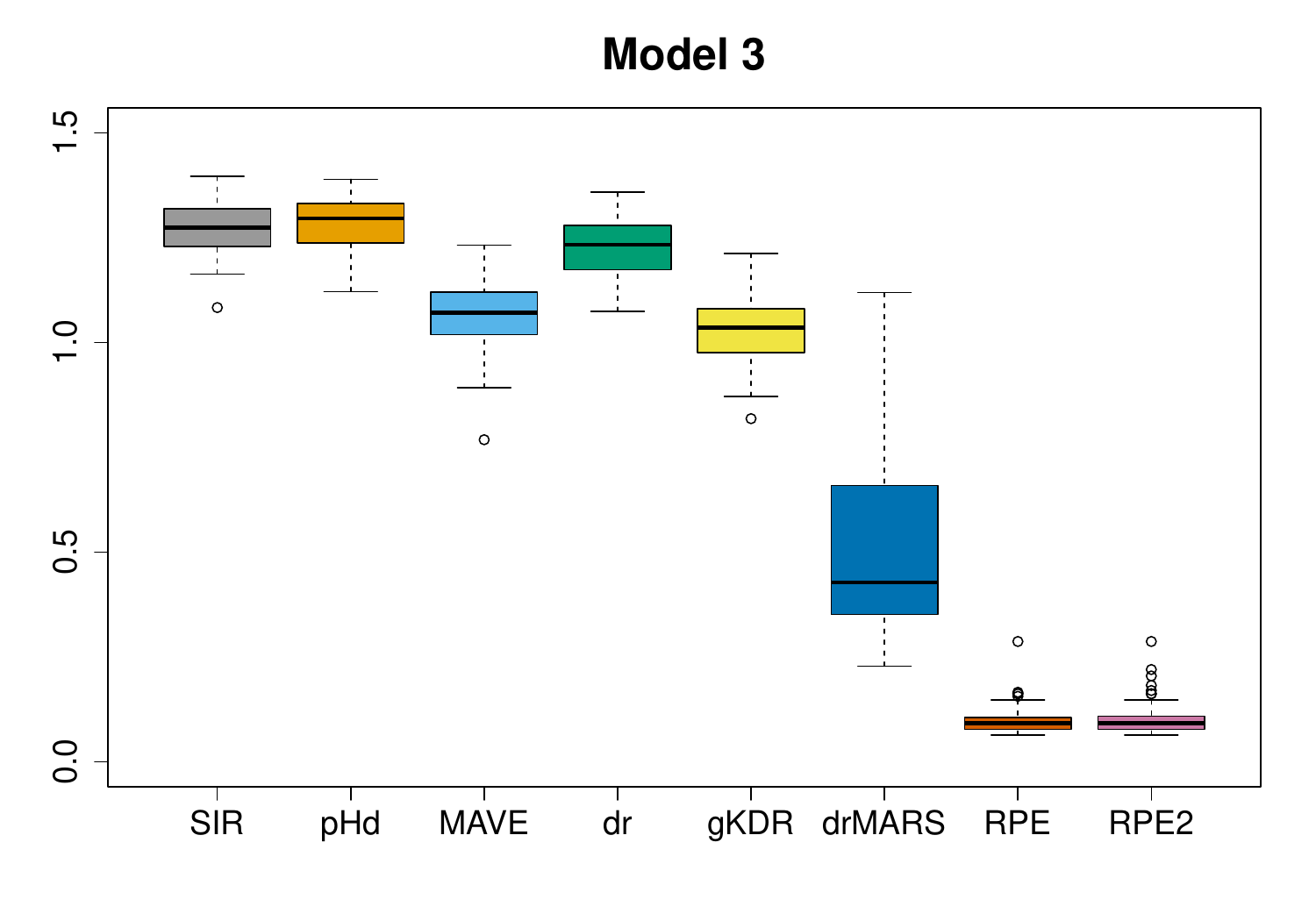}

    \includegraphics[width=0.32\textwidth]{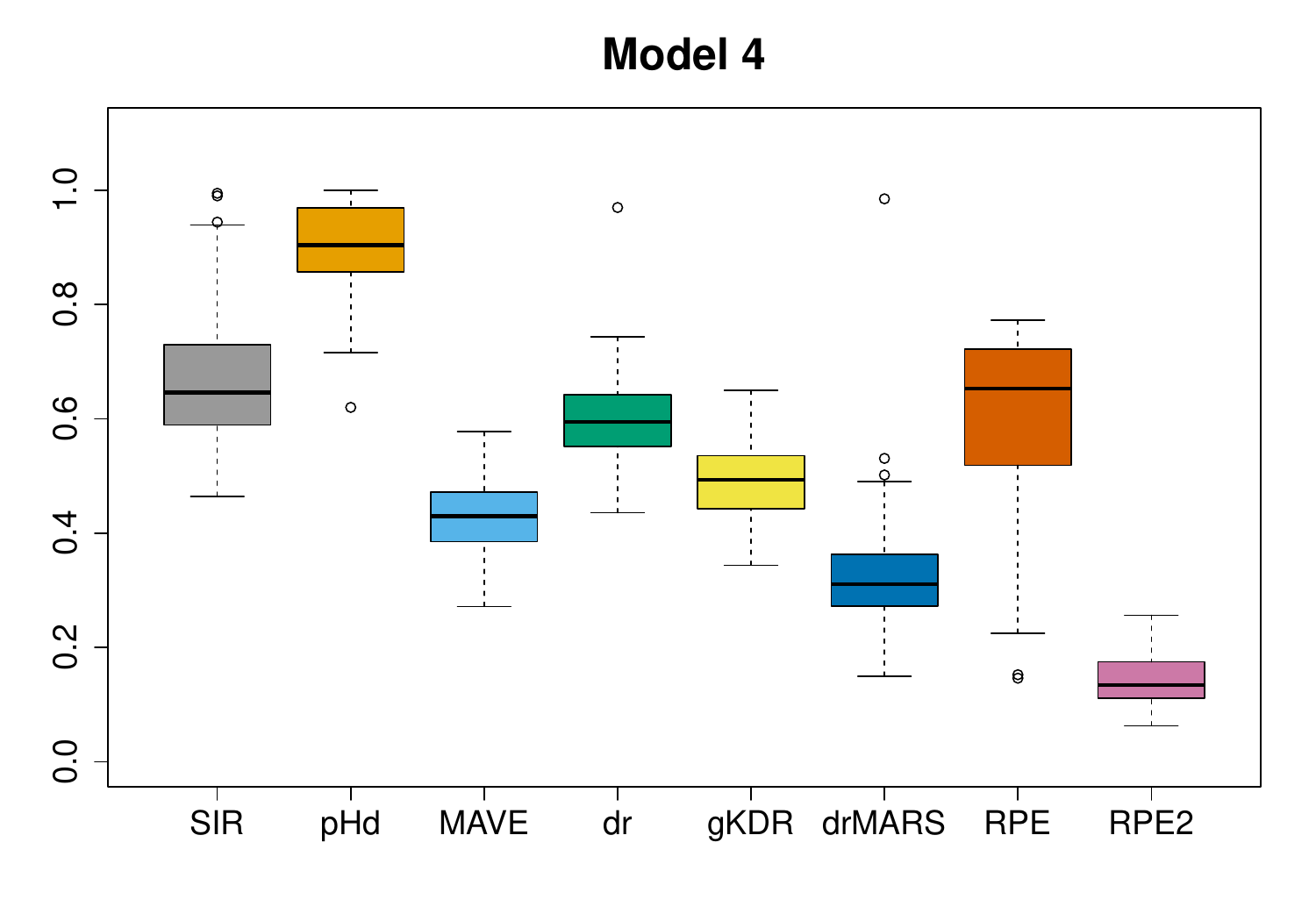}
    \includegraphics[width=0.32\textwidth]{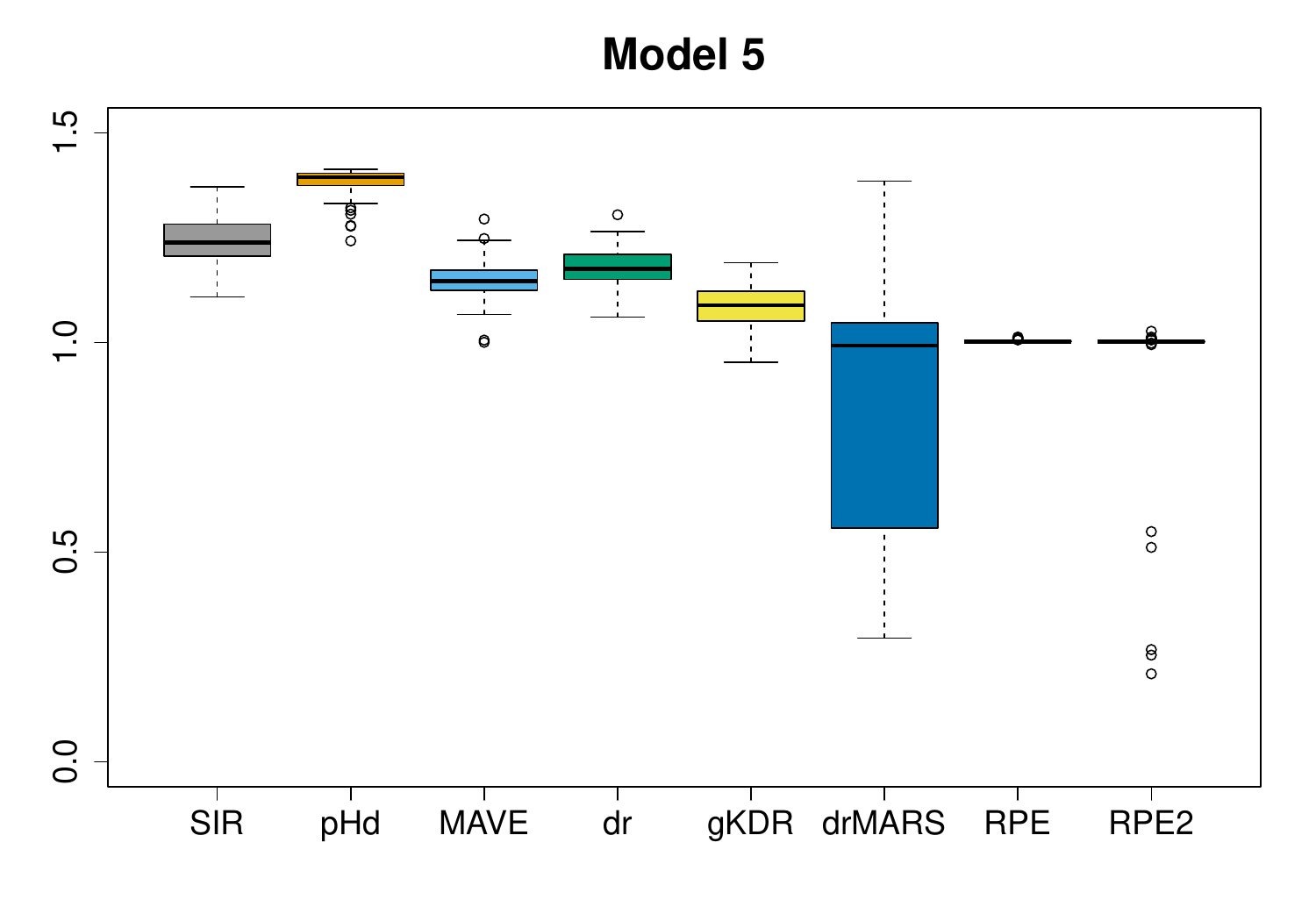}
    \includegraphics[width=0.32\textwidth]{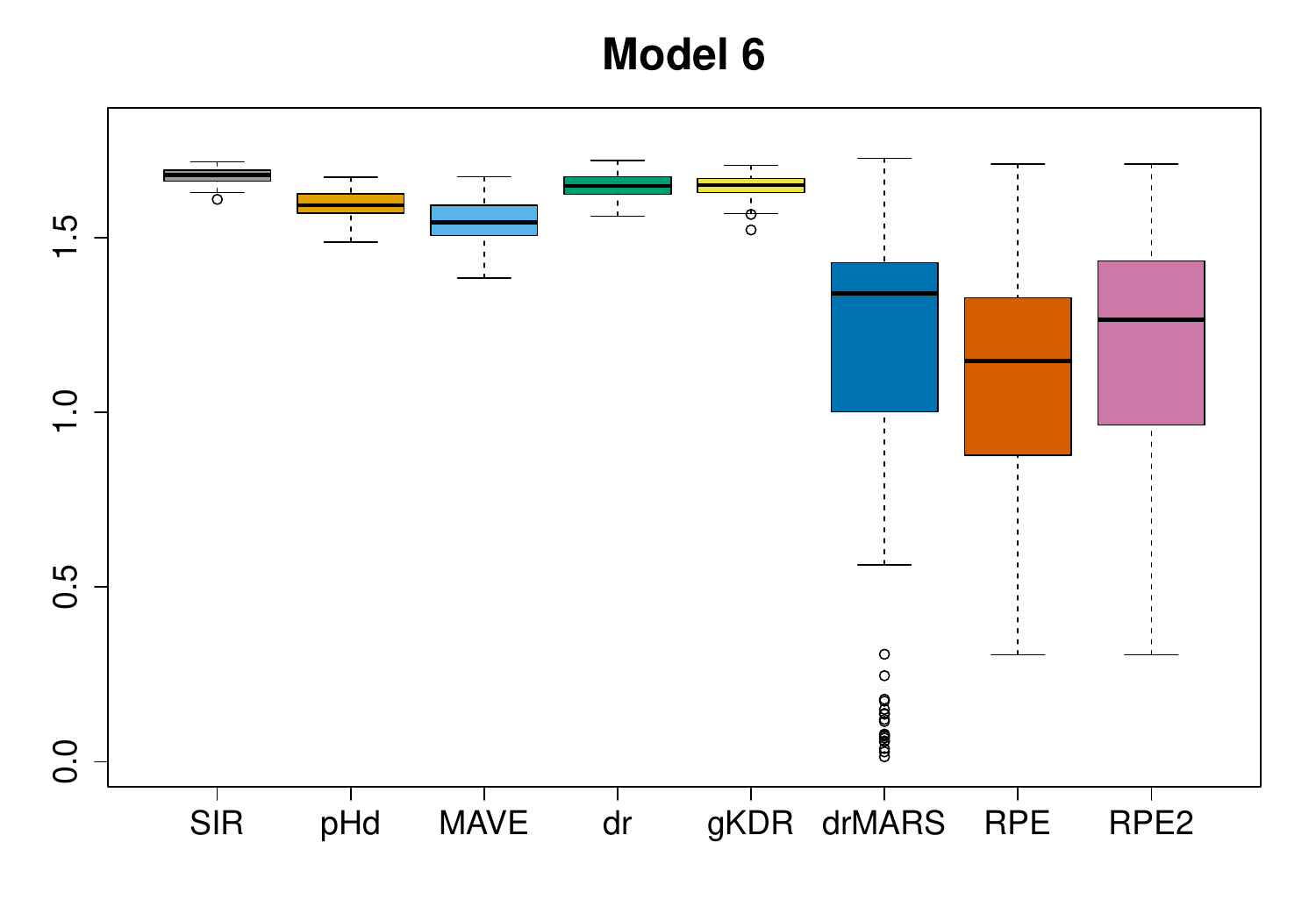}

    \includegraphics[width=0.32\textwidth]{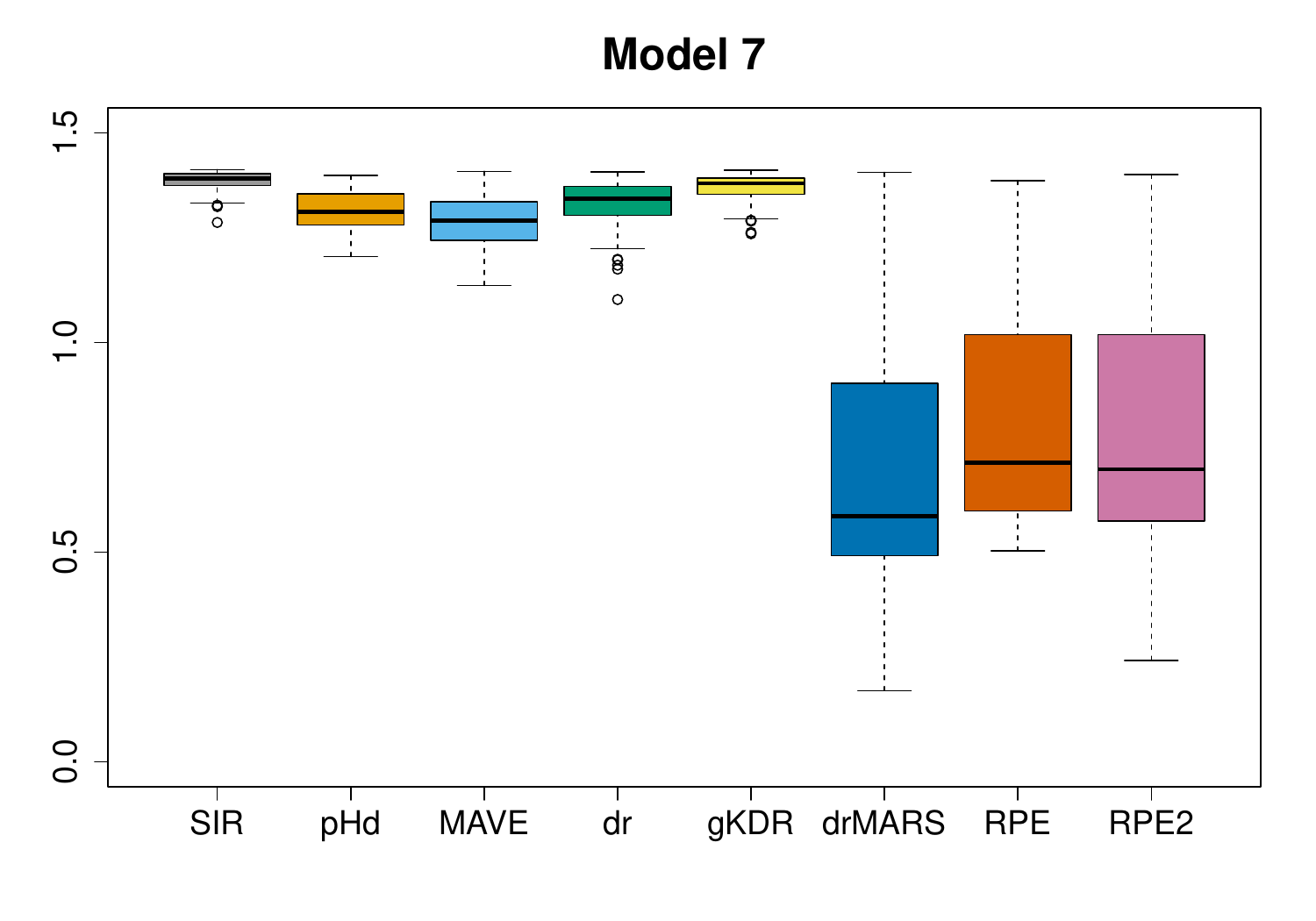}
    \includegraphics[width=0.32\textwidth]{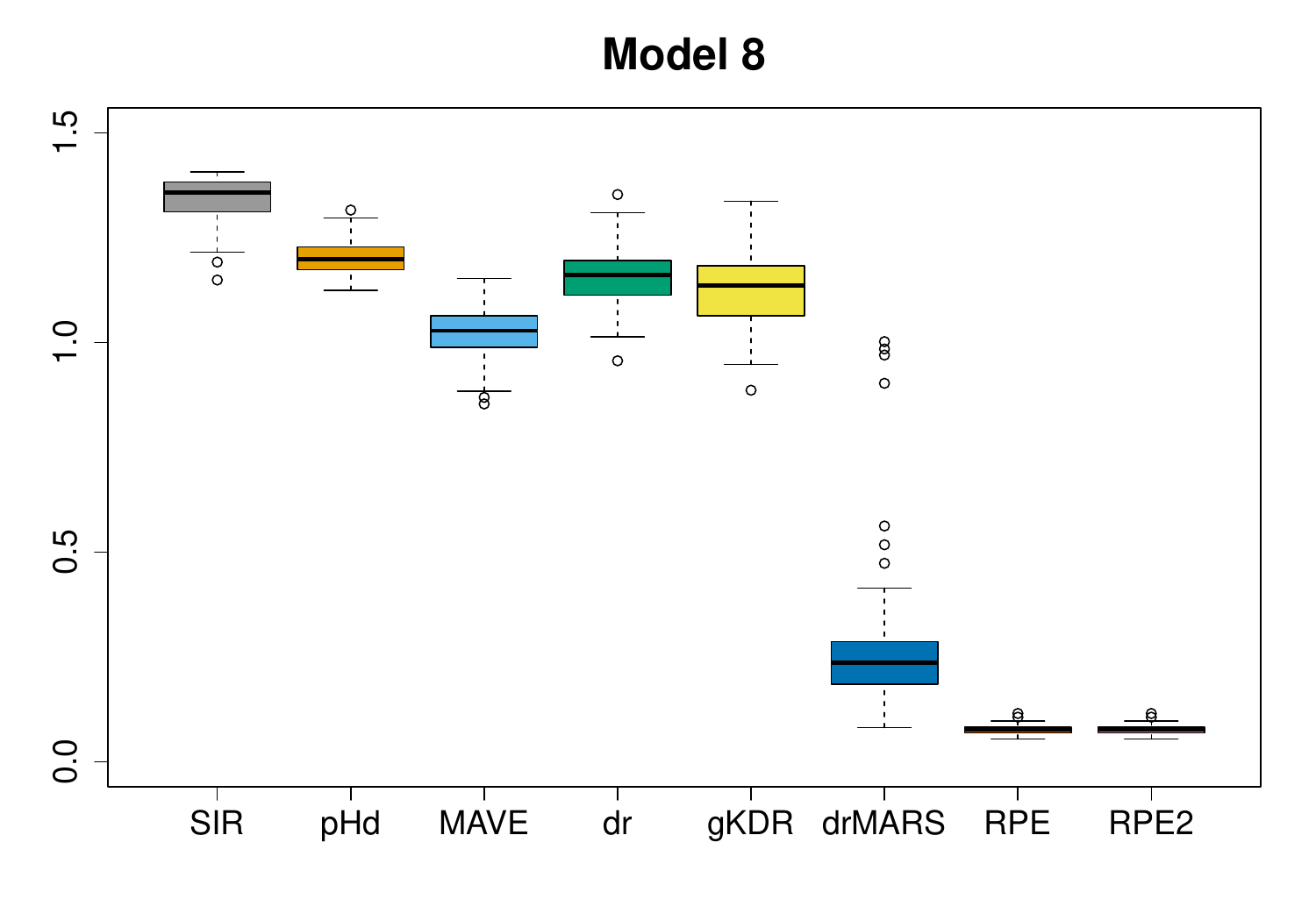}
    \includegraphics[width=0.32\textwidth]{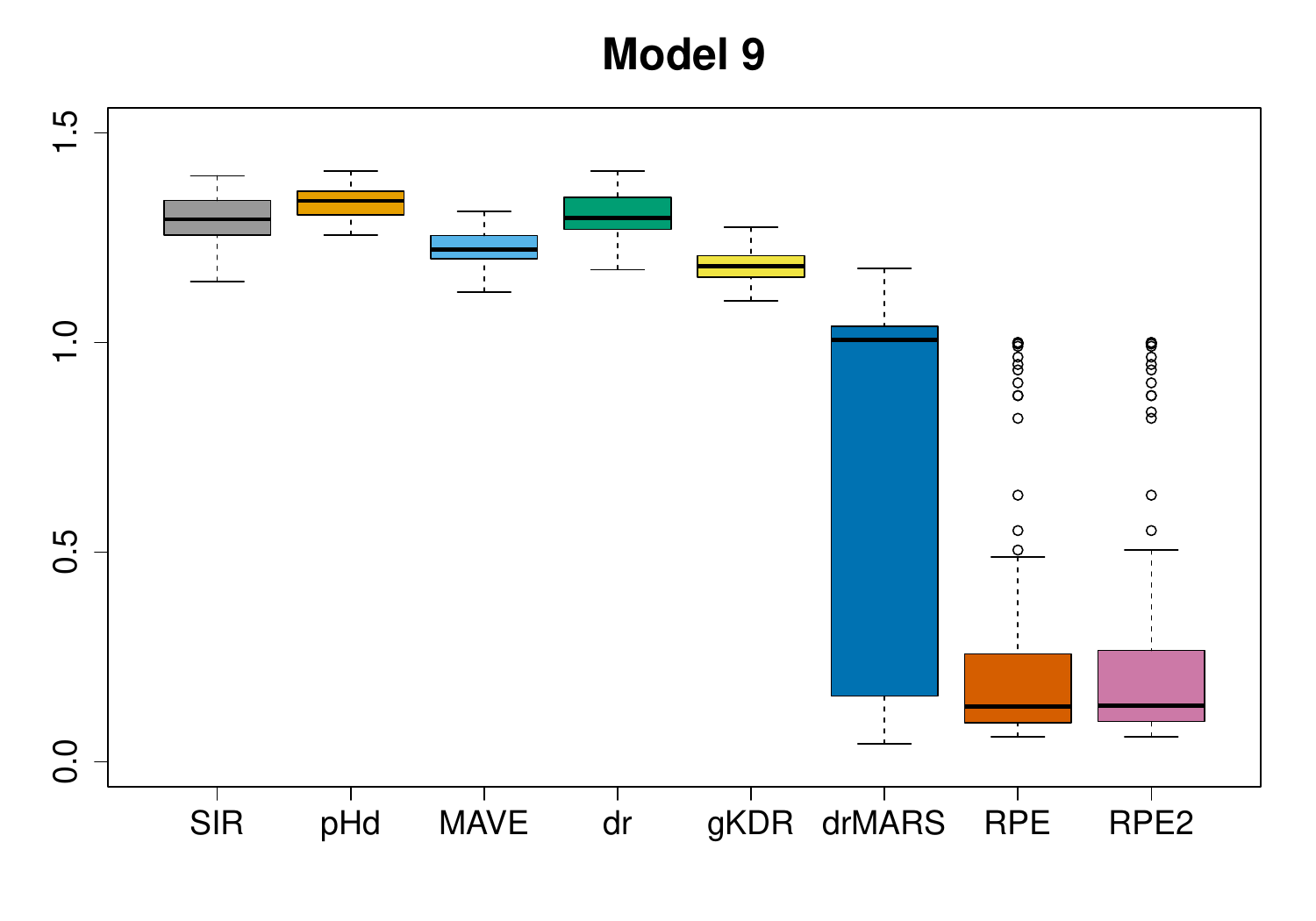}
    
    \caption{{\small Boxplots of sin theta distance between the subspaces spanned by the estimated projection $\hat{A}_0$ and true projection $A_0$ for our random projection based methods and the six competing methods when $d_0$ is known.  We present the results for models 1a-9 over 100 repeats of the experiment, with $n=200$ and $p=50$.}\label{fig:sin theta distance boxplots}}
\end{figure}

Figure~\ref{fig:sin theta distance boxplots} presents boxplots of the sin-theta distance $d_F(\mathcal{S}(\hat{A}_0), \mathcal{S}(A_0))$ across 100 repeats of the experiment for the nine models. We observe that both our RPE and RPE2 approaches are highly effective.  In Models 1a and 4, RPE2 demonstrates a clear improvement over RPE, as we've also highlighted in Section~\ref{sec:doubleRPEDR}. For the other models, a single application of our Algorithm~\ref{alg:RPEDR estimator} performs very well.  The drMARS approach is also competitive in Models 1a, 5, 6 and 7.  The other competitors suffer prohibitively due to the curse of dimensionality and are typically ineffective in the $50$-dimensional examples shown here\footnote{The performance of these approaches is better in lower dimensional settings, though our random projection based methods still offer an improvement when $p=20$ -- see the full results in Section~\ref{sec:fullsims}. }.  There are a few cases where the performance our methods (and in fact all approaches considered here) does not achieve a sin-theta distance close to zero. In Models 5 and 7, our algorithm tends to find the space spanned by $e_1$ and $e_2$, but misses the weak signal in the $e_3$ direction. Model 6 is particularly challenging, both because in order to obtain a non-trivial prediction of the response $Y$, we need some knowledge of all three covariates $X_1, X_2$ and $X_3$, (indeed, for example, $\mathbb{E}(Y|(X_1,X_2)) = 0$) and, even given an oracle projection into the subspace spanned by $e_1, e_2$ and $e_3$, the base regression problem remains difficult. We observe better performance in these models when $p=20$ and $n = 500$ -- see Section~\ref{sec:fullsims}. 

\subsection{Dimension \texorpdfstring{$d_0$}{d0} unknown}
\label{sec:numericald0unknown}
We now turn to the more challenging, and arguably more realistic scenario where $d_0$ is unknown.  In this subsection, our primary goal is to evaluate the performance of our Algorithm~\ref{alg:RPEDR estimator} when the dimension of the final projection $\hat{d}_0$ is chosen using Algorithm~\ref{alg:DimensionEstimator}.  In this case, we focus solely on the performance of a single application of our random projection ensemble algorithm. The double dimension reduction technique in Algorithm~\ref{alg:doubleDR} is not considered here, as it requires a desired projection dimension to be prespecified.   For comparison, we also include the results from the competing methods, where the dimension is chosen using the corresponding recommended default approaches.  Specifically, these are the \emph{marginal dimension test} from the \texttt{R} package \texttt{dr} for SIR and pHd; and \emph{generalised cross-validation}, which is available as part of the \texttt{mave} package (for MAVE), and the corresponding gitHub links given in the previous subsection (for gKDR  and drMARS\footnote{for these two methods, we set the maximum projection dimension to be $\lceil \sqrt{p}\rceil$}). The DR method is only applicable with a prespecified $d_0$ and hence it is excluded from this subsection.  

In our experiments below, we will reuse the Models 1a through 9 described earlier. For our proposal, labelled \textbf{RPE} in the boxplots below, follows the recommendations from Section~\ref{sec:practicalconsiderations} for Algorithm~\ref{alg:RPEDR estimator}, with $\hat{d}_0$ selected via Algorithm~2 using $R = 10000$. Competing methods utilize their respective default values for any tuning parameters.

\begin{figure}[!ht]
    \centering
    \includegraphics[width=0.3\textwidth]{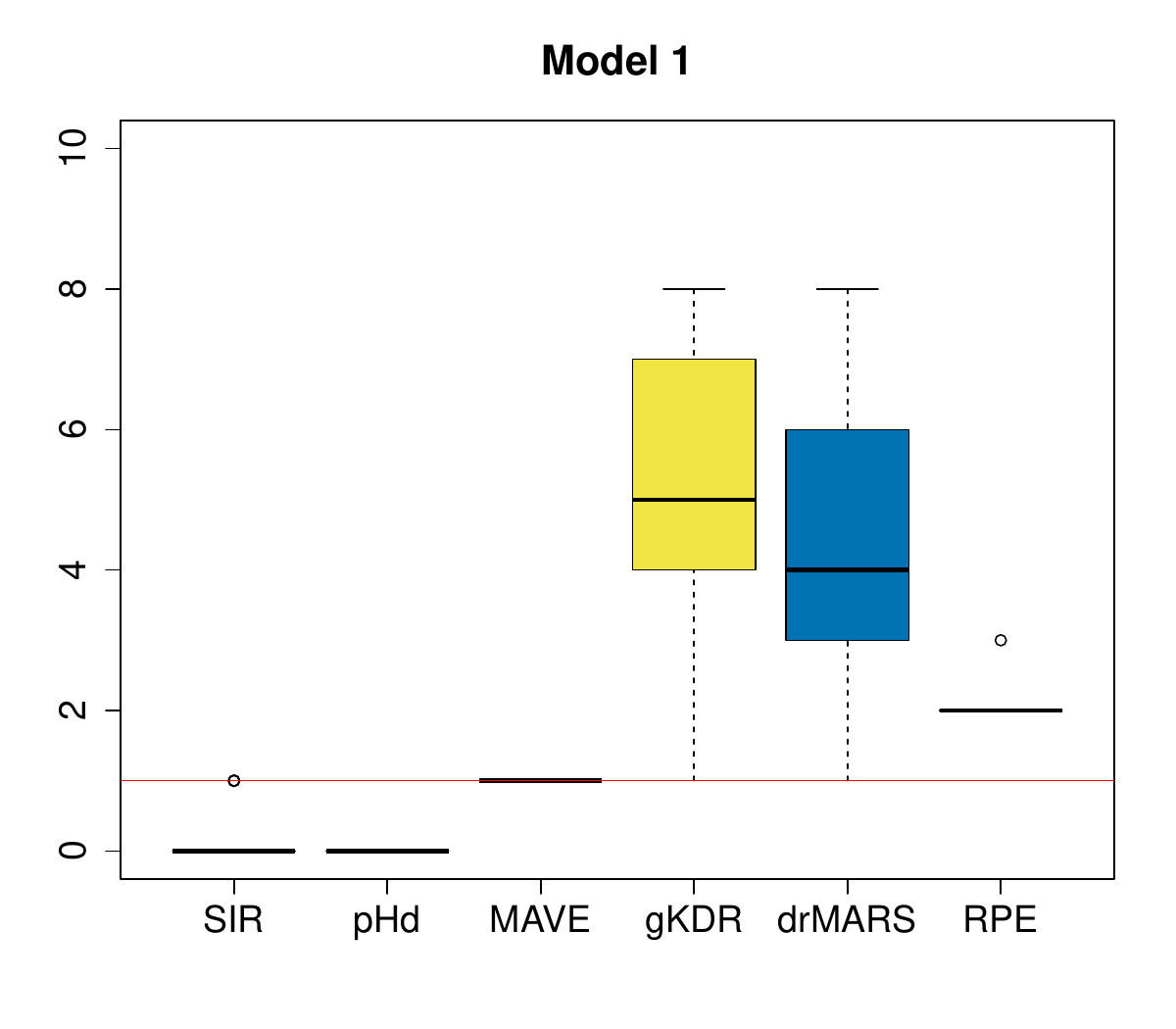}
    \includegraphics[width=0.3\textwidth]{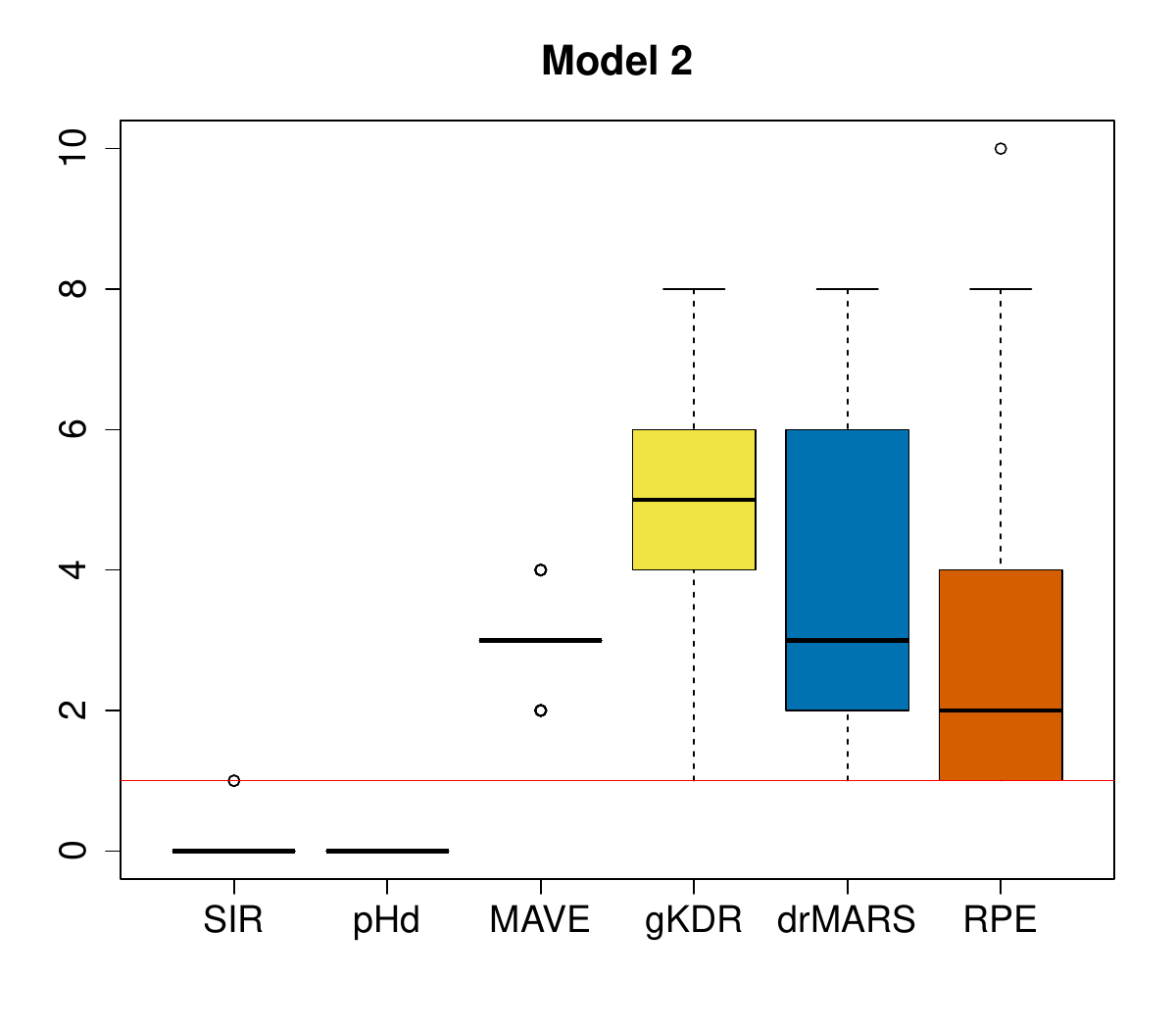}
    \includegraphics[width=0.3\textwidth]{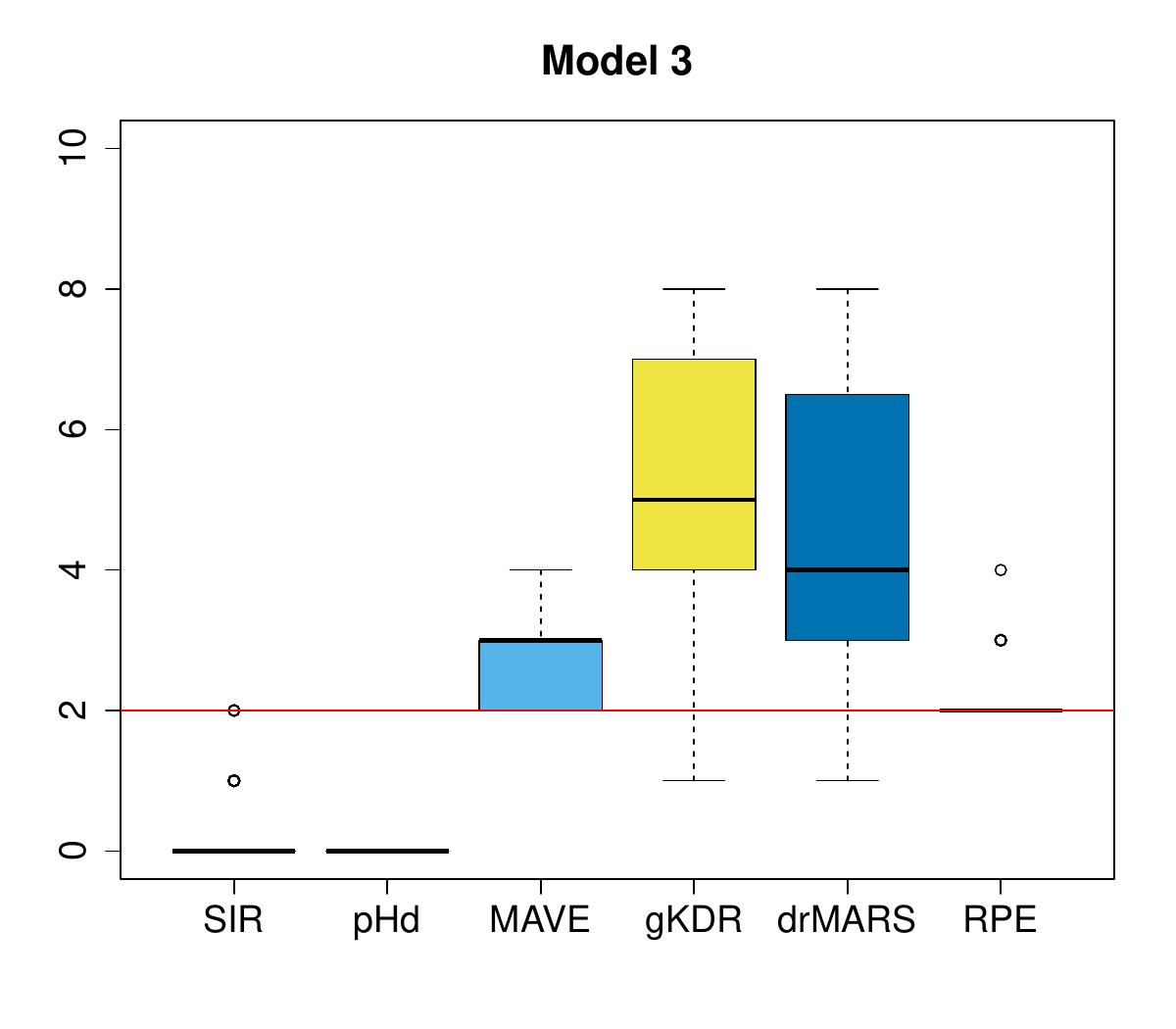}

    \includegraphics[width=0.3\textwidth]{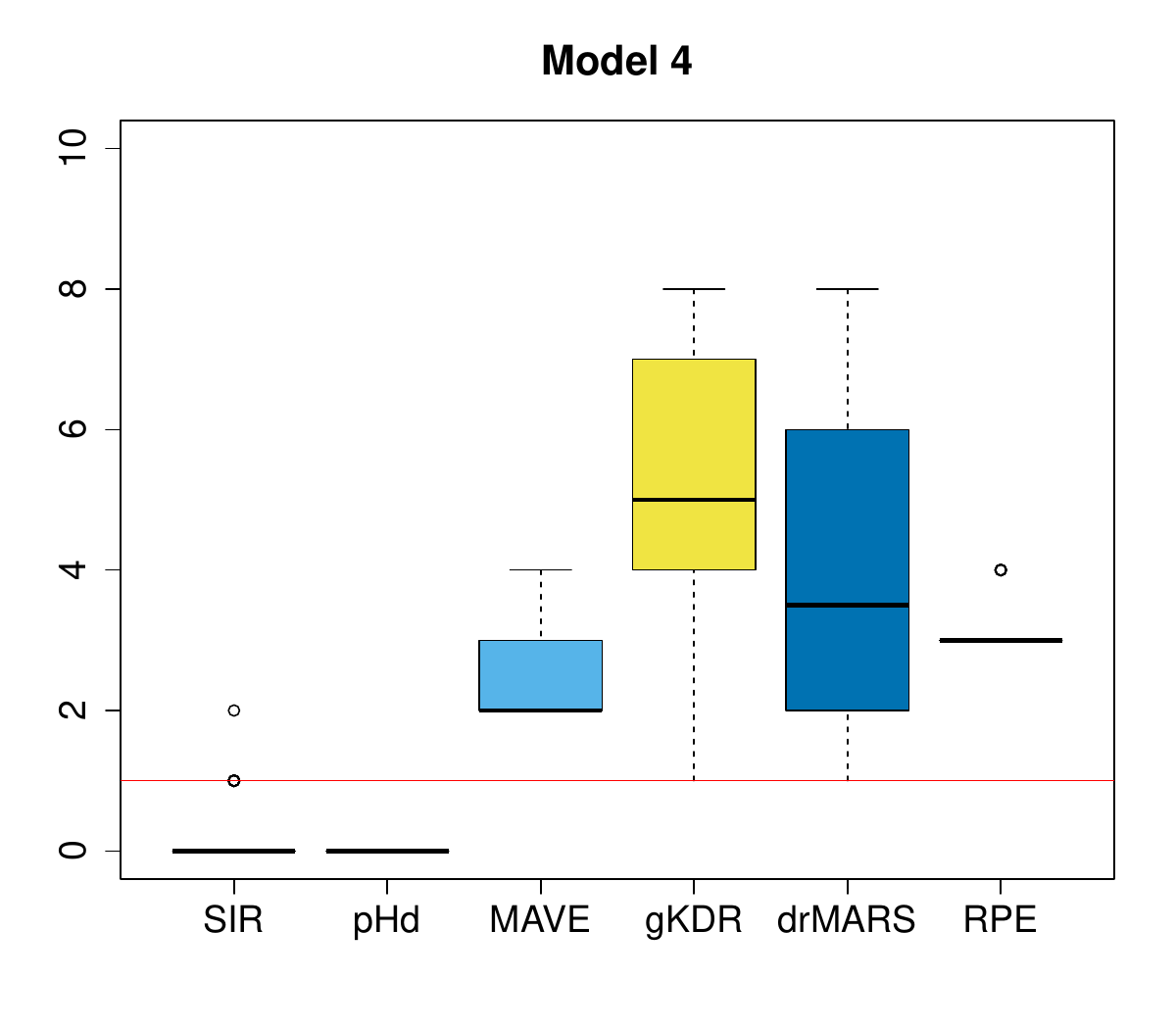}
    \includegraphics[width=0.3\textwidth]{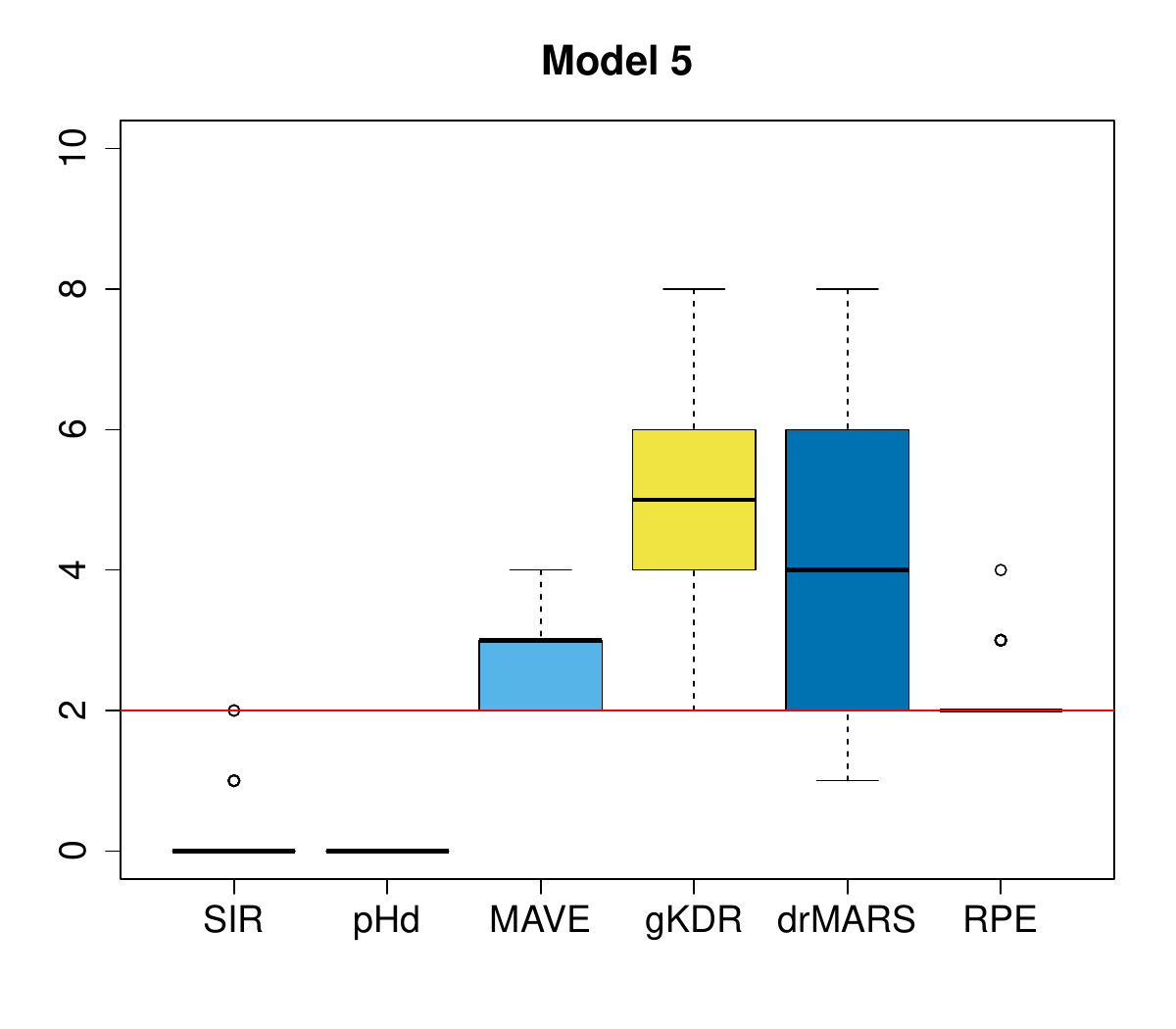}
    \includegraphics[width=0.3\textwidth]{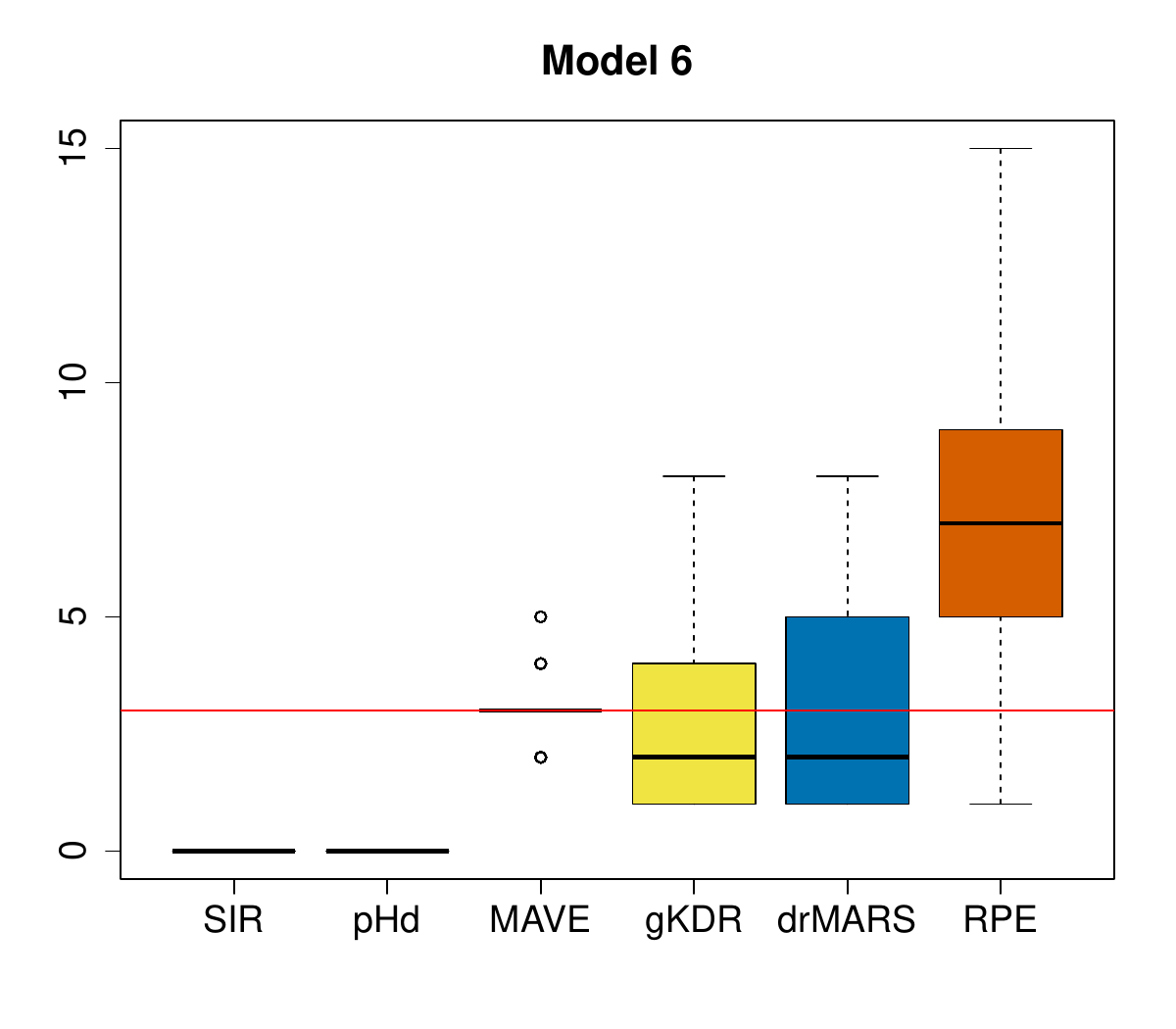}

    \includegraphics[width=0.3\textwidth]{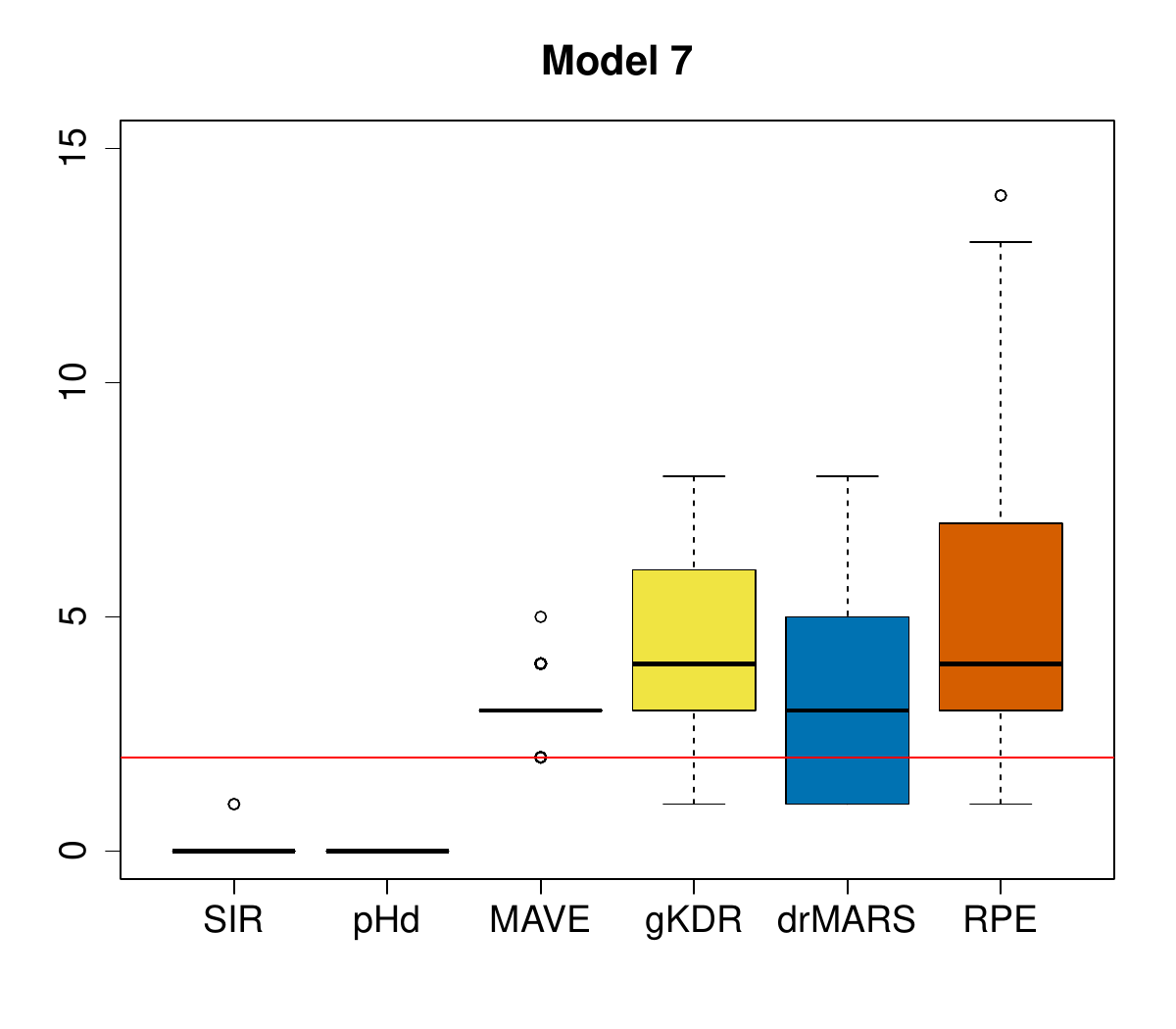}
    \includegraphics[width=0.3\textwidth]{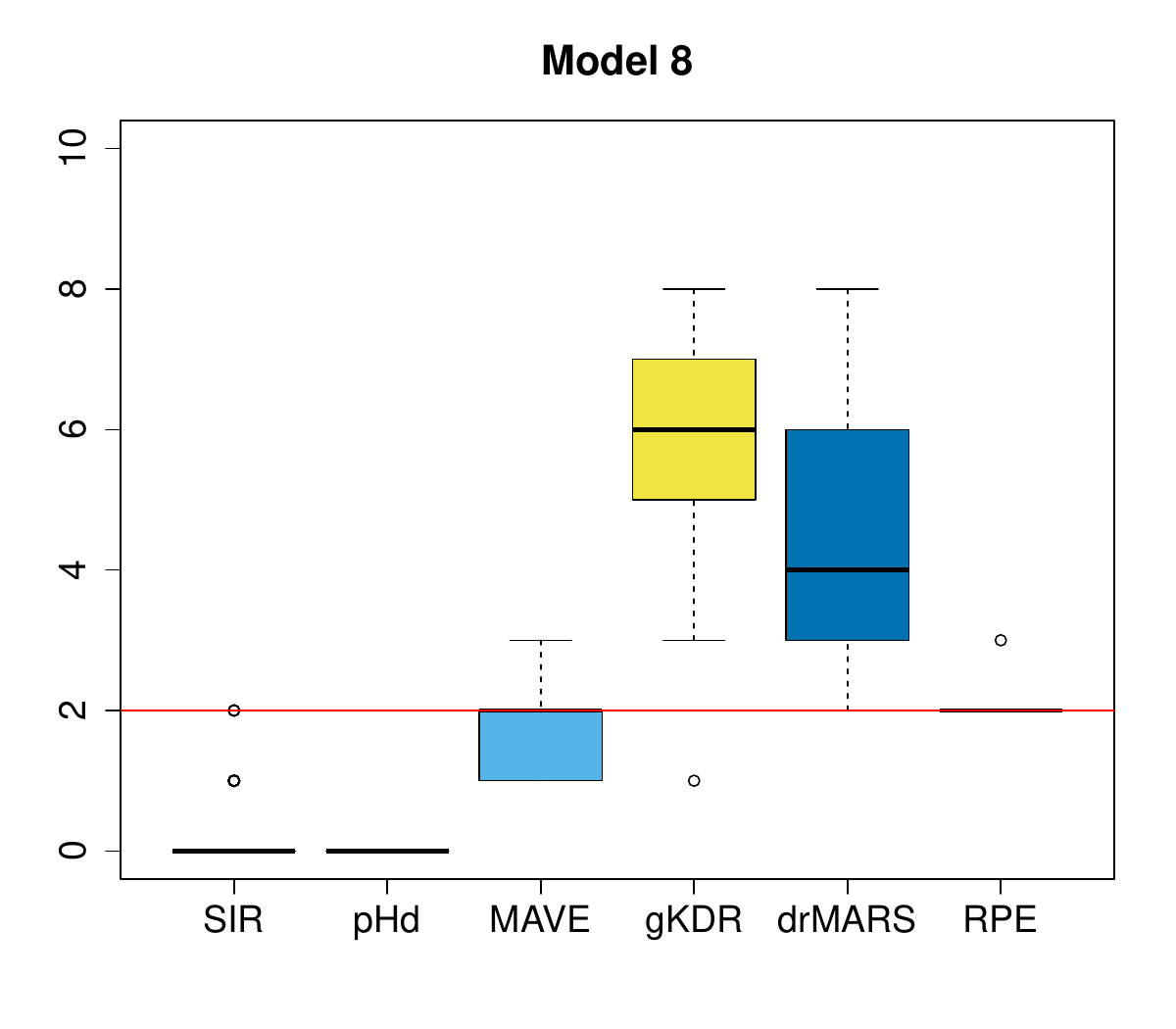}
    \includegraphics[width=0.3\textwidth]{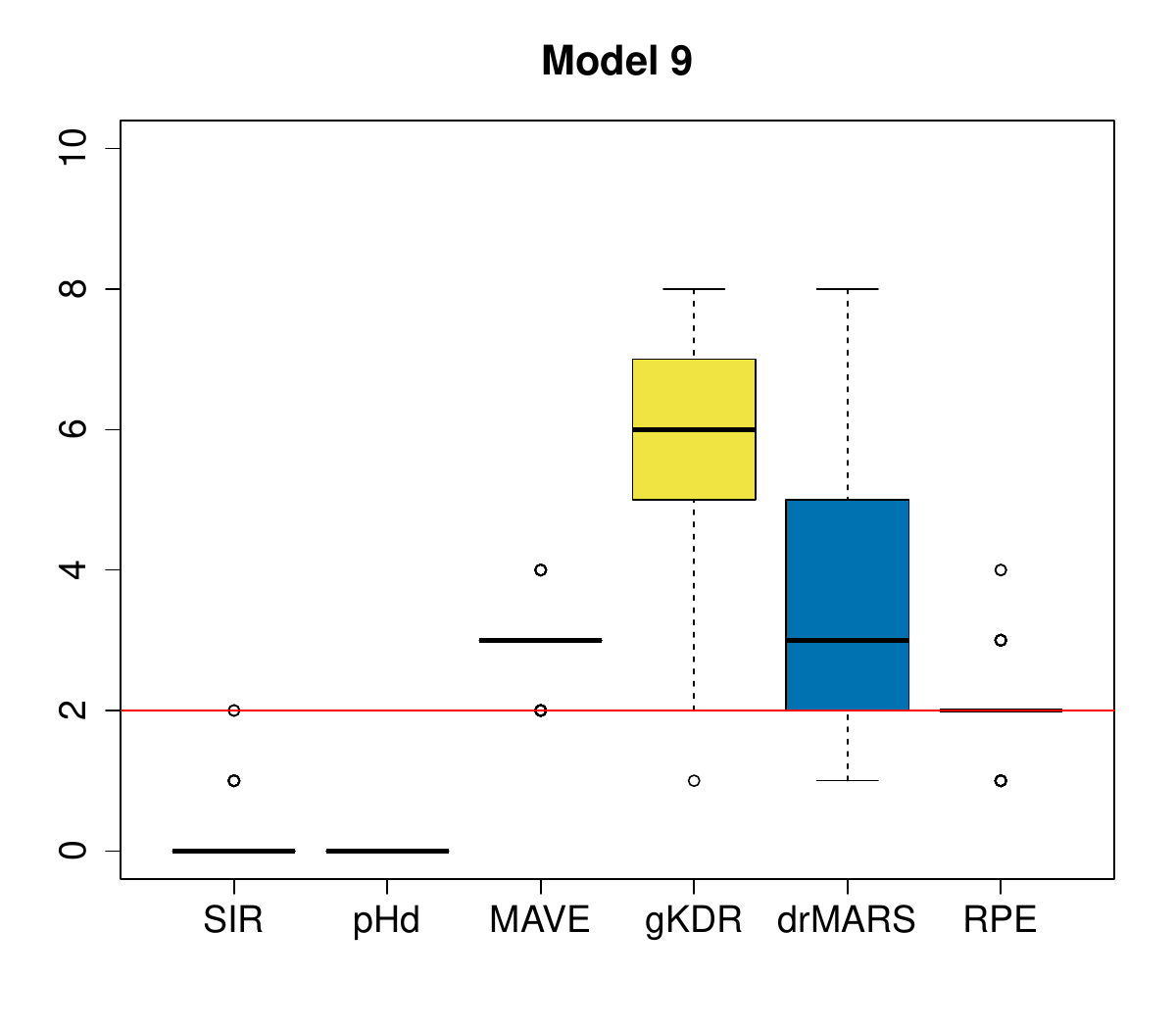}
    
    \caption{{\small Boxplots of $\hat{d}_0$ from Algorithm~\ref{alg:DimensionEstimator} alongside the competing methods over 100 repeats of the simulation for models 1a-9. We present the results for $n = 200$ and $p=50$.  The red horizontal lines show the dimension $d_0$ of the true $A_0$.  }\label{fig:d0 estimation boxplots}}
\end{figure}

The results in Figure~\ref{fig:d0 estimation boxplots} show that, with the exception of SIR and pHd, most methods including ours tend to be somewhat liberal in choosing $\hat{d}_0$, in the sense that they typically take $\hat{d}_0$ to be larger than $d_0$. Conversely, SIR and pHd typically selects $\hat{d}_0 = 0$ in these examples, indicating no signal found -- as mentioned above, both methods appear to be suffering from the curse of dimensionality.

To further assess the performance when $d_0$ is unknown, we introduce two additional evaluation metrics, which measure the \emph{false positives} and \emph{false negatives}:
\begin{equation}
\label{eq:FPmetric}
    d_{\mathrm{FP}}(\hat{A}_0, A_0) := \bigl \| (I - A_0 A_0^T) \hat{A}_0 \bigr \|_F
\end{equation}
\begin{equation}
    d_{\mathrm{FN}}(\hat{A}_0, A_0) := \bigl \| (I - \hat{A}_0 \hat{A}_0^T)A_0 \bigr \|_F
\end{equation}
In both cases, $\hat{A}_0$ is a $p \times \hat{d}_0$ dimensional projection, where $\hat{d}_0$ need not be equal to $d_0$.  The metric $d_{\mathrm{FP}}$ measures \emph{false positives}, in the sense that it quantifies the extent to which $\hat{A}_0$ contains regions of the ambient $p$ dimensional space that are orthogonal to $A_0$ In particular, if $\mathcal{S}(\hat{A}_0) \subseteq \mathcal{S}(A_0)$, then $d_{\mathrm{FP}}(\hat{A}_0, A_0) = 0$.  In contrast, $d_{\mathrm{FN}}$ measures \emph{false negatives} and quantifies the amount of the space spanned by $A_0$ that is \emph{missed} by the projection $\hat{A}_0$. Indeed, we have $d_{\mathrm{FN}}(\hat{A}_0, A_0) = 0$ when $\mathcal{S}(A_0) \subseteq \mathcal{S}(\hat{A}_0)$.  Broadly speaking, increasing $\hat{d}_0$ will typically decrease $d_{\mathrm{FN}}$ at the expense of increasing $d_{\mathrm{FP}}$. Only when $\hat{d}0 = d_0$ can both $d{\mathrm{FP}}$ and $d_{\mathrm{FN}}$ be equal to (or close to) zero.  

\begin{figure}[!ht]
    \centering
    \includegraphics[width=0.3\textwidth]{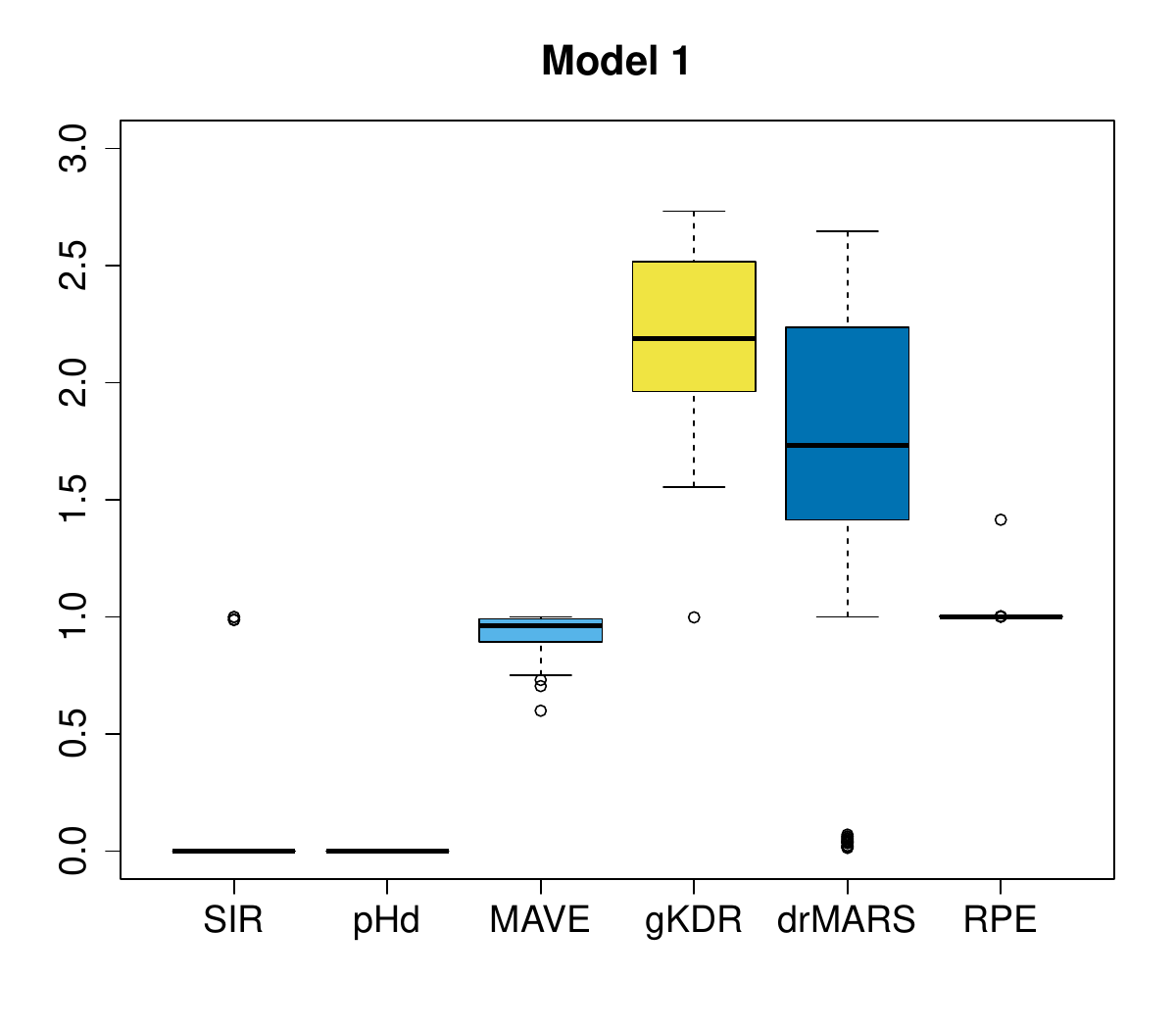}
    \includegraphics[width=0.3\textwidth]{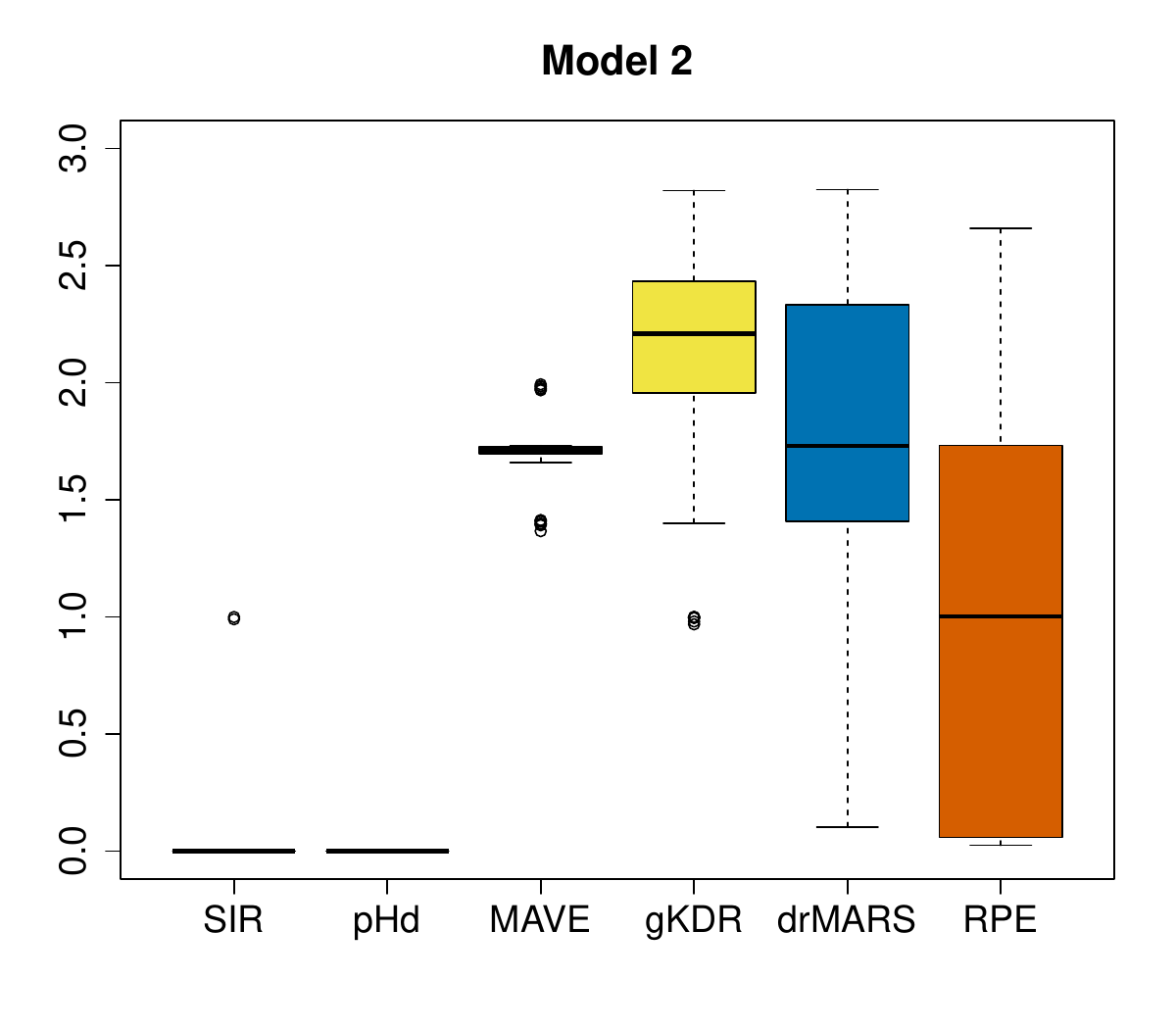}
    \includegraphics[width=0.3\textwidth]{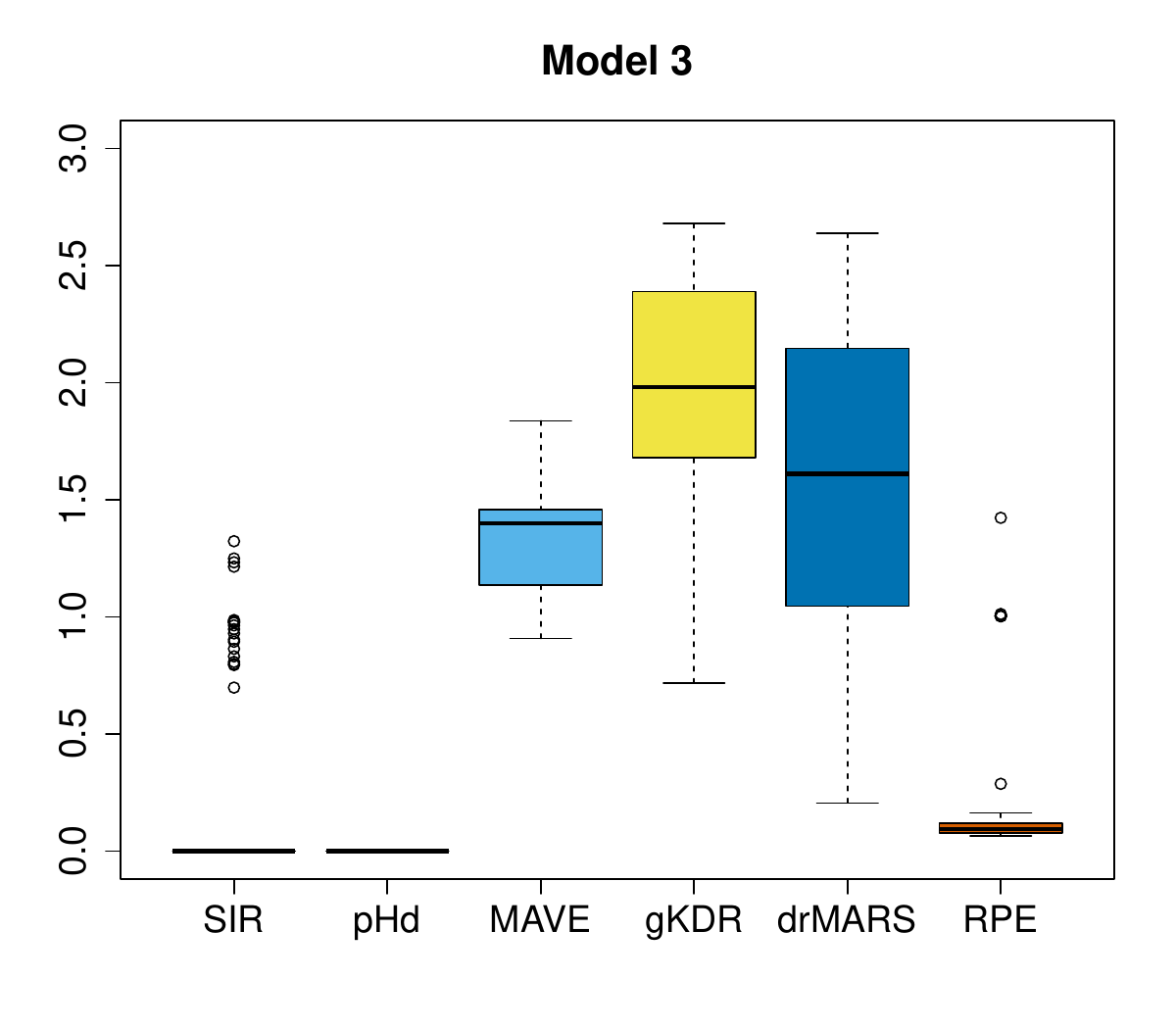}

    \includegraphics[width=0.3\textwidth]{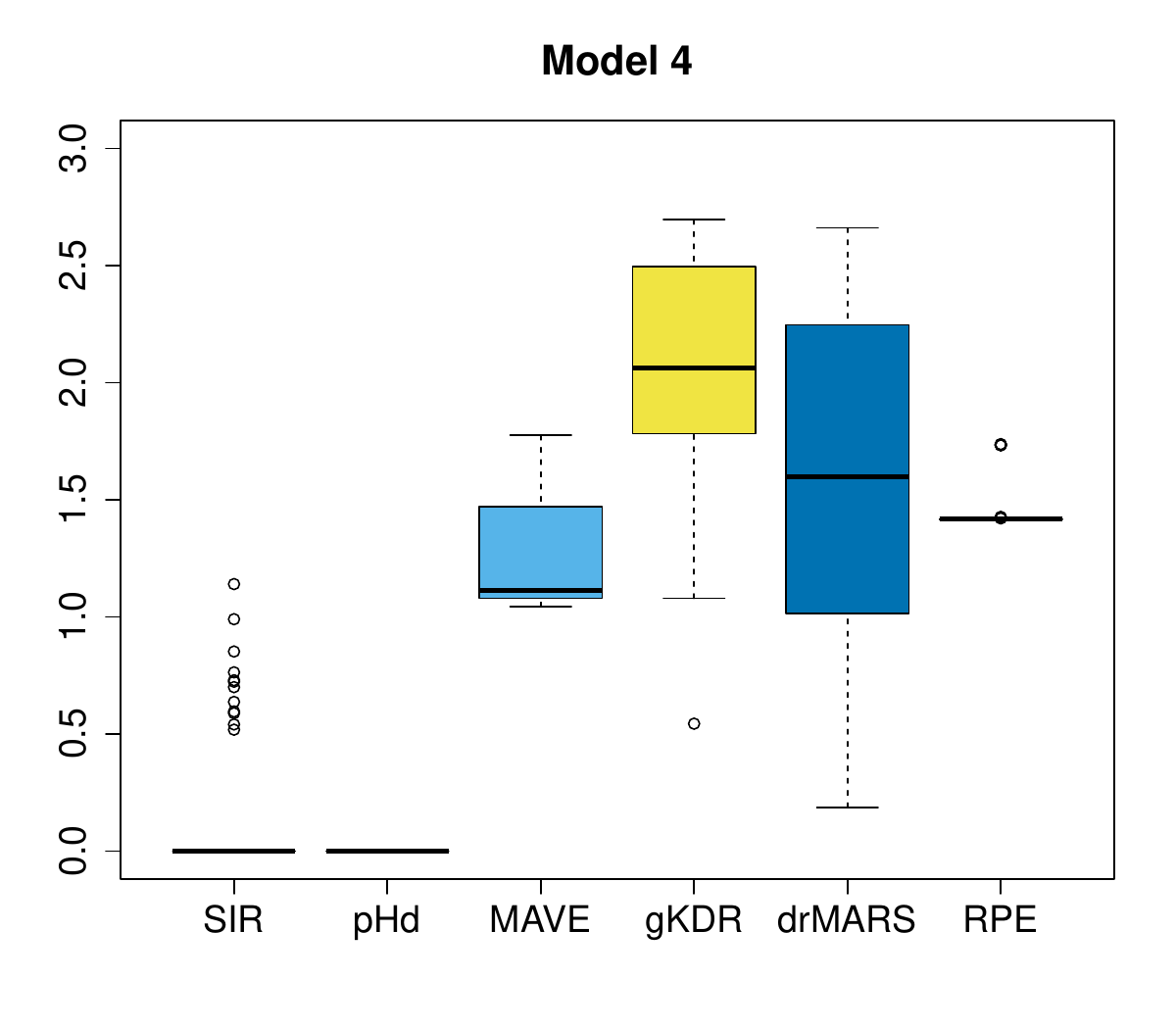}
    \includegraphics[width=0.3\textwidth]{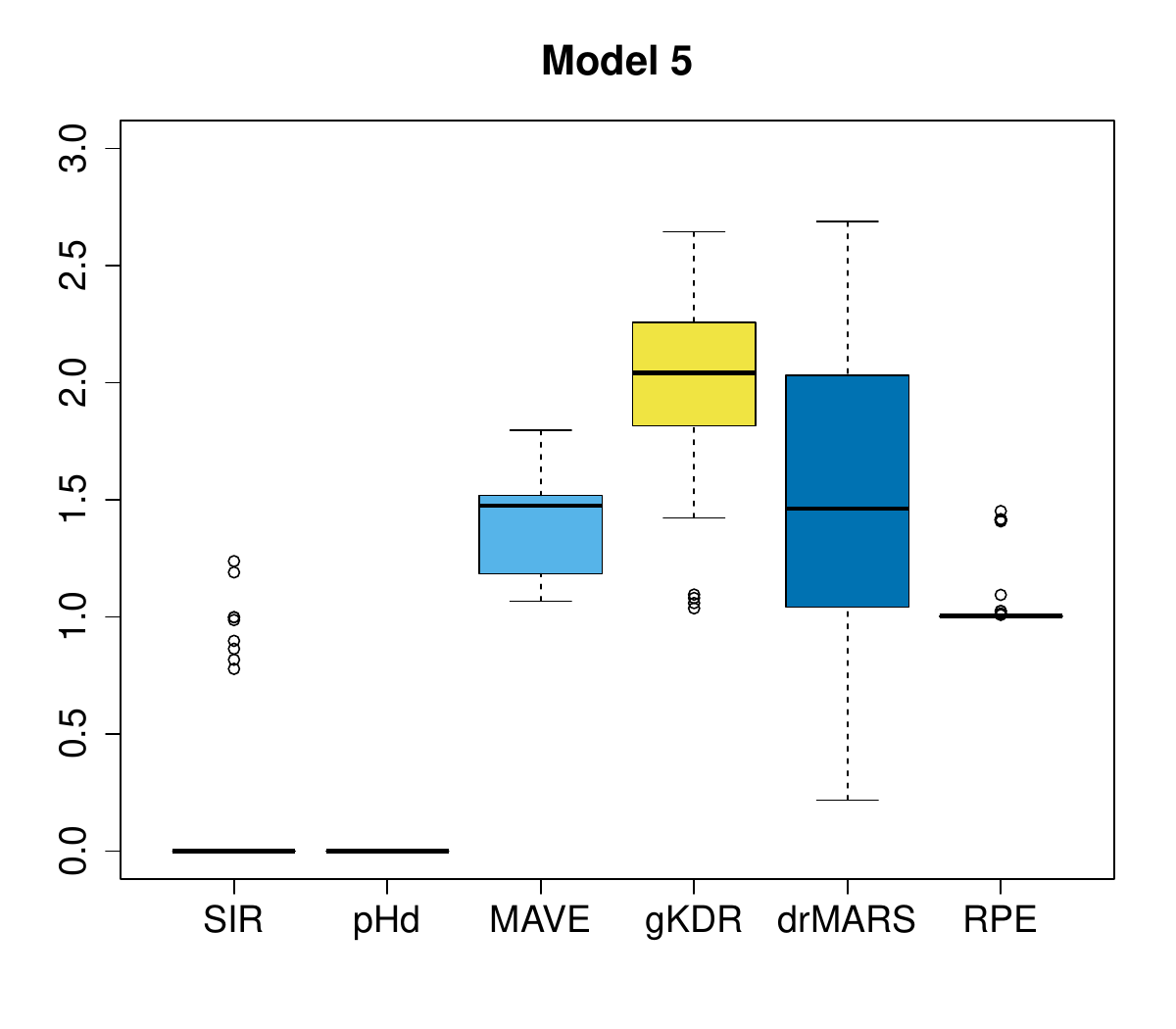}
    \includegraphics[width=0.3\textwidth]{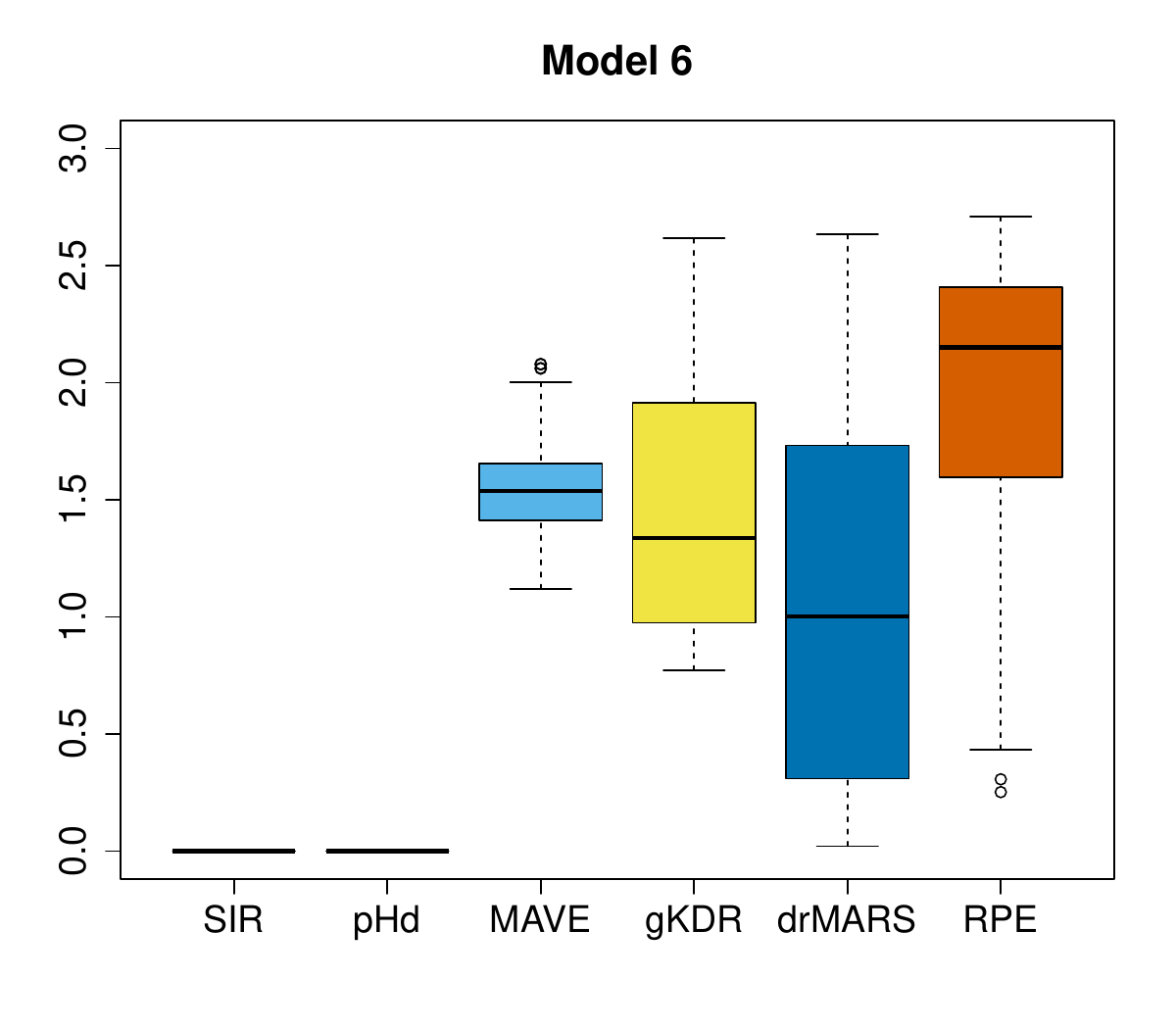}

    \includegraphics[width=0.3\textwidth]{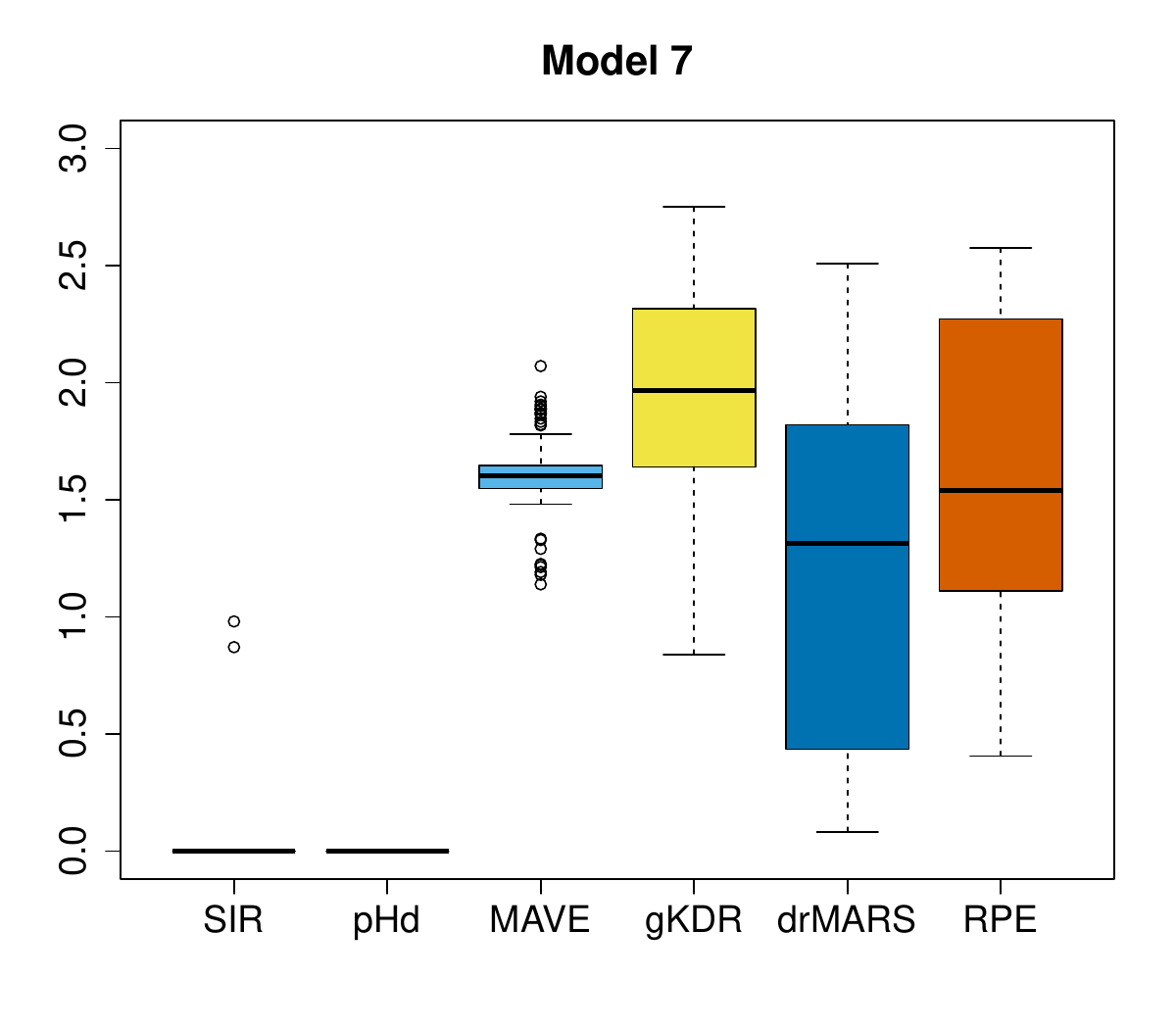}
    \includegraphics[width=0.3\textwidth]{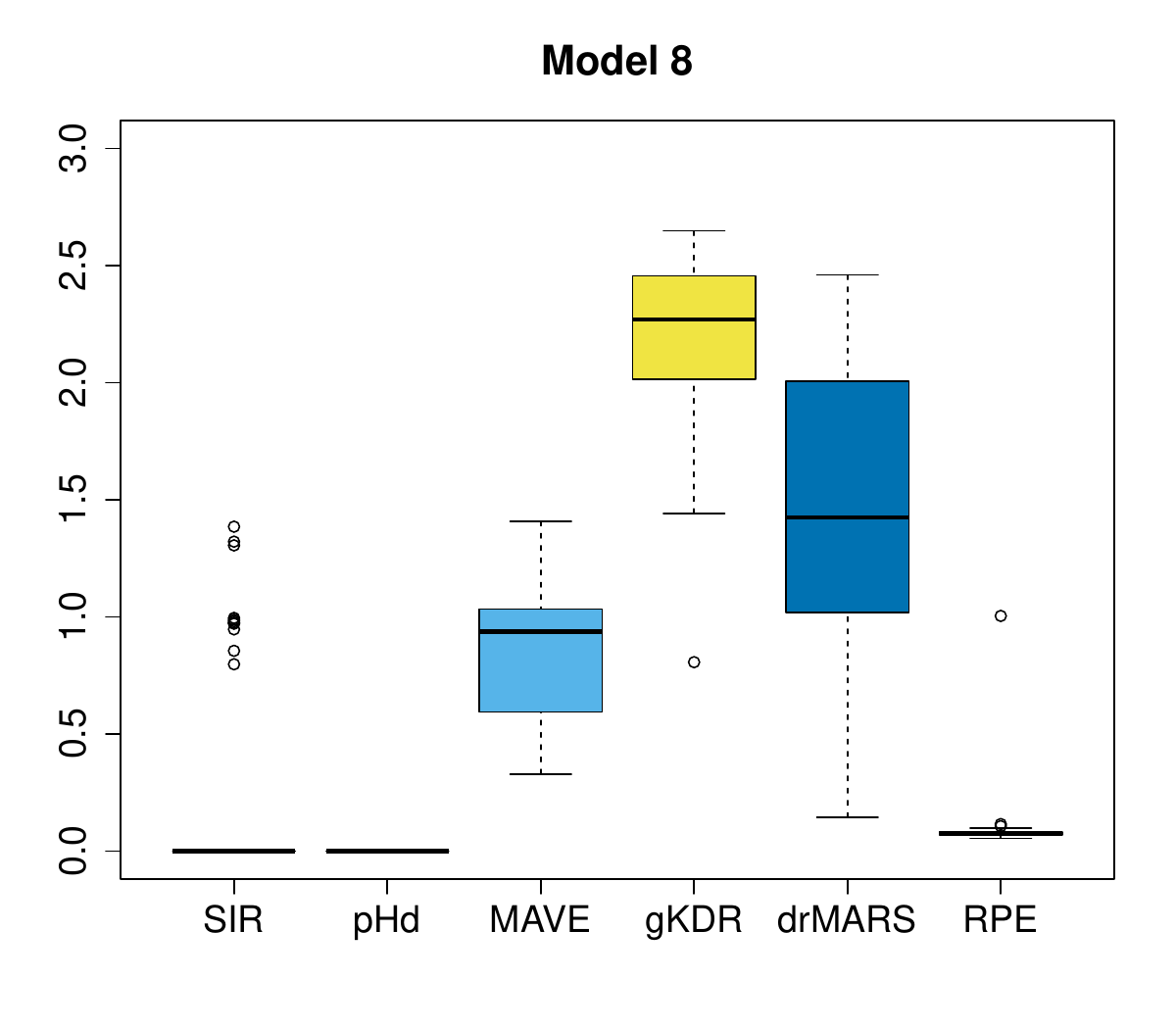}
    \includegraphics[width=0.3\textwidth]{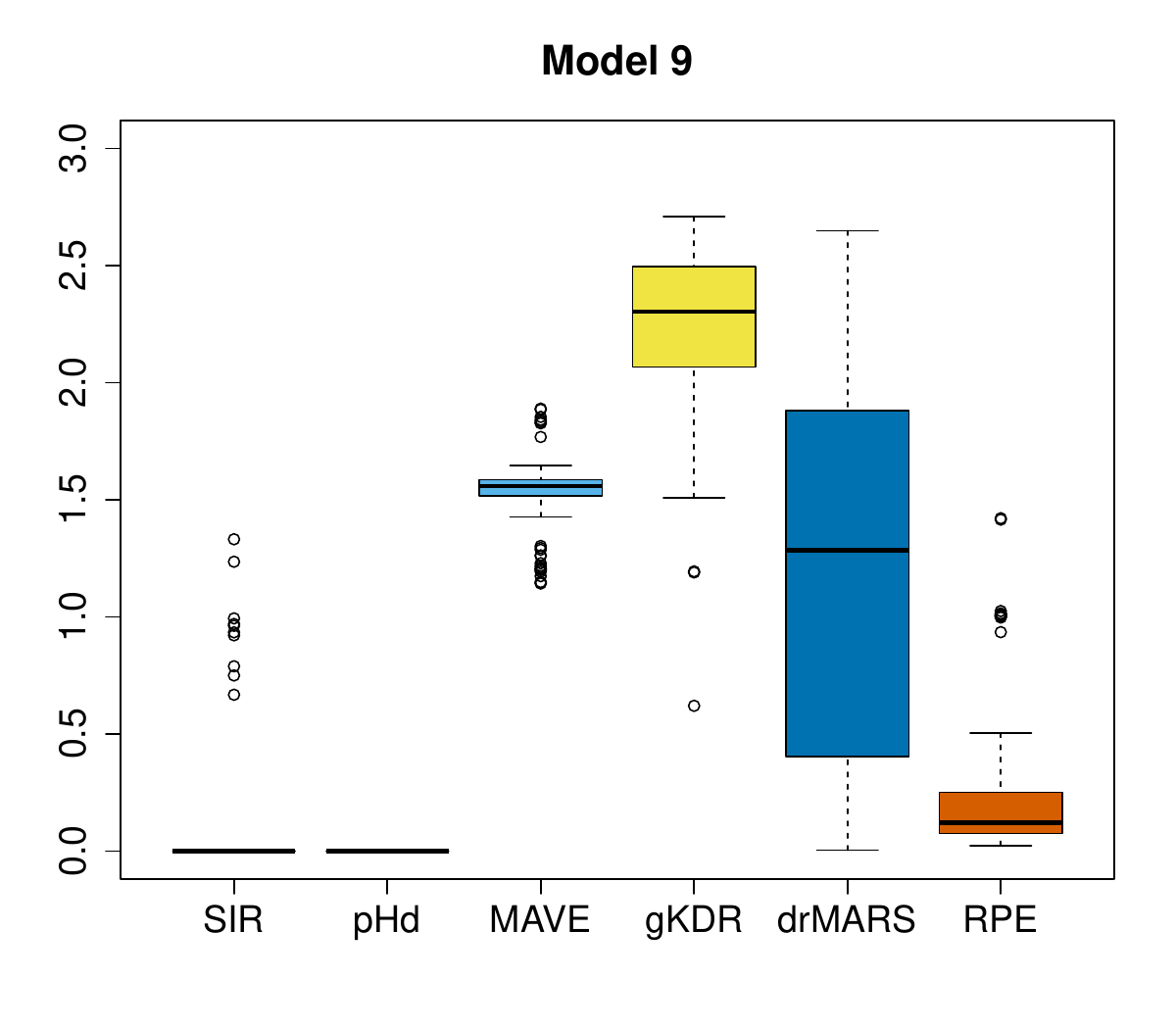}
    
    \caption{{\small Boxplots of $d_{\mathrm{FP}}(\hat{A}_0, A_0)$ for the different methods when $d_0$ is unknown for models 1a-9 over 100 simulations. Our RPE method here is based on the output from Algorithm~\ref{alg:RPEDR estimator}, with the corresponding projection dimension $\hat{d}_0$ chosen via Algorithm~\ref{alg:DimensionEstimator}. For the competing methods, $\hat{d}_0$ is chosen via the corresponding approach described at the beginning of this subsection.   The presented setting is $n=200$, $p=50$.}\label{fig:FP boxplots}}
\end{figure}

\begin{figure}[!ht]
    \centering
    \includegraphics[width=0.3\textwidth]{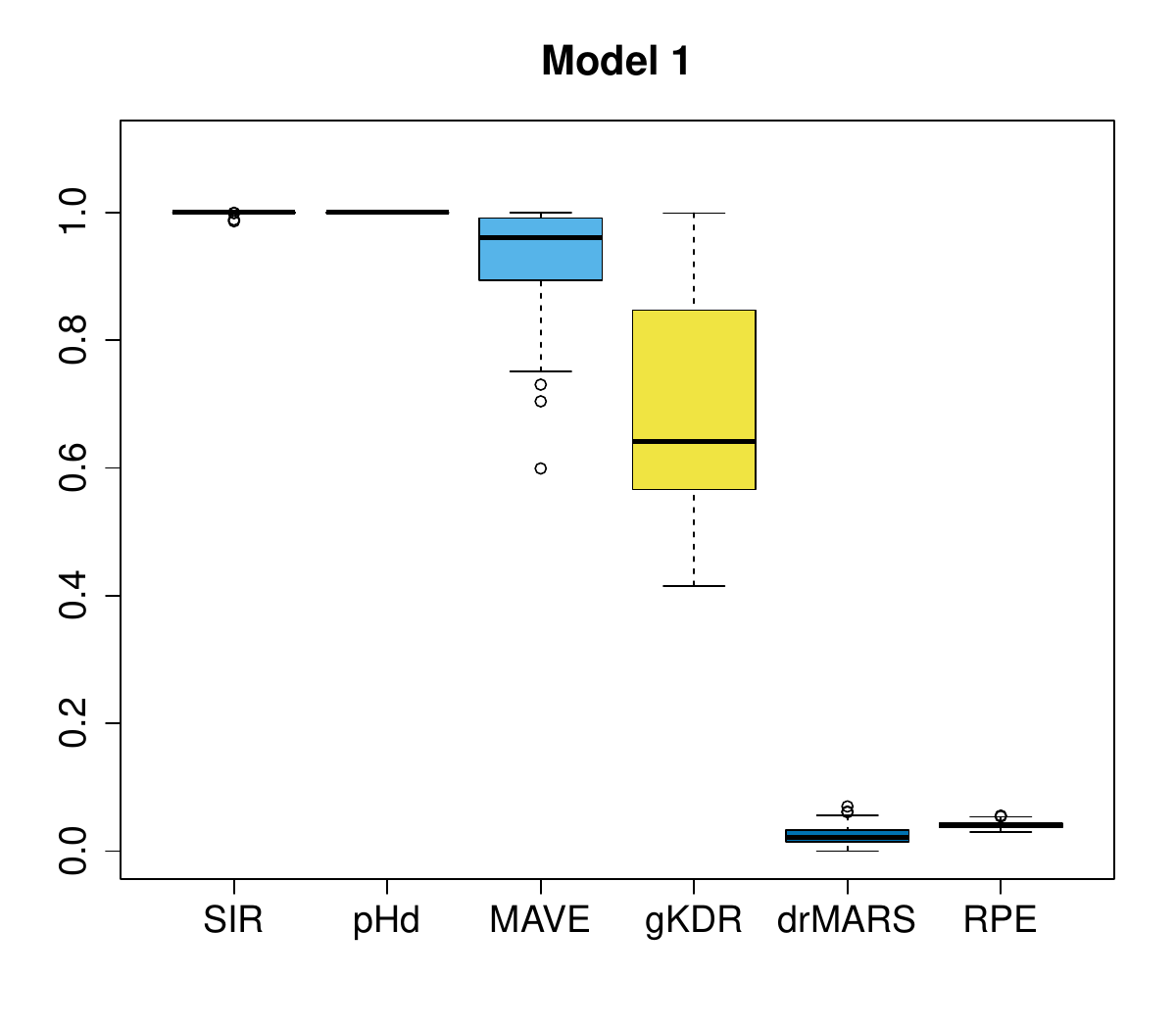}
    \includegraphics[width=0.3\textwidth]{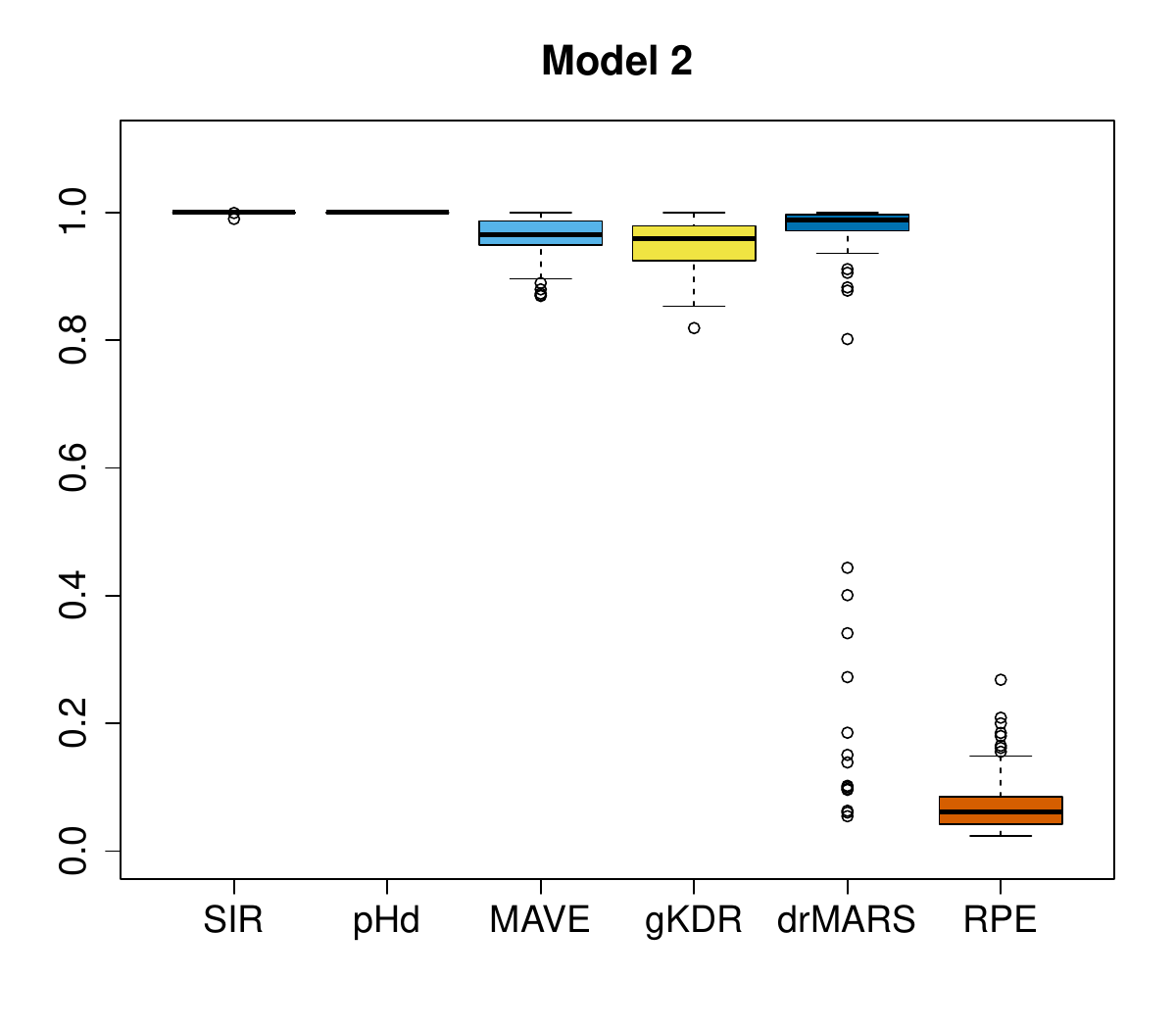}
    \includegraphics[width=0.3\textwidth]{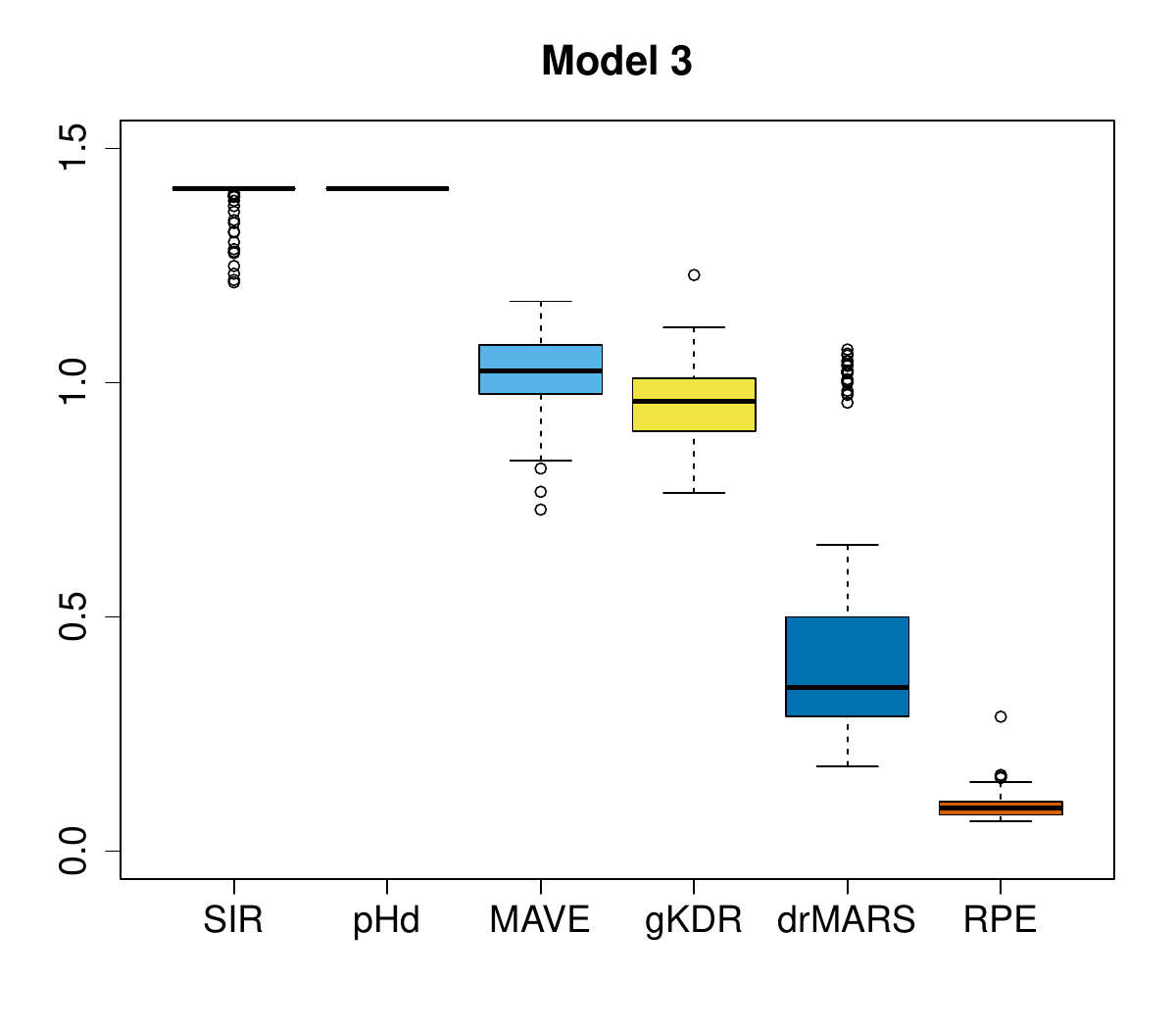}

    \includegraphics[width=0.3\textwidth]{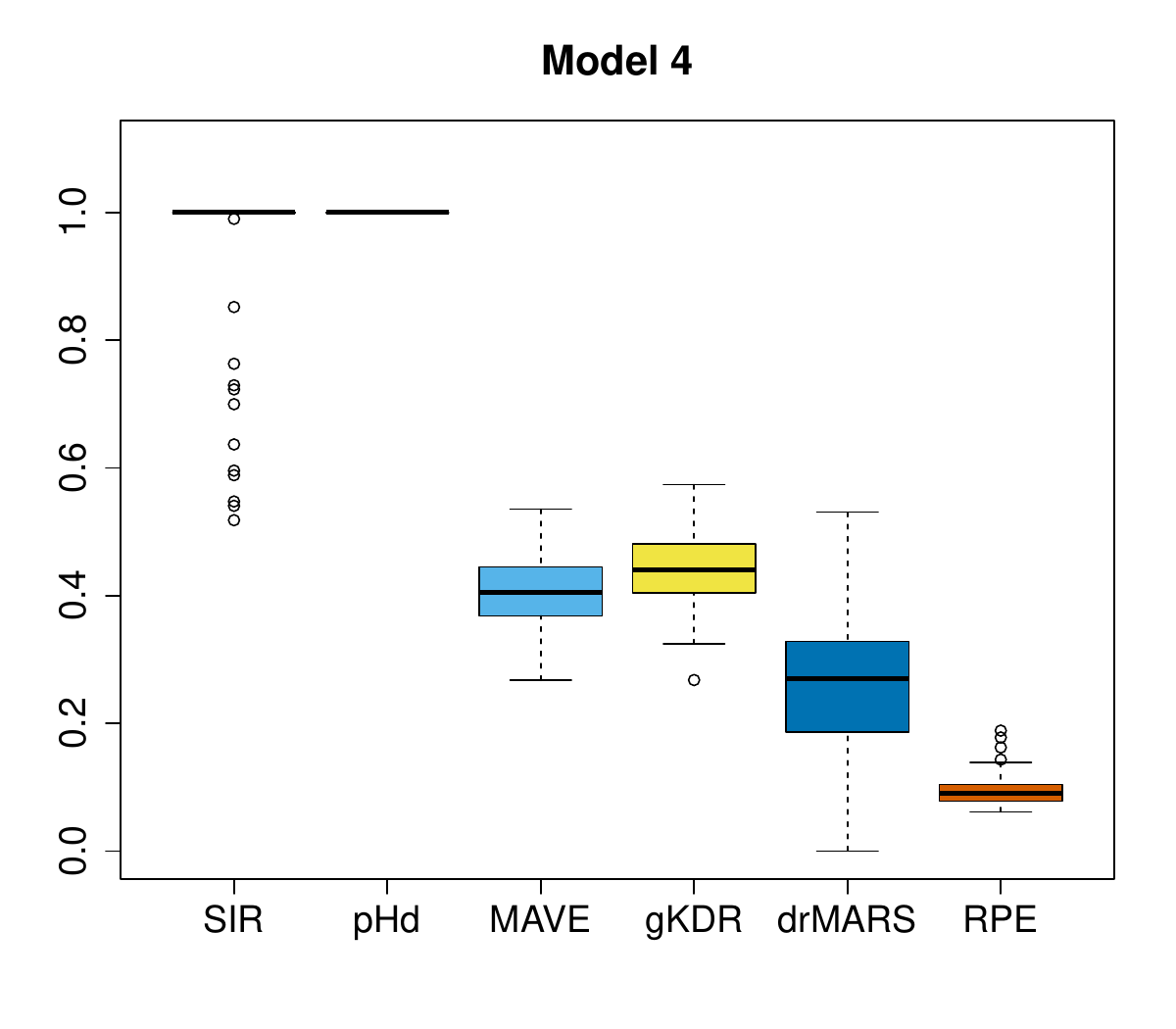}
    \includegraphics[width=0.3\textwidth]{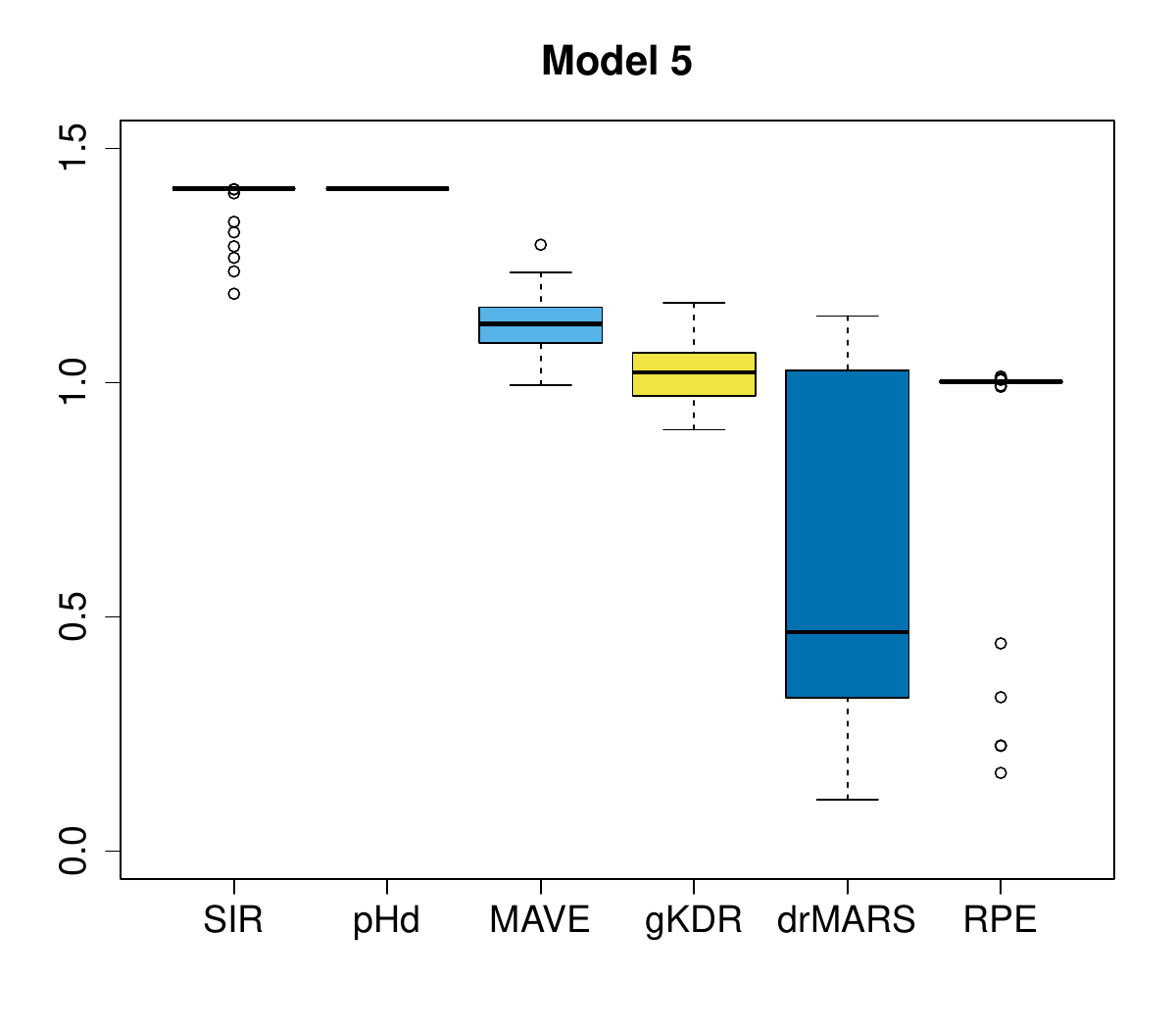}
    \includegraphics[width=0.3\textwidth]{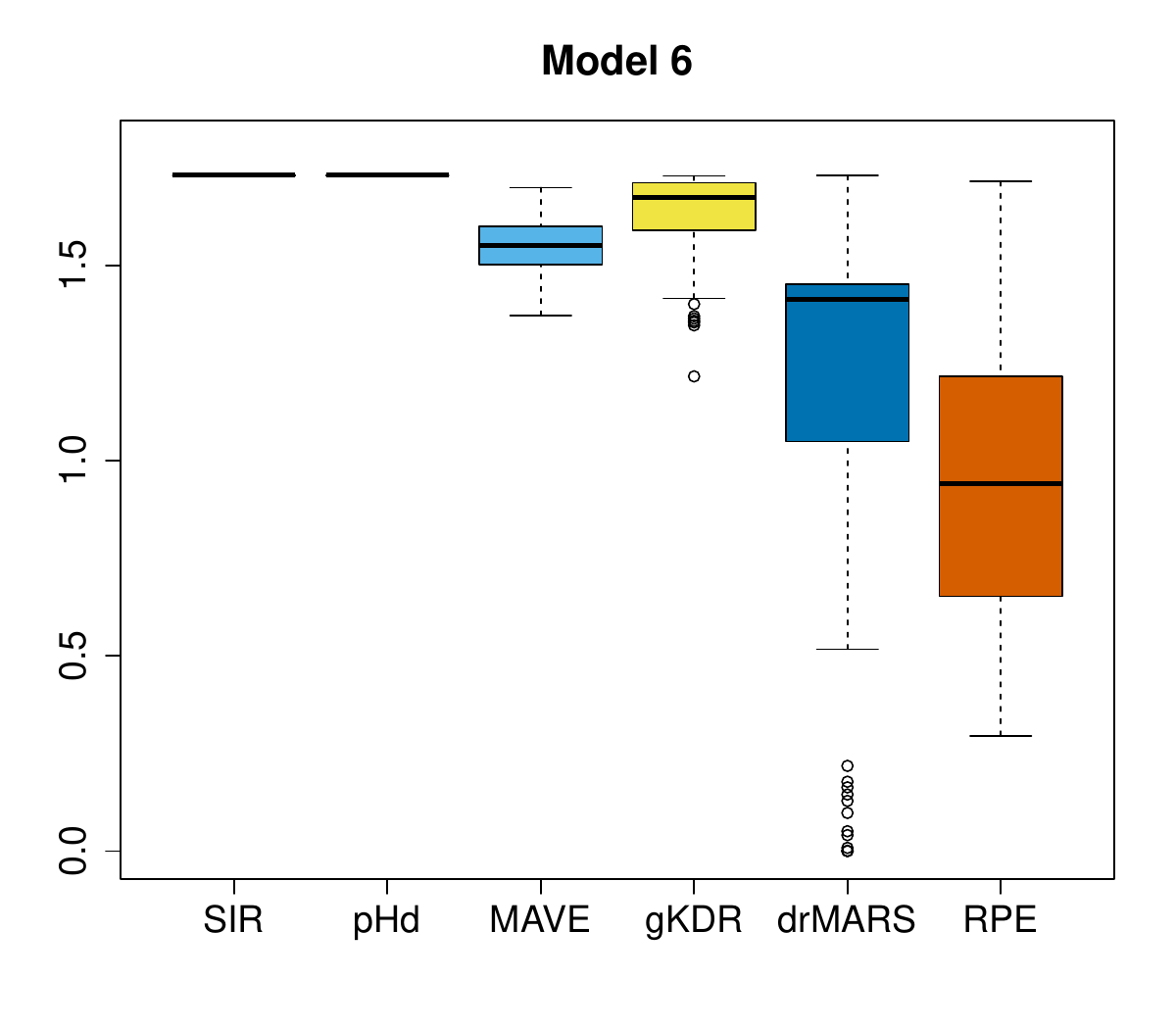}

    \includegraphics[width=0.3\textwidth]{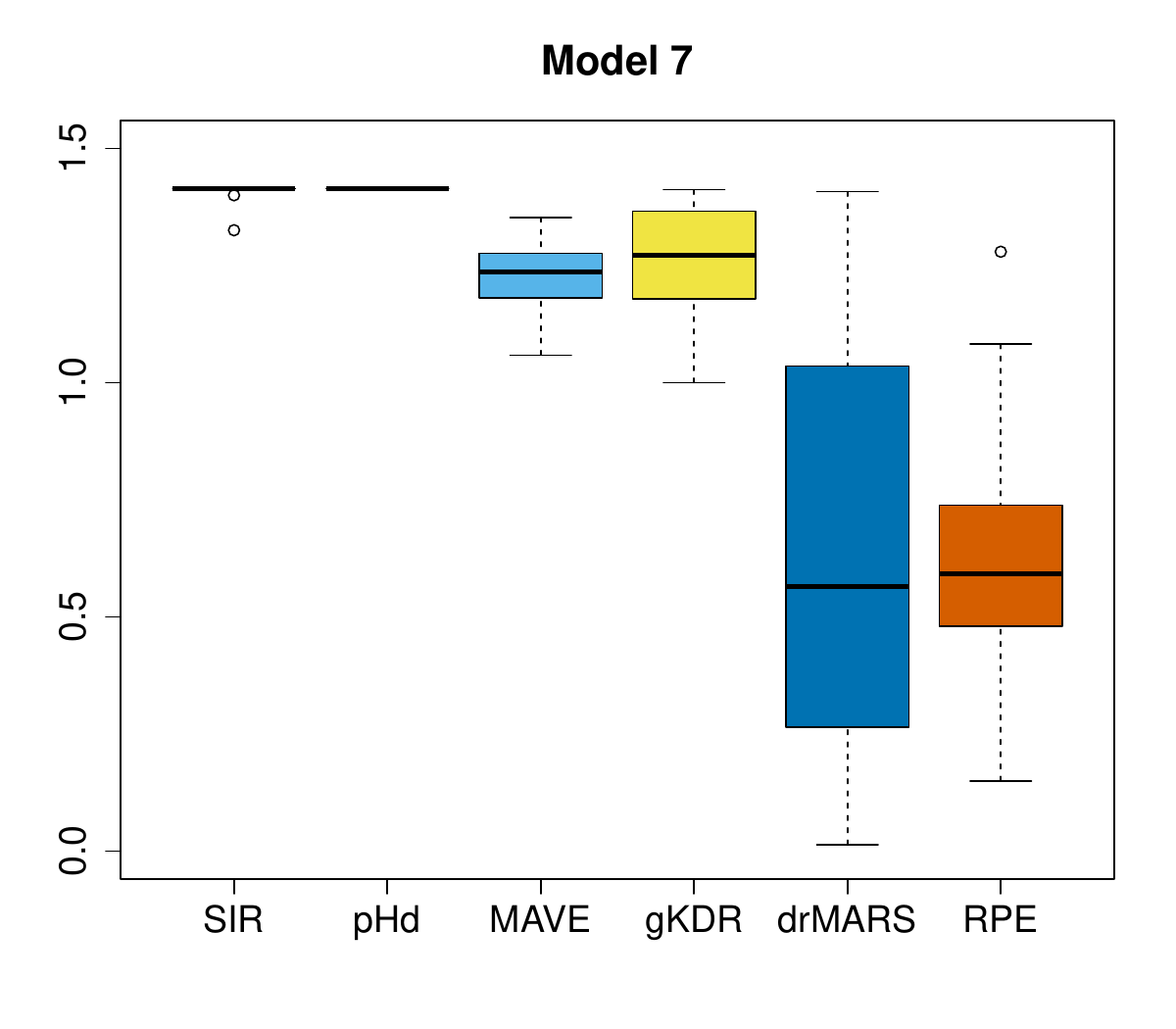}
    \includegraphics[width=0.3\textwidth]{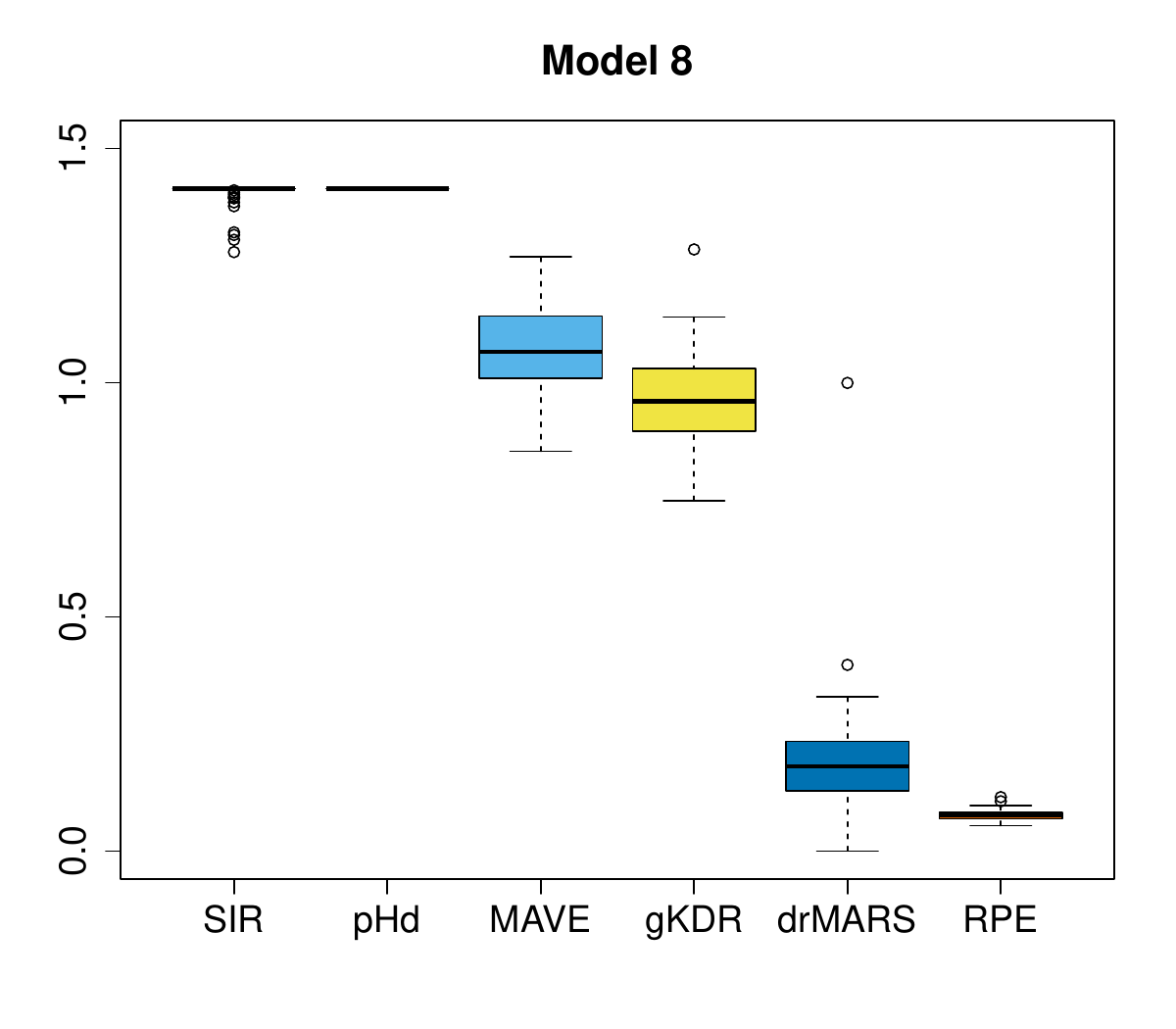}
    \includegraphics[width=0.3\textwidth]{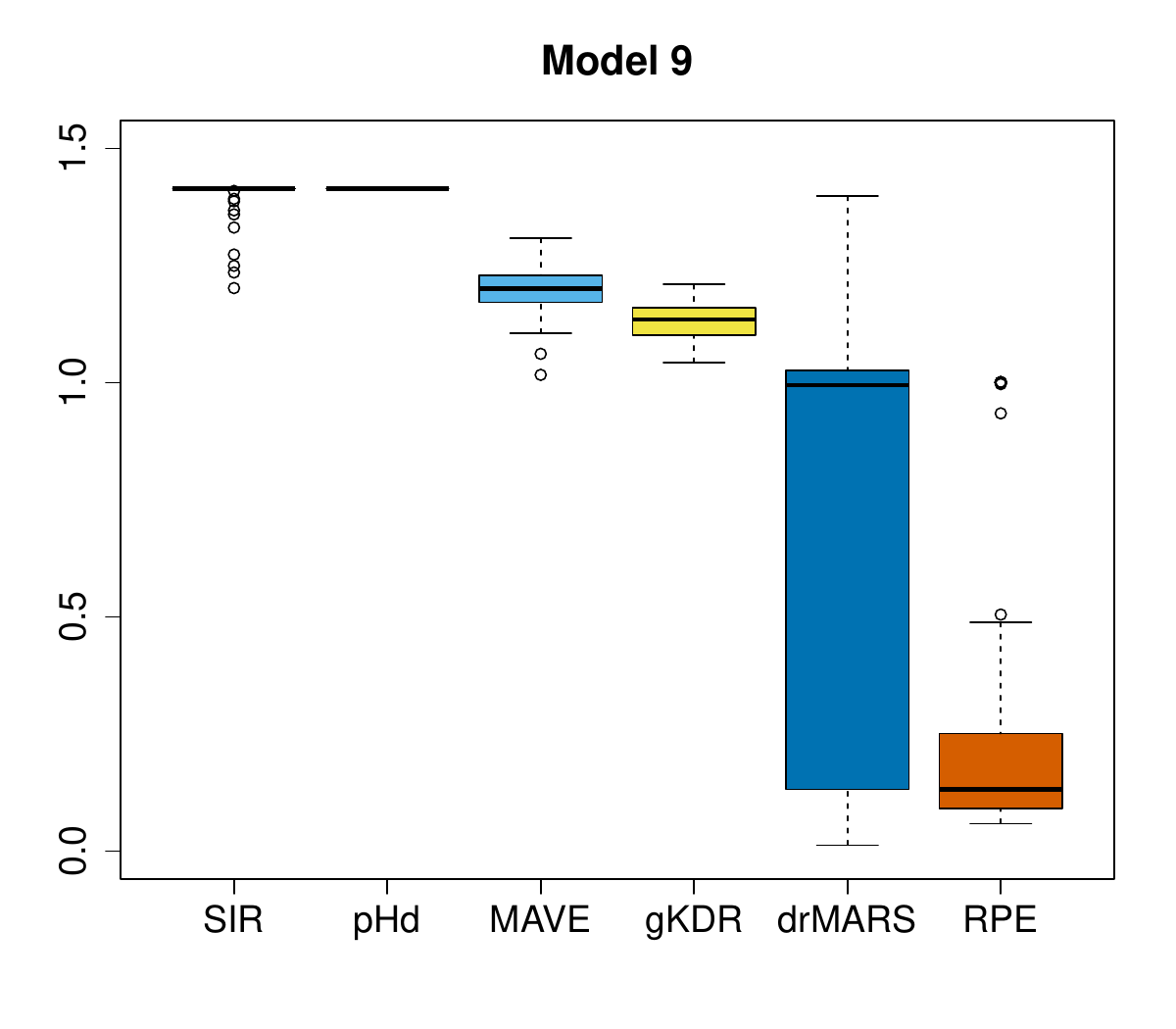}
    
  \caption{{\small Boxplots of $d_{\mathrm{FN}}(\hat{A}_0, A_0)$ for the different methods when $d_0$ is unknown for models 1a-9 over 100 simulations. Our RPE method here is based on the output from Algorithm~\ref{alg:RPEDR estimator}, with the corresponding projection dimension $\hat{d}_0$ chosen via Algorithm~\ref{alg:DimensionEstimator}. For the competing methods, $\hat{d}_0$ is chosen via the corresponding approach described at the beginning of this subsection.   The presented setting is $n=200$, $p=50$.}\label{fig:FN boxplots}}
\end{figure}

Figures~\ref{fig:FP boxplots} and~\ref{fig:FN boxplots} present results of $d_{\mathrm{FP}}$ and $d_{\mathrm{FN}}$, respectively, for different methods across the nine models.   In cases where $\hat{d}_0$ is frequently selected to be $d_0$, namely Models 3, 8 and 9, we observe excellent performance in terms of both false positives and false negatives.  In Models 1 and 4, our algorithm tends to identify the one-dimensional signal direction, by selecting the individual coordinate axes that contribute to the true projection. In these cases, we perform well in terms of the false negative measure, indicating that nearly all of the signal is captured. However, this comes at the cost of higher false positives, though our approach remains competitive with the other methods considered. For Models 2 and 6, we take  $\hat{d}_0$ to be larger $d_0$, and thus include some noise directions alongside the signals. However, our approach still excels in terms of false negatives, outperforming the other methods considered, especially in Model 2, where the competing methods fail to capture the signal directions. In Model 6 we pay a price in terms of false positives for this advantage.   Finally, in Models 5 and 7, our algorithms tends to select the $e_1$ and $e_2$ directions separately, while missing $e_3$, and thus also fails to combine these signals in the optimal way.   As a result, both the false positive and false negative performance are not close to zero. Nevertheless, our method remains competitive with the other algorithms considered in these examples. 

\section{Real world application}
\label{sec:realdata}
In this section, we compare the performance of our methods with competing approaches using real-world data. We make use of the following four datasets from the UC Irvine Machine Learning Repository:
\begin{enumerate}
  \item \textbf{Superconductivity}\footnote{\url{https://archive.ics.uci.edu/dataset/464/superconductivty+data}} \citep{Hamidieh2018ADS}: This dataset consists of measurements on 81 physical and atomic covariates for 21263 superconductors.  The response variable is the critical temperature of the superconductor. We use all $p = 81$ covariates and randomly select a subsample of size $n = 1000$ observations for training, and the remainder used for testing.  
  \item \textbf{Communities and Crime}\footnote{\url{https://archive.ics.uci.edu/dataset/183/communities+and+crime}} \citep{Redmond2002ADS}: This dataset consists of 1994 observations on 122 predictive variables relating to socioeconomic and police data. The response variable of interest is per capita rate of violent crimes. For simplicity, we remove variables with missing values, leaving $p = 105$ predictive covariates in our experiment and we take a random sample of $n = 1329$ observations for training. 
    \item \textbf{Residential Building}\footnote{\url{https://archive.ics.uci.edu/dataset/437/residential+building+data+set}} \citep{Rafiei2016ANM}: This dataset includes construction cost, sales prices, project variables, and economic variables corresponding to real estate single-family residential apartments in Tehran, Iran. There are 372 observations of 105 variables in total.    We use all $p = 105$ predictive variables and take a random sample of $n = 248$ observations for training.  There are two options for the response variable, either sales price or construction cost, both options are included in our experiments.
   \item \textbf{Geographical Origin of Music}\footnote{\url{https://archive.ics.uci.edu/dataset/315/geographical+original+of+music}} \citep{Zhou2014PredictingTG}: This dataset consists of 1059 observations on 116 variables relating to audio features of a piece of music. We use the latitude of the piece's origin as the response variable.  In this case all $p=116$ variables are used, and a subsample of $n = 706$ observations is taken for training.  
\end{enumerate}
In this section's experiments, since the true projection $A_0$ is unknown, we  measure the performance of the different techniques based on regression prediction accuracy after projecting the data.  More precisely, we fixed the projection dimension $\hat{d}_0$, set to either 3 or 10, and then apply the methods described in Section~\ref{sec:simsd0known} using only the training data. The interpretation is that we select the $\hat{d}_0$ most important directions chosen by each method.  We then apply MARS\footnote{We use cross-validation to detemine the degree, with the maximum order of interactions set to four.}  to the projected training data after each method. The performance is measured using the root-mean-squared-error (RMSE) on the corresponding projected test set.  The results of our real data experiments are presented in Table~\ref{table:realdata}, where we see that our methods enjoy competitive performance across all four datasets considered here. 
 
\begin{table}[!ht]
\centering
\small
\begin{tabular}{ c | c | c | c c c c c c | c c }
 & &  & \multicolumn{6}{c|}{Dimension reduction technique} & \multicolumn{2}{c}{Our Methods}\\
\hline
Dataset & No proj & $\hat{d}_0$ & SIR & pHd & MAVE & DR & gKDR & drMARS & RPE & RPE2 \\
\hline
\multirow{2}{*}{1} & \multirow{2}{*}{17.4} & 3 & \textbf{16.5} & 30.2 & 19.3 & 29.9 & 18.1 & 18.8 & 19.4 & 17.5\\
& & 10 & 17.3 & 27.8 & 18.0 & 22.2 & 17.9 & 20.4 & \textbf{17.1} & \textbf{17.1} \\
\hline
\multirow{2}{*}{2} &\multirow{2}{*}{0.180}& 3 & \textbf{0.127} & 0.212 & 0.140 & 0.184 & 0.132 & 0.138 &  0.129 & 0.129 \\
& & 10 & 0.137 & 0.187 & 0.135 & 0.136 & 0.132 & 0.140 & 0.132 & \textbf{0.126} \\
\hline
\multirow{2}{*}{3(sales)} & \multirow{2}{*}{150} & 3 & 295 & 1198 & 218 & / & 386 &  \textbf{158} & 218 & 218  \\
&  & 10 & 295 & 699 & 295 &/ & 429 & 211 & \textbf{124} & \textbf{124} \\
\hline
\multirow{2}{*}{3(cons.)} & \multirow{2}{*}{32.1} & 3 & 70.4 & 203 & 47.6 & / & 161 & 136 &  \textbf{34.7} &  \textbf{34.7}  \\
&  & 10 & 68.4 & 301 & 42.4 & / & 47.1 & 29.6 & \textbf{24.1} & \textbf{24.1} \\
\hline
\multirow{2}{*}{4} & \multirow{2}{*}{25.7} & 3 & 16.9 & 18.1 & 16.7 & / & 17.7 & 17.0 & 16.9 & \textbf{16.1} \\
& & 10 & 17.9 & 17.5 & \textbf{15.8} & / & 17.0 & 19.4 & 17.0 & 16.0 \\
\hline
\end{tabular}
\caption{\small The RMSE on the test set for our real data experiments. For comparison, we include the corresponding results when we apply the MARS algorithm directly to the original data in the ambient space (i.e. without projecting the data first). For datasets 3 and 4 the DR method had  very large RMSE and we therefore omitted the results.}
\label{table:realdata}
\end{table}

\section{Discussion and extensions}
\label{sec:conclusions}
We have introduced a new general approach to sufficient dimension reduction, based on aggregating an ensemble of carefully chosen projections.  The proposed framework is highly flexible, allowing for different distributions of random projections and a user-specified choice of base regression method.  We provide default recommendations for these aspects based on extensive numerical experiments, but we anticipate that many other possible choices will be effective in certain cases.       

There are several ways in which our framework could be extended further.  First, a natural assumption in high-dimensional settings is sparsity, specifically the case where the projection $A_0$ is sparse (i.e.~contains many zero entries).  In this scenario, it may be more effective to consider sparse random projections, but perhaps more importantly, we'd suggest that employing sparse singular value decomposition \citep[see, for example,][]{yang2014sparse} to aggregate the chosen projections in order to ensure the resulting estimate $\hat{A}_0$ is sparse.  

A second possible extension arises in the context of additive index models \citep[see, for example,][]{friedman1981projection,ruan2010dimension}.  In this case, we may seek to sequentially find the signal directions one-by-one. For the first signal direction, we could apply our Algorithm~\ref{alg:RPEDR estimator} directly, keeping only the first singular vector.  To identify the subsequent signal directions, we would then apply a modification of our algorithm that always includes the signals found so far and considers random projections that are orthogonal to the previously identified signals.     

Finally, while this paper has focused on dimension reduction, specifically estimating the projection matrix $A_0$, it would be interesting to consider alternative ways of aggregating the results when the primary objective is to predict the response $Y$.  In particular,  improved predictive performance might be achieved  by directly aggregating the predictions from the chosen projections, as opposed to fitting a new regression model after projecting the data with our estimated projection $\hat{A}_0$. These extensions are left for future work.

\section*{Acknowledgements}
The research of TIC was supported by Engineering and Physical Sciences Research Council (EPSRC) New Investigator Award EP/V002694/1.

\bibliography{SDRbib.bib}
\bibliographystyle{apalike}

\appendix

\section{Technical results and proofs}
\label{sec:proofs}

The main aim of this section is to prove Theorem~\ref{thm:choiceofL}, for which we will make use of the following basic technical lemma, see, for example, \citet[Theorem~4.2.2]{horn2012matrix}.
\begin{lemma} 
\label{lemma1}
    Let $A \in \mathbb{R}^{p \times p}$ be symmetric matrix, let $\lambda_1, \cdots, \lambda_p$ be the eigenvalues of $A$, $\lambda_{max} := \max(\lambda_1,\cdots,\lambda_p)$, $\lambda_{min} := \min (\lambda_1,\cdots,\lambda_p)$. Then $\lambda_{max} = \max_{\left \| v \right \|_2 = 1} v^T A v$, and $\lambda_{min} = \min_{\left \| v \right \|_2 = 1} v^T A v$.
\end{lemma}

\begin{proof}[Proof of Theorem~\ref{thm:choiceofL}]
Recall from Algorithm~\ref{alg:RPEDR estimator} that $\hat{\Pi} = \frac{1}{L} \sum_{\ell = 1}^{L} \mathbf{P}_{\ell, *} \mathbf{P}_{\ell, *}^T$ and $\hat{A}_0^{L}$ denotes the $\mathbb{R}^{p \times d_0}$ which has columns given by the eigenvectors of $\hat{\Pi}$ ordered according to the corresponding eigenvalues. Recall also that $\Pi^{\infty} = \mathbb{E} \mathbf{P}_{\ell, *} \mathbf{P}_{\ell, *}^T$.  Let $\Pi := A_0 A_0^T \in \mathbb{R}^{p\times p}$, and for $j \in [d_0]$ write $a_j$ for the $j$th column of $A_0$, so that $A_0 =(a_1,\cdots, a_{d_0})$.  Since $A_0^TA_0 = I_{d_0\times d_0}$, the $d_0$ largest eigenvalues of $\Pi$ are 1 and the remainder are $0$. Indeed, $\rank{(\Pi)} = \rank{(A_0)} = d_0$, so $\Pi$ has $d_0$ nonzero eigenvalues. Moreover, $\Pi a_j = A_0 A_0 ^T a_j = A_0 (A_0 ^T a_j) = A_0 e_j = a_j$ for $j \in [d_0]$.  Then, since $\Pi$ and $\hat{\Pi}$ are symmetric, by the Davis--Kahan Theorem \citep[Theorem~2]{yu2015useful} we have
\begin{equation}
\label{eq:Davis_kahan}
\begin{split}
    d_{\mathrm{F}}\bigl(\mathcal{S}(\hat{A}_0^{L}), \mathcal{S}({A}_0)\bigr)
    \equiv \bigl \| \sin{\Theta(\hat{A}_0^L, A_0)} \bigr \|_{F} 
    & \leq 2 d_0^{1/2} \bigl \| \hat{\Pi} - \Pi \bigl \|_{\mathrm{op}}  \\
  & \leq 2d_0^{1/2} \Bigl( \| \hat{\Pi} - \Pi^{\infty} \|_{\mathrm{op}} + \| \Pi^{\infty} - \Pi \|_{\mathrm{op}} \Bigr).
\end{split}
\end{equation}
The second term in~\eqref{eq:Davis_kahan} doesn't depend on $L$ nor the randomness in the projections.  The remainder of the proof therefore involves controlling the expectation of the first term in~\eqref{eq:Davis_kahan}. To that end, we first show that 
\begin{equation}
\label{eq:OPtailbound}
\begin{split}
    \mathbb{P} \Bigl ( \bigl \| \hat{\Pi} - \Pi^{\infty} \bigr \|_{\mathrm{op}} \geq t \Bigr ) \leq p \cdot e^{-t^2 L / 8}
\end{split}
\end{equation}
for $t > 0$.  To see this, for $\ell \in [L]$ let $\Pi_{\ell} := \mathbf{P}_{\ell, *} \mathbf{P}_{\ell, *}^T$ and write $J_{\ell} := \frac{1}{L}(\Pi_{\ell} - \Pi^{\infty})$, then
\[
    \hat{\Pi} - \Pi^{\infty} = \frac{1}{L} \sum_{\ell=1}^L \mathbf{P}_{\ell, *} \mathbf{P}_{\ell, *}^T - \Pi^{\infty}   = \sum_{\ell=1}^L J_{\ell}.
\]
Further $\mathbb{E}(J_{\ell}) = 0$ and $J_{\ell}$ (and thus also $J_{\ell}^2$) is symmetric.  Now, we show that $J_{\ell}^2 \preccurlyeq \frac{1}{L^2} I_{p \times p}$, in other words $\frac{1}{L^2}I_{p \times p} - J_{\ell}^2$ is positive semidefinite, almost surely, or equivalently that 
\begin{equation}
\label{eq:possemideff}
\sup_{\{v: \|v\|=1\}} v^T J_{\ell}^2 v \leq \frac{1}{L^2}.
\end{equation}
First, since $\mathbf{P}_{\ell,*}^T\mathbf{P}_{\ell,*} = I_{d \times d}$, the eigenvalues of $\Pi_\ell$ take values in $\{0,1\}$ and we have by Lemma~\ref{lemma1} that $0 \leq v^T \Pi_l v \leq 1$ whenever $\| v \|_2=1$.  Moreover, by the linearity of expectation, we also have
\[
v^T \Pi^{\infty} v = v^T \mathbb{E} (\Pi_{\ell}) v = \mathbb{E} (v^T \Pi_{\ell} v ) \in [0,1] 
\] 
 whenever $\| v \|_2=1$.  It follows that 
\begin{equation*}
\label{Pi_l - Pi_infty inequality}
v^T(\Pi_{\ell} - \Pi^{\infty})v = v^T \Pi_{\ell} v - v^T \Pi^{\infty} v \in [-1,1], \qquad \textrm{for all} \ \left \| v \right \|_2=1. 
\end{equation*}
Hence we can write $\Pi_{\ell} - \Pi^{\infty} = U_{\ell} D_{\ell} U_{\ell}^T$, where $U_{\ell} \in \mathbb{R}^{p\times p}$ is orthonormal and $D_{\ell}$ is a diagonal $p\times p$ matrix with all entries between $-1$ and $1$  \citep[Corollary~2.5.11]{horn2012matrix}.  We deduce that, for all $v \in \mathbb{R}^p$ with  $\| v \|_2=1$, we have
\[
   v^T J_{\ell}^2 v = \frac{1}{L^2}  v^T (\Pi_{\ell} - \Pi^{\infty})^2 v = \frac{1}{L^2}  v^T U_{\ell} D_{\ell}^2 U_{\ell}^T  v \leq \frac{1}{L^2} \|U_{\ell}^T  v\| = \frac{1}{L^2},
\]
which establishes~\eqref{eq:possemideff}.

Finally, because each of the projections $\mathbf{P}_{\ell,*}$ is chosen from a disjoint group of independently generated random projections, we have that $\mathbf{P}_{\ell,*}$ for $\ell \in [L]$ are independent and therefore $J_{\ell}$ for $\ell \in [L]$ are also independent.  It then follows by the Matrix Hoeffding inequality (see, for example \citet[Theorem~1.3]{tropp2012user}), that 
\[
    \mathbb{P} \Bigl ( \bigl \| \hat{\Pi} - \Pi^{\infty} \bigr \|_{\mathrm{op}} \geq t \Bigr ) \leq p \cdot e^{-\frac{t^2}{8 \| \sum_{\ell = 1}^{L} \frac{I_{p \times p}}{L^2} \|_{\mathrm{op}}}} =  p \cdot e^{-\frac{t^2 L}{8}},
\]
which establishes the bound in \eqref{eq:OPtailbound}.  To complete the proof we bound the expectation as follows:
\begin{equation*}
\label{expectation}
\begin{split}
    \mathbb{E} \Bigl (\bigl \| \hat{\Pi} - \Pi^{\infty} \bigr \|_{\mathrm{op}} \Bigr ) 
    &= \int_0^{\infty} \mathbb{P} \Bigl ( \bigl \| \hat{\Pi} - \Pi^{\infty} \bigr \|_{\mathrm{op}} \geq t \Bigr ) dt \\
    &\leq  \int_0^{\infty} p \cdot e^{-\frac{t^2 L}{8}} dt  = \frac{p \sqrt{2 \pi}}{\sqrt{L}},
\end{split}
\end{equation*}
as required. 
\end{proof}

\section{Full simulation results}
\label{sec:fullsims}

In this section, we present the full results of our simulation study, which follows the design of the experiments described in Section~\ref{sec:numerical} in the main text.  Here, for each model, we consider the settings with $p \in \{20, 50, 100\}$ and $n \in \{50, 200, 500\}$.  The full description of the algorithms considered, along with how any tuning parameters were chosen, is given in Section~\ref{sec:numerical}.   To save space, rather than presenting boxplots here, we instead focus on the mean and standard error of the the sin-theta distance (when $d_0$ is known) over 100 repeats of the experiments. The results are given in Table~\ref{table:1-4 sin theta distance} (for Models 1a-4) and Table~\ref{table:5-9 sin theta distance} (for Models 5-9).

\begin{table}[!ht]
\centering
\scriptsize
\begin{tabular}{ c c c  | c c c c c c | c c }
 \multicolumn{3}{c|}{Setting} & \multicolumn{6}{c}{Competing Methods} & \multicolumn{2}{|c}{Our Methods}\\
\hline
Model& $p$ & $n$ & SIR & pHd & MAVE & DR & gKDR & drMARS & RPE & RPE2 \\
\hline 
& \multirow{3}{*}{20} & 50
& $0.97_{0.05}$ & $0.81_{0.10}$ & $0.87_{0.14}$ & $0.97_{0.05}$ & $0.92_{0.09}$ & $0.13_{0.18}$ & $0.32_{0.14}$ & $\mathbf{0.12_{0.04}}$ \\
                   &                      & 200 
& $0.97_{0.05}$ & $0.40_{0.08}$ & $0.24_{0.04}$ & $0.33_{0.06}$ & $0.29_{0.06}$ & $\mathbf{0.03_{0.01}}$ & $0.20_{0.11}$ & $0.06_{0.02}$ \\ 
                &                      & 500 
& $0.96_{0.07}$ & $0.24_{0.05}$ & $0.10_{0.02}$ & $0.19_{0.03}$ & $0.14_{0.03}$ & $\mathbf{0.01_{0.01}}$ & $0.13_{0.08}$ & $0.05_{0.02}$ \\  
\multirow{2}{*}{1a}                & \multirow{3}{*}{50}  & 50
& / & / & $0.93_{0.09}$ & / & $0.97_{0.03}$ & $0.34_{0.28}$ & $0.46_{0.17}$ & $\mathbf{0.14_{0.08}}$ \\
\multirow{3}{*}{($d_0 = 1$)}           &                      & 200 
 & $0.99_{0.02}$ & $0.74_{0.09}$ & $0.93_{0.08}$ & $0.65_{0.12}$ & $0.98_{0.04}$ & $\mathbf{0.04_{0.01}}$ & $0.32_{0.17}$ & $0.06_{0.02}$ \\ 
             &                      & 500 
& $0.98_{0.04}$ & $0.43_{0.06}$ & $0.42_{0.12}$ & $0.34_{0.04}$ & $0.57_{0.20}$ & $\mathbf{0.02_{0.02}}$ & $0.28_{0.15}$ & $0.05_{0.02}$ \\ 
                   & \multirow{3}{*}{100} & 50
& / & / & $0.95_{0.07}$ & / & $0.98_{0.02}$ & $0.33_{0.28}$ & $0.57_{0.17}$ & $\mathbf{0.29_{0.29}}$\\                   
                   &              & 200 
& $1.00_{0.00}$ & $0.95_{0.03}$ & $0.97_{0.04}$ & $0.99_{0.01}$ & $0.99_{0.02}$ & $\mathbf{0.05_{0.02}}$ & $0.46_{0.17}$ & $\mathbf{0.06_{0.02}}$ \\  
                   &                      & 500 
& $0.99_{0.01}$ & $0.72_{0.07}$ & $0.94_{0.08}$ & $0.54_{0.04}$ & $0.99_{0.02}$ & $\mathbf{0.01_{0.01}}$ & $0.38_{0.16}$ & $0.05_{0.02}$ \\
\hline

                   & \multirow{3}{*}{20}  & 50 
& $0.98_{0.03}$ & $0.97_{0.03}$ & $0.97_{0.04}$ & $0.97_{0.05}$ & $0.97_{0.03}$ & $\mathbf{0.95_{0.17}}$ & $0.97_{0.08}$ & $0.97_{0.06}$ \\
                   &                      & 200 
& $0.98_{0.03}$ & $0.98_{0.03}$ & $0.98_{0.02}$ & $0.97_{0.03}$ & $0.98_{0.03}$ & $0.56_{0.46}$ & $\mathbf{0.06_{0.05}}$ & $0.10_{0.16}$\\ 
                   &                      & 500 
& $0.98_{0.03}$ & $0.97_{0.04}$ & $0.97_{0.03}$ & $0.97_{0.05}$ & $0.98_{0.03}$ & $0.18_{0.32}$ & $\mathbf{0.03_{0.00}}$ & $\mathbf{0.03_{0.00}}$ \\
\multirow{2}{*}{2} & \multirow{3}{*}{50}  & 50
& / & / & $0.98_{0.04}$ & / & $0.99_{0.01}$ & $0.99_{0.01}$ & $\mathbf{0.96_{0.13}}$ & $0.98_{0.05}$ \\                   
\multirow{3}{*}{($d_0 = 1$)}                   
                   &                      & 200 
& $0.99_{0.02}$ & $0.99_{0.01}$ & $0.99_{0.02}$ & $0.99_{0.02}$ & $0.99_{0.01}$ & $0.93_{0.23}$ & $\mathbf{0.09_{0.14}}$ &  $0.27_{0.37}$\\             
                   &                      & 500  
& $0.99_{0.01}$ & $0.99_{0.02}$ & $0.99_{0.01}$ & $0.99_{0.01}$ & $0.99_{0.01}$ & $0.60_{0.47}$ & $\mathbf{0.03_{0.01}}$ & $\mathbf{0.03_{0.01}}$ \\
                   & \multirow{3}{*}{100} & 50
& / & / & $0.99_{0.02}$ & / & $1.00_{0.01}$ & $\mathbf{0.98_{0.08}}$ & $1.00_{0.01}$ & $1.00_{0.01}$ \\
                   &                      & 200 
& $0.99_{0.04}$ & $0.99_{0.01}$ & $0.99_{0.02}$ & $1.00_{0.01}$ & $1.00_{0.01}$ & $0.94_{0.21}$ & $\mathbf{0.16_{0.26}}$ & $0.40_{0.42}$ \\  
                   &                      & 500 
& $1.00_{0.01}$ & $0.99_{0.01}$ & $0.99_{0.01}$ & $1.00_{0.01}$ & $0.99_{0.01}$ & $0.77_{0.41}$ & $\mathbf{0.02_{0.01}}$ & $0.03_{0.01}$ \\
\hline

                   & \multirow{3}{*}{20}  & 50
& $1.30_{0.07}$ & $1.27_{0.08}$ & $1.07_{0.10}$ & $1.29_{0.07}$ & $1.16_{0.11}$ & $0.92_{0.28}$ & $\mathbf{0.50_{0.31}}$ & $\mathbf{0.50_{0.30}}$ \\                   
 &                      & 200 
& $0.85_{0.13}$ & $0.88_{0.12}$ & $0.72_{0.13}$ & $0.99_{0.11}$ & $0.89_{0.12}$ & $0.48_{0.33}$ & $\mathbf{0.10_{0.02}}$ & $\mathbf{0.10_{0.02}}$\\                           &                      & 500 
& $0.51_{0.08}$ & $0.55_{0.08}$ & $0.39_{0.07}$ & $0.67_{0.11}$ & $0.60_{0.11}$ & $0.27_{0.26}$ & $\mathbf{0.07_{0.01}}$ & $\mathbf{0.07_{0.01}}$ \\
\multirow{2}{*}{3} & \multirow{3}{*}{50}  & 50
& / & / & $1.12_{0.10}$ & / & $1.32_{0.05}$ & $1.21_{0.17}$ & $\mathbf{0.66_{0.34}}$ & $\mathbf{0.70_{0.33}}$ \\
\multirow{3}{*}{($d_0 = 2$)}
 &                      & 200 
& $1.27_{0.06}$ & $1.29_{0.06}$ & $1.07_{0.08}$ & $1.23_{0.07}$ & $1.03_{0.08}$ & $0.54_{0.27}$ & $\mathbf{0.10_{0.03}}$ & $\mathbf{0.10_{0.04}}$ \\                          &                      & 500  
& $0.87_{0.10}$ & $0.93_{0.09}$ & $0.80_{0.10}$ & $0.98_{0.09}$ & $0.79_{0.07}$ & $0.32_{0.25}$ & $\mathbf{0.06_{0.01}}$ & $\mathbf{0.06_{0.01}}$ \\
                   & \multirow{3}{*}{100} & 50
& / & / & $1.16_{0.11}$ & / & $1.37_{0.02}$ & $1.18_{0.19}$ & $\mathbf{0.71_{0.35}}$ & $0.80_{0.31}$ \\                   
 &                      & 200 
& $1.38_{0.05}$ & $1.38_{0.02}$ & $1.09_{0.08}$ & $1.29_{0.04}$ & $1.22_{0.04}$ & $0.63_{0.25}$ & $\mathbf{0.10_{0.05}}$ & $\mathbf{0.11_{0.05}}$ \\  
                   &                      & 500 
& $1.22_{0.05}$ & $1.27_{0.06}$ & $1.06_{0.06}$ & $1.16_{0.06}$ & $0.94_{0.05}$ & $0.36_{0.27}$ & $\mathbf{0.06_{0.01}}$ & $0.07_{0.01}$ \\  
\hline

                   & \multirow{3}{*}{20}  & 50
& $0.90_{0.10}$ & $0.93_{0.08}$ & $\mathbf{0.51_{0.12}}$ & $0.85_{0.13}$ & $0.65_{0.13}$ & $\mathbf{0.50_{0.21}}$ & $0.53_{0.12}$ & $\mathbf{0.48_{0.19}}$ \\
                   &                      & 200 
& $0.30_{0.06}$ & $0.61_{0.11}$ & $0.29_{0.05}$ & $0.37_{0.07}$ & $0.38_{0.08}$ & $0.28_{0.22}$ & $0.34_{0.14}$ & $\mathbf{0.16_{0.04}}$ \\ 
                   &                      & 500 
& $0.17_{0.03}$ & $0.38_{0.08}$ & $0.18_{0.03}$ & $0.23_{0.04}$ & $0.23_{0.04}$ & $0.14_{0.13}$ & $0.29_{0.12}$ & $\mathbf{0.11_{0.03}}$ \\
\multirow{2}{*}{4} & \multirow{3}{*}{50}  & 50
& / & / & $\mathbf{0.55_{0.15}}$ & / & $0.86_{0.09}$ & $0.82_{0.20}$ & $0.66_{0.12}$ & $\mathbf{0.55_{0.19}}$ \\
\multirow{3}{*}{($d_0 = 1$)}
                   &                      & 200 
& $0.67_{0.12}$ & $0.90_{0.08}$ & $0.43_{0.06}$ & $0.60_{0.07}$ & $0.49_{0.06}$ & $0.32_{0.10}$ & $0.61_{0.15}$ & $\mathbf{0.14_{0.04}}$ \\                                                     &                      & 500  
& $0.30_{0.03}$ & $0.66_{0.09}$ & $0.30_{0.03}$ & $0.37_{0.04}$ & $0.34_{0.04}$ & $0.18_{0.09}$ & $0.53_{0.15}$ & $\mathbf{0.09_{0.03}}$ \\
                   & \multirow{3}{*}{100} & 50
& / & / & $\mathbf{0.61_{0.16}}$ & / & $0.95_{0.04}$ & $0.82_{0.20}$ & $0.74_{0.10}$ & $\mathbf{0.62_{0.18}}$ \\
                   &                      & 200 
& $0.99_{0.01}$ & $0.98_{0.02}$ & $0.46_{0.06}$ & $0.77_{0.07}$ & $0.64_{0.06}$ & $0.41_{0.09}$ & $0.70_{0.10}$ & $\mathbf{0.14_{0.04}}$ \\  
                   &                      & 500 
& $0.52_{0.05}$ & $0.88_{0.06}$ & $0.39_{0.04}$ & $0.52_{0.04}$ & $0.41_{0.03}$ & $0.20_{0.10}$ & $0.63_{0.15}$ & $\mathbf{0.09_{0.03}}$ \\  
\hline
\end{tabular}
\caption{\small The average sin-theta distance between $\hat{A}_0$ and $A_0$ over 100 repeats of the experiment for Models~1a, 2, 3 and 4, with $p \in \{20,50,100\}$ and $n \in \{50, 200,500\}$. We compare two versions of our method with seven competing existing approaches. For each setting, we also present $10$ times the standard error for each method in subscript.  The ``/" entries denote the the method is not applicable. }
\label{table:1-4 sin theta distance}
\end{table}

\begin{table}[!ht]
\centering
\scriptsize
\begin{tabular}{ c c c  | c c c c c c | c c }
 \multicolumn{3}{c|}{Setting} & \multicolumn{6}{c}{Competing Methods} & \multicolumn{2}{|c}{Our Methods}\\
\hline
Model& $p$ & $n$ & SIR & pHd & MAVE & DR & gKDR & drMARS & RPE & RPE2 \\
\hline
                   & \multirow{3}{*}{20}  & 50
& $1.28_{0.07}$ & $1.33_{0.06}$ & $1.15_{0.09}$ & $1.25_{0.09}$ & $1.15_{0.08}$ & $1.06_{0.19}$ & $1.06_{0.09}$ & $\mathbf{1.03_{0.13}}$ \\
                   &                      & 200 
& $1.00_{0.06}$ & $1.07_{0.08}$ & $0.95_{0.12}$ & $1.01_{0.08}$ & $0.95_{0.09}$ & $\mathbf{0.73_{0.31}}$ & $1.00_{0.00}$ & $0.93_{0.22}$\\  
                   &                      & 500 
& $0.85_{0.12}$ & $0.99_{0.05}$ & $0.64_{0.14}$ & $0.85_{0.12}$ & $0.74_{0.13}$ & $\mathbf{0.44_{0.33}}$ & $1.00_{0.00}$ & $0.66_{0.39}$ \\  
\multirow{2}{*}{5} & \multirow{3}{*}{50}  & 50
& / & / & $1.20_{0.10}$ & / & $1.30_{0.05}$ & $1.29_{0.12}$ & $\mathbf{1.12_{0.11}}$ & $\mathbf{1.13_{0.10}}$ \\
\multirow{3}{*}{($d_0 = 2$)}
                   &                      & 200 
& $1.24_{0.05}$ & $1.38_{0.03}$ & $1.15_{0.05}$ & $1.18_{0.05}$ & $1.09_{0.05}$ & $\mathbf{0.84_{0.29}}$ & $1.00_{0.00}$ & $0.97_{0.14}$\\                                    &                      & 500 
& $1.04_{0.04}$ & $1.11_{0.07}$ & $1.06_{0.05}$ & $1.05_{0.04}$ & $0.95_{0.05}$ & $\mathbf{0.50_{0.30}}$ & $1.00_{0.00}$ & $0.70_{0.36}$ \\ 
                   & \multirow{3}{*}{100} & 50
& / & / & $1.25_{0.08}$ & / & $1.36_{0.03}$ & $1.30_{0.12}$ & $\mathbf{1.15_{0.13}}$ & $1.19_{0.11}$ \\
                   &                      & 200 
& $1.40_{0.02}$ & $1.40_{0.01}$ & $1.18_{0.04}$ & $1.27_{0.03}$ & $1.20_{0.04}$ & $1.02_{0.21}$ & $1.00_{0.00}$ & $\mathbf{0.96_{0.15}}$ \\ 
                   &                      & 500 
& $1.17_{0.03}$ & $1.39_{0.03}$ & $1.15_{0.03}$ & $1.15_{0.03}$ & $1.05_{0.04}$ & $\mathbf{0.53_{0.29}}$ & $1.00_{0.00}$ & $0.74_{0.34}$ \\   
\hline

                   & \multirow{2}{*}{20}  & 50
& $1.59_{0.06}$ & $1.50_{0.08}$ & $1.38_{0.12}$ & $1.59_{0.07}$ & $1.52_{0.08}$ & $1.16_{0.46}$ & $\mathbf{1.14_{0.27}}$ & $1.20_{0.26}$ \\
                   &                      & 200 
& $1.58_{0.07}$ & $1.39_{0.09}$ & $1.26_{0.15}$ & $1.38_{0.10}$ & $1.16_{0.13}$ & $\mathbf{0.66_{0.56}}$ & $\mathbf{0.73_{0.28}}$ & $\mathbf{0.73_{0.29}}$ \\  
                   &                      & 500 
& $1.59_{0.06}$ & $1.36_{0.09}$ & $0.41_{0.05}$ & $0.98_{0.13}$ & $0.69_{0.10}$ & $0.49_{0.56}$ & $\mathbf{0.36_{0.12}}$ & $\mathbf{0.36_{0.13}}$ \\
\multirow{2}{*}{6} & \multirow{3}{*}{50}  & 50
& / & / & $1.47_{0.11}$ & / & $1.64_{0.03}$ & $1.53_{0.16}$ & $\mathbf{1.36_{0.23}}$ & $1.44_{0.16}$ \\
\multirow{3}{*}{($d_0 = 3$)}
                   &                      & 200 
& $1.68_{0.02}$ & $1.59_{0.04}$ & $1.55_{0.06}$ & $1.65_{0.03}$ & $1.65_{0.03}$ & $\mathbf{1.09_{0.51}}$ & $\mathbf{1.08_{0.35}}$ & $1.18_{0.36}$\\  
                   &                      & 500  
& $1.67_{0.03}$ & $1.54_{0.05}$ & $1.53_{0.06}$ & $1.49_{0.07}$ & $1.47_{0.06}$ & $\mathbf{0.66_{0.65}}$ & $\mathbf{0.67_{0.31}}$ & $\mathbf{0.66_{0.32}}$ \\
                   & \multirow{3}{*}{100} & 50
& / & / & $1.54_{0.11}$ & / & $1.68_{0.02}$ & $\mathbf{1.50_{0.21}}$ & $\mathbf{1.51_{0.17}}$ & $1.56_{0.14}$ \\
                   &                      & 200 
& $1.71_{0.03}$ & $1.68_{0.02}$ & $1.56_{0.05}$ & $1.70_{0.01}$ & $1.69_{0.02}$ & $\mathbf{1.29_{0.46}}$ & $\mathbf{1.31_{0.23}}$ & $1.41_{0.21}$ \\
                   &                      & 500 
& $1.70_{0.01}$ & $1.64_{0.02}$ & $1.63_{0.04}$ & $1.67_{0.03}$ & $1.69_{0.02}$ & $\mathbf{0.88_{0.65}}$ & $1.02_{0.34}$ & $1.12_{0.38}$ \\  
\hline

                   & \multirow{3}{*}{20}  & 50
& $1.34_{0.05}$ & $1.29_{0.06}$ & $1.23_{0.11}$ & $1.33_{0.06}$ & $1.30_{0.08}$ & $\mathbf{1.09_{0.24}}$ & $1.13_{0.19}$ & $1.15_{0.18}$ \\                   
                   &                      & 200 
& $1.34_{0.05}$ & $1.13_{0.11}$ & $1.05_{0.15}$ & $1.07_{0.08}$ & $1.02_{0.07}$ & $\mathbf{0.60_{0.23}}$ & $0.65_{0.19}$ & $\mathbf{0.61_{0.22}}$ \\   
                   &                      & 500 
& $1.33_{0.06}$ & $0.99_{0.13}$ & $0.54_{0.22}$ & $0.95_{0.09}$ & $0.78_{0.16}$ & $0.55_{0.16}$ & $\mathbf{0.55_{0.02}}$ & $\mathbf{0.54_{0.17}}$ \\
\multirow{2}{*}{7} & \multirow{3}{*}{50}  & 50
& / & / & $1.30_{0.09}$ & / & $1.36_{0.03}$ & $1.28_{0.16}$ & $\mathbf{1.22_{0.18}}$ & $1.26_{0.17}$\\
\multirow{3}{*}{($d_0 = 2$)}
                   &                      & 200 
& $1.39_{0.02}$ & $1.31_{0.05}$ & $1.29_{0.07}$ & $1.33_{0.06}$ & $1.37_{0.04}$ & $\mathbf{0.67_{0.29}}$ & $0.79_{0.22}$ & $0.78_{0.27}$\\  
                   &                      & 500  
& $1.39_{0.02}$ & $1.21_{0.09}$ & $1.20_{0.09}$ & $1.10_{0.05}$ & $1.20_{0.09}$ & $0.56_{0.20}$ & $0.57_{0.01}$ & $\mathbf{0.53_{0.11}}$ \\ 
                   & \multirow{3}{*}{100} & 50
& / & / & $1.35_{0.07}$ & / & $1.39_{0.02}$ & $\mathbf{1.25_{0.17}}$ & $1.31_{0.13}$ & $1.33_{0.11}$ \\
                   &                      & 200 
& $1.40_{0.01}$ & $1.39_{0.02}$ & $1.32_{0.06}$ & $1.40_{0.01}$ & $1.39_{0.02}$ & $\mathbf{0.87_{0.31}}$ & $0.95_{0.26}$ & $1.01_{0.26}$ \\  
                   &                      & 500 
& $1.40_{0.01}$ & $1.33_{0.04}$ & $1.33_{0.06}$ & $1.30_{0.06}$ & $1.39_{0.02}$ & $\mathbf{0.53_{0.17}}$ & $0.58_{0.01}$ & $0.56_{0.08}$ \\
\hline

                   & \multirow{3}{*}{20}  & 50
& $1.30_{0.08}$ & $1.21_{0.07}$ & $1.02_{0.11}$ & $1.29_{0.07}$ & $1.16_{0.13}$ & $0.38_{0.30}$ & $\mathbf{0.21_{0.15}}$ & $\mathbf{0.20_{0.12}}$ \\
                   &                      & 200 
& $1.03_{0.15}$ & $1.03_{0.07}$ & $0.71_{0.12}$ & $0.92_{0.13}$ & $0.83_{0.13}$ & $0.20_{0.13}$ & $\mathbf{0.08_{0.02}}$ & $\mathbf{0.08_{0.02}}$ \\  
                   &                      & 500 
& $0.68_{0.13}$ & $0.88_{0.19}$ & $0.38_{0.07}$ & $0.68_{0.14}$ & $0.58_{0.12}$ & $0.07_{0.05}$ & $\mathbf{0.07_{0.01}}$ & $\mathbf{0.07_{0.01}}$ \\
\multirow{2}{*}{8} & \multirow{3}{*}{50}  & 50
& / & / & $1.11_{0.11}$ & / & $1.32_{0.05}$ & $0.73_{0.30}$ & $\mathbf{0.24_{0.24}}$ & $0.27_{0.24}$ \\
\multirow{3}{*}{($d_0 = 2$)}
                   &                      & 200 
& $1.34_{0.05}$ & $1.20_{0.04}$ & $1.02_{0.06}$ & $1.16_{0.06}$ & $1.13_{0.09}$ & $0.26_{0.16}$ & $\mathbf{0.08_{0.01}}$ & $\mathbf{0.08_{0.01}}$\\  
                   &                      & 500  
& $1.06_{0.11}$ & $1.07_{0.02}$ & $0.82_{0.08}$ & $0.94_{0.10}$ & $0.83_{0.08}$ & $0.14_{0.10}$ & $\mathbf{0.07_{0.01}}$ & $\mathbf{0.07_{0.01}}$ \\ 
                   & \multirow{3}{*}{100} & 50
& / & / & $1.15_{0.10}$ & / & $1.37_{0.02}$ & $0.81_{0.27}$ & $\mathbf{0.33_{0.31}}$ & $\mathbf{0.35_{0.29}}$\\
                   &                      & 200 
& $1.40_{0.01}$ & $1.36_{0.02}$ & $1.06_{0.06}$ & $1.32_{0.04}$ & $1.28_{0.05}$ & $0.36_{0.14}$ & $\mathbf{0.07_{0.01}}$ & $0.08_{0.02}$\\
                   &                      & 500 
& $1.34_{0.04}$ & $1.19_{0.04}$ & $0.99_{0.05}$ & $1.10_{0.05}$ & $1.06_{0.07}$ & $0.17_{0.12}$ & $\mathbf{0.07_{0.01}}$ & $\mathbf{0.07_{0.01}}$ \\ 
\hline

                   & \multirow{3}{*}{20}  & 50
& $1.30_{0.07}$ & $1.31_{0.06}$ & $1.21_{0.09}$ & $1.32_{0.07}$ & $1.22_{0.07}$ & $\mathbf{0.95_{0.32}}$ & $1.00_{0.15}$ & $1.02_{0.16}$ \\
                   &                      & 200 
& $1.06_{0.05}$ & $1.17_{0.08}$ & $1.12_{0.06}$ & $1.18_{0.09}$ & $1.17_{0.07}$ & $0.31_{0.43}$ & $\mathbf{0.18_{0.21}}$ & $\mathbf{0.18_{0.22}}$ \\ 
                   &                      & 500 
& $1.00_{0.04}$ & $1.08_{0.05}$ & $1.03_{0.04}$ & $1.06_{0.04}$ & $1.08_{0.05}$ & $0.14_{0.31}$ & $\mathbf{0.07_{0.02}}$ & $\mathbf{0.07_{0.02}}$ \\
\multirow{2}{*}{9} & \multirow{3}{*}{50}  & 50
& / & / & $1.25_{0.08}$ & / & $1.33_{0.04}$ & $1.24_{0.15}$ & $\mathbf{1.07_{0.16}}$ & $1.12_{0.17}$ \\
\multirow{3}{*}{($d_0 = 2$)}
                   &                      & 200 
& $1.30_{0.06}$ & $1.33_{0.04}$ & $1.22_{0.04}$ & $1.30_{0.05}$ & $1.18_{0.04}$ & $0.69_{0.43}$ & $\mathbf{0.27_{0.30}}$ & $\mathbf{0.28_{0.30}}$ \\  
                   &                      & 500  
& $1.09_{0.03}$ & $1.21_{0.05}$ & $1.15_{0.03}$ & $1.18_{0.06}$ & $1.15_{0.04}$ & $0.47_{0.48}$ & $\mathbf{0.09_{0.06}}$ & $\mathbf{0.09_{0.06}}$ \\
                   & \multirow{3}{*}{100} & 50
& / & / & $1.28_{0.07}$ & / & $1.37_{0.02}$ & $1.24_{0.16}$ & $\mathbf{1.07_{0.18}}$ & $1.13_{0.15}$ \\
                   &                      & 200 
& $1.40_{0.03}$ & $1.39_{0.02}$ & $1.23_{0.05}$ & $1.34_{0.03}$ & $1.27_{0.04}$ & $0.89_{0.36}$ & $\mathbf{0.35_{0.35}}$ & $\mathbf{0.37_{0.36}}$ \\
                   &                      & 500 
& $1.26_{0.05}$ & $1.33_{0.04}$ & $1.22_{0.03}$ & $1.26_{0.04}$ & $1.15_{0.02}$ & $0.53_{0.48}$ & $\mathbf{0.11_{0.18}}$ & $\mathbf{0.11_{0.18}}$ \\ 
\hline 
\end{tabular}
\caption{\small The average sin-theta distance between $\hat{A}_0$ and $A_0$ over 100 repeats of the experiment for Models~5-9, with $p \in \{20,50,100\}$ and $n \in \{50, 200,500\}$. We compare two versions of our method with seven competing existing approaches. For each setting, we also present $10$ times the standard error for each method in subscript.  The ``/" entries denote the the method is not applicable. }
\label{table:5-9 sin theta distance}
\end{table}

The broad message from the results in Tables~\ref{table:1-4 sin theta distance} and~\ref{table:5-9 sin theta distance} is that our random projection based algorithms are competitive across the 9 models, for different values of $n$ and $p$.  Indeed, the relative performance is somewhat similar to the $p=50$ and $n = 200$ case presented in the main text.  There are some additional points worth noting here.  First,  when $n = 50$ and $p \in \{50, 100\}$ so that $p \geq n$, the SIR, pHd and DR methods are not applicable.  Overall, our approach enjoys the best performance for Models 2, 3, 4, 8 and 9.  For the remaining models, the performance of our random projection based method and drMARS is close, and these outperform the other competitors.   

\subsection{Computational cost}
\label{sec:runtime}
In this section we investigate the actual runtime in seconds of the different methods we consider.  In Table~\ref{table:model runtime}, we present the run time for one run of the experiment for Model 1a, for different values of $n$ and $p$. The second right column, with label `RPE', shows the running time of the experiment on a single core of a 3.20 GHz Intel i9-14900KF computer.  We also demonstrate the potential speed up that may be gained using parallel computing -- the RPE (10 cores) column presents the runtime of algorithm using ten cores of a 3.20 GHz Intel i9-14900KF computer.  Using parallel computing in this example results in approximately seven times faster compute time.

\begin{table}[!ht]
\centering
\small
\begin{tabular}{ c c  | c c c c c c | c c c }
 \multicolumn{2}{c|}{Setting} & \multicolumn{6}{c}{Competing Methods} & \multicolumn{2}{|c}{Our Methods}\\
\hline
$p$ & $n$ & SIR & pHd & MAVE & DR & gKDR & drMARS & RPE & RPE(10 cores) \\
\hline 
\multirow{3}{*}{20} & 50
& $0.007$ & $0.007$ & $0.051$ & $0.009$ & $0.015$ & $0.138$ & $120.151$ & $18.043$  \\
                                         & 200 
& $0.007$ & $0.005$ & $0.408$ & $0.007$ & $0.078$ & $0.211$ & $146.465$ & $22.750$  \\ 
                                      & 500 
& $0.008$ & $0.006$ & $0.961$ & $0.008$ & $0.381$ & $0.276$ & $194.996$ & $28.899$ \\ 
\hline
\multirow{3}{*}{50}  & 50
& / & / & $0.047$ & / & $0.026$ & $0.185$ & $367.642$ & $56.797$ \\
                      & 200 
& $0.011$ & $0.007$ & $0.827$ & $0.013$ & $0.224$ & $0.478$ & $508.912$ & $73.372$ \\ 
                                   & 500 
& $0.014$ & $0.009$ & $2.988$ & $0.020$ & $1.006$ & $0.635$ & $697.639$ & $98.825$  \\ 
\hline
\multirow{3}{*}{100} & 50
& / & / & $0.047$ & / & $0.109$ & $0.188$ & $1302.905$ & $276.516$ \\                   
& 200 
& $0.032$ & $0.016$ & $0.820$ & $0.018$ & $0.778$ & $0.841$ & $1724.973$ & $300.544$  \\  
                                        & 500 
& $0.033$ & $0.020$ & $7.945$ & $0.031$ & $3.327$ & $1.162$ & $2040.819$ & $356.924$  \\
\hline 
\end{tabular}
\caption{\small The runtime in seconds of the methods considered for the simulations for Model~1a, with $p \in \{20,50,100\}$ and $n \in \{50, 200,500\}$. The "/" entries denote the the method is not applicable. }
\label{table:model runtime}
\end{table}

\subsection{Dimension \texorpdfstring{$d_0$}{d0} unknown}

Tables~\ref{table:FP} and~\ref{table:FN} show the full results for the false positive and false negative metrics of additional experiments corresponding to those in Section~\ref{sec:numericald0unknown}. We present the results for $p \in \{20, 50, 100\}$ and $n = 200$. Results for $n = 50$ and $n = 500$, which are broadly similar, are omitted for brevity.

\begin{table}[!ht]
\centering
\small
\begin{tabular}{ c c | c c c c c | c  }
 \multicolumn{2}{c|}{Setting} & \multicolumn{5}{c}{Competing Methods} & \multicolumn{1}{|c}{Our Methods}\\
\hline
Model& $p$ & SIR & pHd & MAVE & gKDR & drMARS & RPE \\
\hline 
\multirow{3}{*}{1a ($d_0 = 1)$}& 20
& $0.06_{0.23}$ & $0.07_{0.17}$ & $0.24_{0.04}$ & $1.48_{0.60}$ & $1.15_{0.76}$ & $1.00_{0.00}$ \\ 

                   & 50 
& $0.04_{0.20}$ & $0.00_{0.00}$ & $0.93_{0.08}$ & $2.22_{0.33}$ & $1.54_{0.80}$ & $1.00_{0.04}$ \\  

                   & 100 
& $0.00_{0.00}$ & $0.00_{0.00}$ & $0.97_{0.04}$ & $2.87_{0.26}$ & $1.82_{0.84}$ & $1.07_{0.21}$ \\

\hline 
\multirow{3}{*}{2 ($d_0=1$)} & 20 
& $0.04_{0.22}$ & $0.00_{0.00}$ & $1.68_{0.20}$ & $1.51_{0.42}$ & $1.33_{0.58}$ & $0.59_{0.59}$ \\ 

                   & 50  
& $0.02_{0.14}$ & $0.00_{0.00}$ & $1.70_{0.12}$ & $2.13_{0.47}$ & $1.81_{0.62}$ & $1.02_{0.83}$ \\  

                   & 100  
& $0.00_{0.00}$ & $0.00_{0.00}$ & $1.68_{0.13}$ & $2.63_{0.38}$ & $1.90_{0.70}$ & $1.65_{1.07}$ \\  

\hline 
\multirow{3}{*}{3 ($d_0=2$)} & 20  
& $0.46_{0.25}$ & $0.04_{0.15}$ & $0.92_{0.30}$ & $1.23_{0.52}$ & $1.06_{0.63}$ & $0.10_{0.02}$ \\     

                   & 50  
& $0.19_{0.40}$ & $0.00_{0.00}$ & $1.34_{0.21}$ & $1.94_{0.46}$ & $1.50_{0.73}$ & $0.19_{0.29}$ \\     

                   & 100  
& $0.00_{0.00}$ & $0.00_{0.00}$ & $1.35_{0.20}$ & $2.55_{0.41}$ & $1.54_{0.84}$ & $0.35_{0.47}$ \\ 

\hline 

\multirow{3}{*}{4 ($d_0=1$)} & 20 
& $0.32_{0.19}$ & $0.00_{0.00}$ & $1.06_{0.35}$ & $1.32_{0.61}$ & $1.09_{0.67}$ & $1.42_{0.00}$ \\

                   & 50  
& $0.09_{0.25}$ & $0.00_{0.00}$ & $1.27_{0.21}$ & $2.10_{0.49}$ & $1.50_{0.81}$ & $1.44_{0.09}$ \\  

                   & 100 
& $0.00_{0.00}$ & $0.00_{0.00}$ & $1.28_{0.20}$ & $2.49_{0.45}$ & $1.72_{0.83}$ & $1.54_{0.18}$ \\

\hline
\multirow{3}{*}{5 ($d_0=2$)} & 20 
& $0.35_{0.21}$ & $0.01_{0.10}$ & $0.98_{0.34}$ & $1.45_{0.47}$ & $0.97_{0.55}$ & $1.01_{0.01}$ \\ 

                   & 50   
& $0.09_{0.28}$ & $0.00_{0.00}$ & $1.40_{0.19}$ & $2.06_{0.39}$ & $1.47_{0.67}$ & $1.05_{0.13}$ \\  

                   & 100 
& $0.00_{0.00}$ & $0.00_{0.00}$ & $1.41_{0.18}$ & $2.52_{0.40}$ & $1.55_{0.72}$ & $1.19_{0.25}$ \\

\hline 

\multirow{3}{*}{6 ($d_0=3$)} & 20   
& $0.03_{0.15}$ & $0.00_{0.00}$ & $1.26_{0.29}$ & $0.85_{0.37}$ & $0.59_{0.57}$ & $1.01_{0.46}$ \\  

                   & 50  
& $0.00_{0.00}$ & $0.00_{0.00}$ & $1.55_{0.23}$ & $1.47_{0.56}$ & $1.10_{0.79}$ & $1.97_{0.56}$ \\  

                   & 100  
& $0.00_{0.00}$ & $0.00_{0.00}$ & $1.56_{0.22}$ & $2.55_{0.50}$ & $1.36_{0.73}$ & $2.47_{0.62}$ \\
 
\hline 
\multirow{3}{*}{7 ($d_0=2$)} & 20  
& $0.02_{0.13}$ & $0.00_{0.00}$ & $1.31_{0.32}$ & $1.00_{0.53}$ & $0.87_{0.57}$ & $0.99_{0.44}$ \\

                   & 50  
& $0.02_{0.13}$ & $0.00_{0.00}$ & $1.61_{0.17}$ & $1.91_{0.55}$ & $1.23_{0.77}$ & $1.63_{0.65}$ \\  

                   & 100  
& $0.00_{0.00}$ & $0.00_{0.00}$ & $1.63_{0.16}$ & $2.56_{0.41}$ & $1.52_{0.82}$ & $2.26_{0.79}$ \\ 

\hline 
\multirow{3}{*}{8 ($d_0=2$)} & 20   
& $0.37_{0.34}$ & $0.11_{0.17}$ & $0.48_{0.24}$ & $1.56_{0.40}$ & $0.99_{0.63}$ & $0.08_{0.02}$\\

                   & 50   
& $0.12_{0.35}$ & $0.00_{0.00}$ & $0.82_{0.25}$ & $2.23_{0.35}$ & $1.42_{0.71}$ & $0.09_{0.09}$ \\  

                   & 100 
& $0.00_{0.00}$ & $0.00_{0.00}$ & $0.90_{0.22}$ & $2.65_{0.35}$ & $1.62_{0.80}$ & $0.18_{0.31}$ \\

\hline 
\multirow{3}{*}{9 ($d_0=2$)} & 20  
& $0.29_{0.28}$ & $0.00_{0.00}$ & $1.45_{0.24}$ & $1.63_{0.40}$ & $0.86_{0.64}$ & $0.11_{0.14}$ \\ 

                   & 50   
& $0.10_{0.29}$ & $0.00_{0.00}$ & $1.53_{0.16}$ & $2.23_{0.38}$ & $1.26_{0.80}$ & $0.28_{0.36}$ \\ 

                   & 100  
& $0.00_{0.00}$ & $0.00_{0.00}$ & $1.52_{0.16}$ & $2.67_{0.39}$ & $1.74_{0.81}$ & $0.41_{0.51}$ \\

\hline 
\end{tabular}

\caption{\small The average of $d_{\mathrm{FP}}(\mathcal{S}(\hat{A}_0), \mathcal{S}(A_0))$ over 100 repeats of the experiment for Models~1a-9, with $p \in \{20,50,100\}$ and $n =200$. For each setting, we also present $10$ times the standard error for each method in subscript. }
\label{table:FP}
\end{table}

\begin{table}[!ht]
\centering
\small
\begin{tabular}{ c c | c c c c c | c  }
 \multicolumn{2}{c|}{Setting} & \multicolumn{5}{c}{Competing Methods} & \multicolumn{1}{|c}{Our Methods}\\
\hline
Model& $p$ & SIR & pHd & MAVE & gKDR & drMARS & RPE \\
\hline 
\multirow{3}{*}{1a ($d_0 = 1)$}& \multirow{1}{*}{20}  
& $0.99_{0.04}$ & $0.88_{0.25}$ & $0.24_{0.04}$ & $0.25_{0.06}$ & $0.02_{0.02}$ & $0.05_{0.01}$ \\
                      & \multirow{1}{*}{50} 
& $1.00_{0.00}$ & $1.00_{0.00}$ & $0.93_{0.08}$ & $0.69_{0.15}$ & $0.02_{0.01}$ & $0.04_{0.01}$ \\
                           & \multirow{1}{*}{100} 
& $1.00_{0.00}$ & $1.00_{0.00}$ & $0.97_{0.04}$ & $0.83_{0.07}$ & $0.03_{0.02}$ & $0.04_{0.00}$ \\ 
                \hline 
\multirow{3}{*}{2 ($d_0=1$)} & \multirow{1}{*}{20} 
&  $1.00_{0.00}$ & $1.00_{0.00}$ & $0.92_{0.07}$ & $0.94_{0.05}$ & $0.48_{0.46}$ & ${0.05_{0.02}}$ \\ 
                             & \multirow{1}{*}{50} 
& $1.00_{0.00}$ & $1.00_{0.00}$ & $0.96_{0.03}$ & $0.95_{0.04}$ & $0.87_{0.29}$ & ${0.07_{0.04}}$ \\ 
                    & \multirow{1}{*}{100} 
& $1.00_{0.00}$ & $1.00_{0.00}$ & $0.97_{0.03}$ & $0.96_{0.02}$ & $0.90_{0.25}$ & ${0.09_{0.08}}$ \\ 
                  \hline 
\multirow{3}{*}{3 ($d_0=2$)} & \multirow{1}{*}{20} 
& $1.07_{0.16}$ & $1.40_{0.07}$ & $0.69_{0.12}$ & $0.88_{0.17}$ & $0.43_{0.34}$ & ${0.10_{0.02}}$ \\ 
               & \multirow{1}{*}{50} 
& $1.40_{0.05}$ & $1.41_{0.00}$ & $1.02_{0.09}$ & $0.96_{0.08}$ & $0.48_{0.28}$ & ${0.10_{0.03}}$ \\  
                  & \multirow{1}{*}{100}
& $1.41_{0.00}$ & $1.41_{0.00}$ & $1.05_{0.09}$ & $1.12_{0.06}$ & $0.56_{0.30}$ & ${0.10_{0.09}}$ \\ 
               \hline 

\multirow{3}{*}{4 ($d_0=1$)} & \multirow{1}{*}{20} 
& $0.35_{0.19}$ & $1.00_{0.00}$ & $0.27_{0.05}$ & $0.35_{0.07}$ & $0.18_{0.08}$ & ${0.11_{0.03}}$ \\
                  & \multirow{1}{*}{50} 
& $0.96_{0.11}$ & $1.00_{0.00}$ & $0.40_{0.05}$ & $0.44_{0.05}$ & $0.26_{0.10}$ & ${0.09_{0.02}}$ \\
                 & \multirow{1}{*}{100}
& $1.00_{0.00}$ & $1.00_{0.00}$ & $0.43_{0.06}$ & $0.58_{0.05}$ & $0.32_{0.12}$ & ${0.09_{0.02}}$ \\
            \hline
\multirow{3}{*}{5 ($d_0=2$)} & \multirow{1}{*}{20} 
& $1.11_{0.13}$ & $1.41_{0.00}$ & $0.94_{0.14}$ & $0.88_{0.14}$ & ${0.55_{0.40}}$ & $0.93_{0.23}$ \\
                  & \multirow{1}{*}{50} 
& $1.41_{0.04}$ & $1.41_{0.00}$ & $1.13_{0.06}$ & $1.02_{0.06}$ & ${0.61_{0.34}}$ & $0.97_{0.16}$ \\  
                 & \multirow{1}{*}{100}
& $1.41_{0.00}$ & $1.41_{0.00}$ & $1.16_{0.05}$ & $1.15_{0.04}$ & ${0.89_{0.31}}$ & $0.95_{0.17}$ \\
            \hline 
\multirow{3}{*}{6 ($d_0=3$)} & \multirow{1}{*}{20}
& $1.73_{0.02}$ & $1.73_{0.00}$ & $1.26_{0.17}$ & $1.37_{0.22}$ & $1.04_{0.56}$ & ${0.63_{0.25}}$ \\ 
                              & \multirow{1}{*}{50}
& $1.71_{0.00}$ & $1.73_{0.00}$ & $1.55_{0.07}$ & $1.63_{0.11}$ & $1.22_{0.48}$ & ${0.93_{0.35}}$ \\  
                                & \multirow{1}{*}{100}
& $1.73_{0.00}$ & $1.73_{0.00}$ & $1.57_{0.07}$ & $1.63_{0.05}$ & $1.28_{0.50}$ & ${1.15_{0.28}}$ \\
               \hline 
\multirow{3}{*}{7 ($d_0=2$)} & \multirow{1}{*}{20} 
& $1.41_{0.00}$ & $1.41_{0.00}$ & $0.98_{0.15}$ & $1.01_{0.11}$ & $0.61_{0.39}$ & ${0.56_{0.20}}$ \\
                                      & \multirow{1}{*}{50}  
& $1.41_{0.01}$ & $1.41_{0.00}$ & $1.23_{0.06}$ & $1.27_{0.11}$ & ${0.61_{0.40}}$ & ${0.62_{0.23}}$ \\ 
                                    & \multirow{1}{*}{100}
& $1.41_{0.00}$ & $1.41_{0.00}$ & $1.28_{0.08}$ & $1.32_{0.05}$ & ${0.81_{0.37}}$ & ${0.80_{0.26}}$ \\ 
               \hline 
\multirow{3}{*}{8 ($d_0=2$)} &20 
& $1.25_{0.16}$ & $1.31_{0.16}$ & $0.84_{0.19}$ & $0.73_{0.14}$ & $0.18_{0.18}$ & ${0.08_{0.02}}$ \\  
                                  & 50  
& $1.41_{0.02}$ & $1.41_{0.00}$ & $1.07_{0.09}$ & $0.96_{0.10}$ & $0.19_{0.11}$ & ${0.08_{0.01}}$ \\ 
                                    &100  
& $1.41_{0.00}$ & $1.41_{0.00}$ & $1.11_{0.09}$ & $1.15_{0.06}$ & $0.25_{0.10}$ & ${0.07_{0.01}}$ \\
               \hline 
\multirow{3}{*}{9 ($d_0=2$)} & 20 
& $1.24_{0.15}$ & $1.41_{0.08}$ & $1.08_{0.07}$ & $1.11_{0.10}$ & $0.31_{0.44}$ & ${0.26_{0.34}}$ \\ 
                                    & 50 
& $1.40_{0.04}$ & $1.41_{0.00}$ & $1.20_{0.05}$ & $1.13_{0.04}$ & $0.68_{0.44}$ & ${0.29_{0.33}}$ \\  
                                 & 100 
& $1.41_{0.00}$ & $1.41_{0.00}$ & $1.21_{0.05}$ & $1.22_{0.04}$ & $0.85_{0.37}$ & ${0.34_{0.36}}$ \\
            \hline 
\end{tabular}

\caption{\small  The average of $d_{\mathrm{FN}}(\mathcal{S}(\hat{A}_0), \mathcal{S}(A_0))$ over 100 repeats of the experiment for Models~1a-9, with $p \in \{20,50,100\}$ and $n=200$. For each setting, we also present $10$ times the standard error for each method in subscript.}
\label{table:FN}
\end{table}

\end{document}